\documentclass[11pt,appendixprefix=Appendix]{article}
\usepackage[a4paper, width=6in, bottom= 1.2in, top = 1.2in]{geometry}

\usepackage{times}
\usepackage{amsmath,amssymb,amsfonts,amsthm,graphicx,subfiles,
microtype,bm,titlesec,subcaption,mathtools,soul,xcolor,enumitem,array,
physics,upgreek}
\usepackage[colorlinks,linkcolor=blue,citecolor=blue,urlcolor=blue]{hyperref}
\usepackage[authoryear, semicolon]{natbib}
\usepackage{setspace}
\usepackage{enumitem}
\usepackage{textcomp,verbatim,changebar,appendix}

\usepackage{xr}

\newcounter{algno}
\titleformat{\chapter}{\singlespacing\bf\huge\raggedright}{\thechapter.}{20pt}{\huge\bf\nonumber}
\titlespacing{\chapter}{0pt}{0pt}{40pt}

\newcommand{\ORD}{\mathcal{O}}
\newcommand{\E}{\mathbf{E}}
\newcommand{\Wt}{\widetilde{W}}
\newcommand{\It}{\tilde{I}}
\newcommand{\cov}{\mathrm{cov}}

\newcommand{\tsbs}{\textsubscript}
\newcommand{\tsps}{\textsuperscript}

\def\beq{\begin{equation}}
\def\eeq{\end{equation}}
\def\beqa{\begin{eqnarray}}
\def\eeqa{\end{eqnarray}}
\def\beqann{\begin{eqnarray*}}
\def\eeqann{\end{eqnarray*}}
\def\bcb{\begin{changebar}}
\def\ecb{\end{changebar}}

\def\Z{{\mathbb Z}}

\theoremstyle{definition}
\newtheorem{definition}{Definition}
\theoremstyle{definition}
\newtheorem{proposition}{Proposition}
\theoremstyle{definition}
\newtheorem{lemma}{Lemma}

\graphicspath{{./arXiv_plots/}}

\title{Multiscale modelling of replicated nonstationary time series}

\usepackage{authblk}
\author[1]{Jonathan Embleton\thanks{Corresponding author: je687@york.ac.uk}}
\author[1]{Marina I. Knight}
\author[2]{Hernando Ombao}
\affil[1]{Department of Mathematics, University of York, UK}
\affil[2]{King Abdullah University of Science and Technology (KAUST), Saudi Arabia}

\begin{document}
\begin{spacing}{1}
\maketitle

\begin{abstract}
Within the neurosciences, to observe variability across time in the dynamics of an underlying brain process is neither new nor unexpected.  Wavelets are essential in analyzing brain signals because, even within a single trial, brain signals exhibit nonstationary behaviour. However, neurological signals generated within an experiment may also potentially exhibit evolution across trials (replicates). As neurologists consider localised spectra of brain signals to be most informative, here we develop a novel wavelet-based tool capable to formally represent process nonstationarities across both time and replicate dimensions. Specifically, we propose the Replicate Locally Stationary Wavelet (RLSW) process, that captures the potential nonstationary behaviour within and across trials. Estimation using wavelets gives a natural desired time- and replicate-localisation of the process dynamics. We develop the associated  spectral estimation framework and establish its asymptotic properties. By means of thorough simulation studies, we demonstrate the theoretical estimator properties hold in practice. A real data investigation into the evolutionary dynamics of the hippocampus and nucleus accumbens during an associative learning experiment, demonstrate the applicability of our proposed methodology, as well as the new insights it provides.
\end{abstract}


\section[1]{Introduction}\label{sec:intro}

In an experimental setting consisting of repeated trials, inference is typically carried out on the average dynamics of the underlying process over all trials. However, a recent study on neurological signals \citep{fo:medbp} suggests that this approach is naive due to its failure to account for the possibility of a change in the process dynamics over the course of the experiment. Their data example focusses on the hippocampus (Hc) and the nucleus accumbens (NAc), both known to play important roles in cognitive processing as they are individually associated with memory recall and the processing of reward, respectively. Recordings of electrical activity (at approximately 1000Hz) using local field potentials (LFPs) were obtained from the Hc and NAc of an awake behaving macaque during an associative learning experiment. For each trial, the macaque was presented with one of four pictures and was then tasked with associating this picture with one of the four doors appearing on the screen. Upon making a correct association, the macaque was rewarded with a small quantity of juice. Plots of the LFPs obtained from trials in which a correct association was made are shown in Figures~\ref{fig:rb7_Ct_Hc} and \ref{fig:rb7_Ct_NAc}. Variability in neuronal activity within both brain regions has also been observed over the trials of a learning experiment in other recent studies,  including \cite{segcin:dynamics, Gorrombao:monkeydepend, abela:2015}. Such traits present the challenge of modelling time series that display potential nonstationary behaviour not only across time, but also across trials.

\begin{figure}[ht!]
\begin{center}
\hspace*{-2em}
\resizebox{!}{0.6\textwidth}{\includegraphics{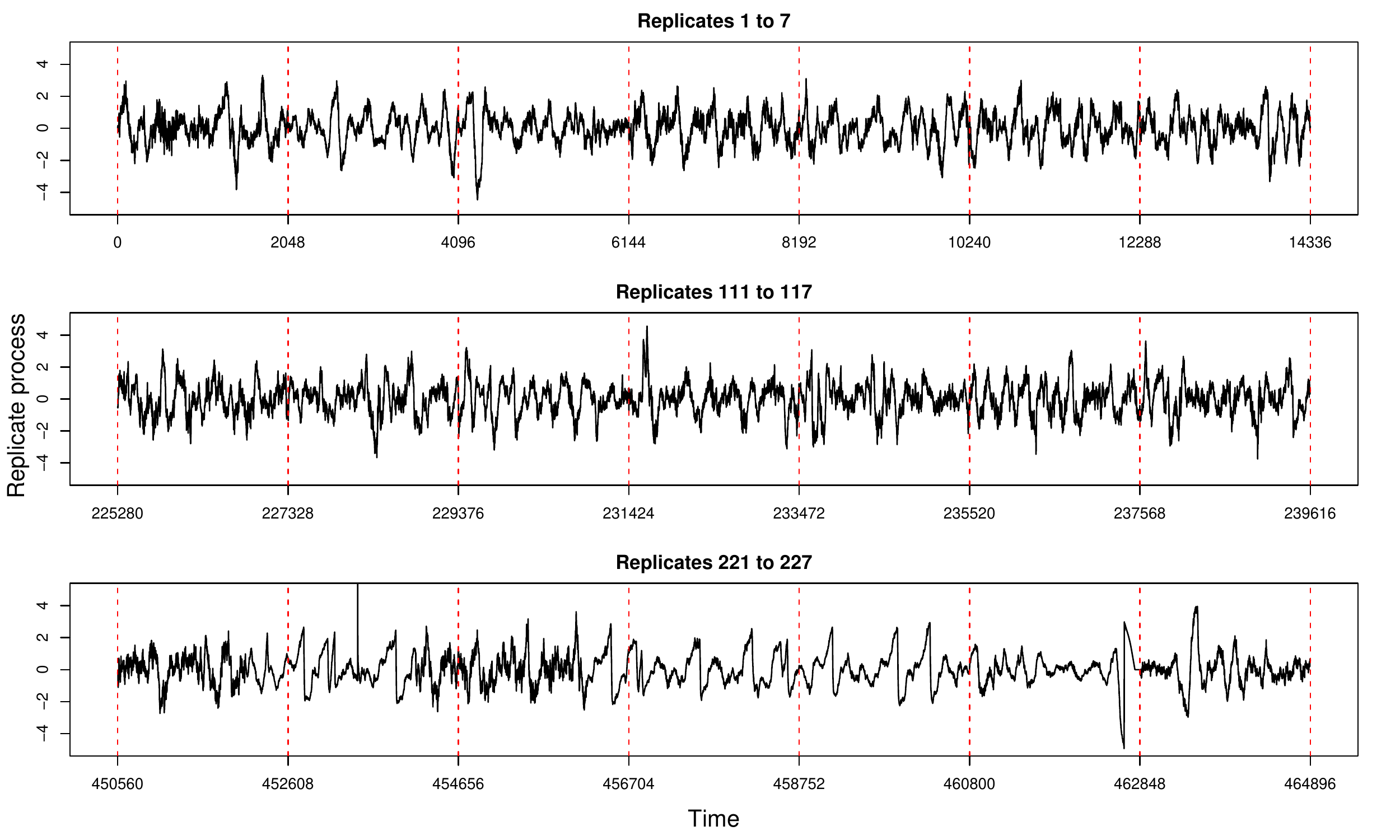}}
\caption{\small Concatenated series of the hippocampus (Hc) data in the correct response trials.}
\label{fig:rb7_Ct_Hc}
\end{center}
\end{figure}

\begin{figure}[ht!]
\begin{center}
\hspace*{-2em}
\resizebox{!}{0.6\textwidth}{\includegraphics{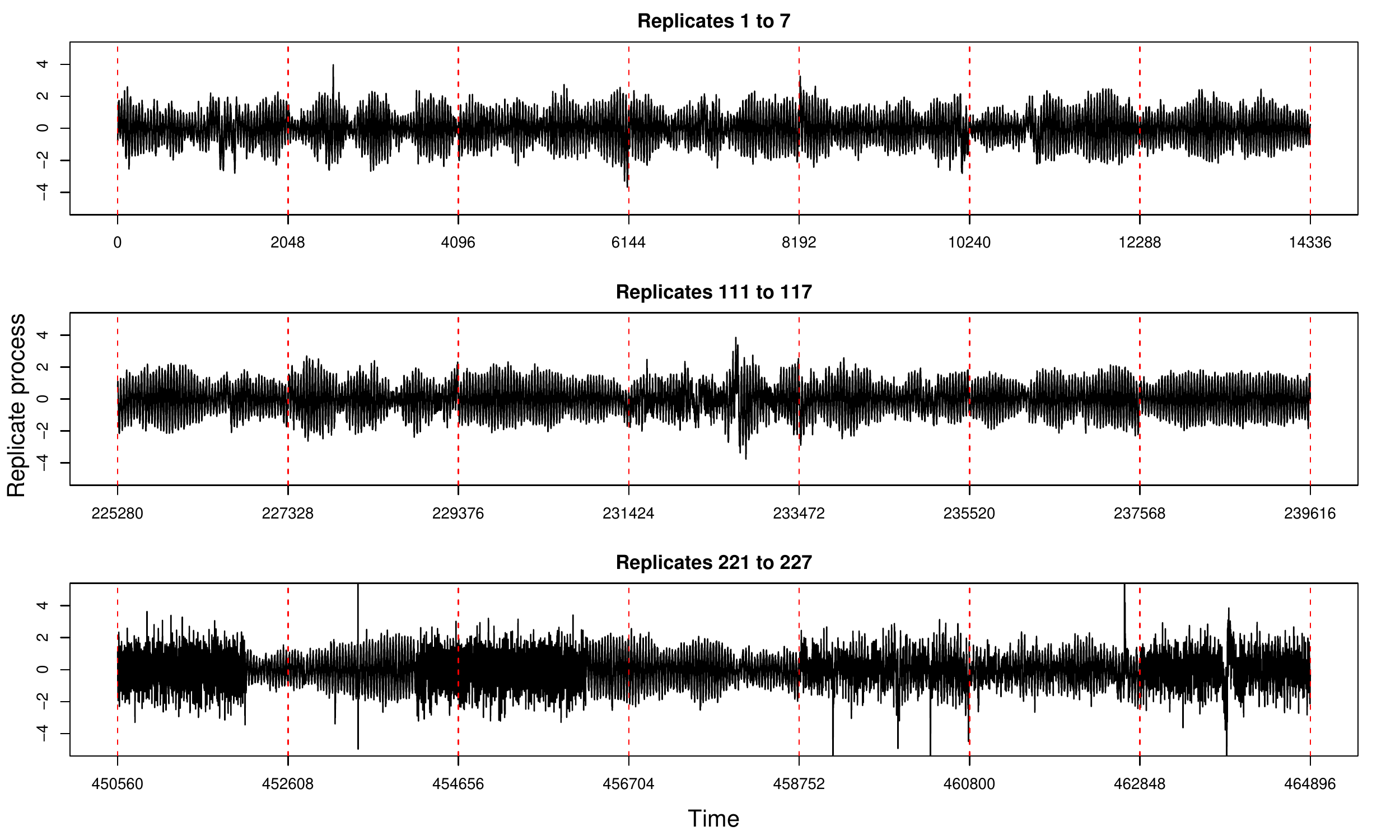}}
\caption{\small Concatenated series of the nucleus accumbens (NAc) data in the correct response trials.}
\label{fig:rb7_Ct_NAc}
\end{center}
\end{figure}

In the specific context of brain signals, the usefulness of time-scale spectral decompositions that are typical of wavelet constructions has already been established in the literature \citep{sand:2010, park:mvlsw}. Our aim is therefore to develop a {\em wavelet}-based model that, in the spectral domain, captures in a scale-dependent manner not only the evolutionary dynamics of the underlying brain process across time (within each trial) but also across the trials of the experiment.  An important assumption typically undertaken is that of independent/ uncorrelated trials. However, some studies document evidence of correlation across trials in an experiment \citep{arieli:1996, huk:2018}. Thus, one important feature of our model is that it accounts for the dependence between trials by means of a coherence quantity that acts as a measure for cross-trial dependence.

Time series data encountered in practice are often of a nonstationary nature and much recent research has been concerned with developing statistical models that capture this behaviour within the series. \cite{priest:1965} was first to study in detail nonstationary processes with time varying spectral characteristics. Further models using Fourier waveforms and adaptations of, have been developed by \cite{dahl:1997} who introduced the methodology of locally stationary processes which allowed for asymptotics to be established, and \cite{omb:2002} whom utilise the SLEX (smoothed localised complex exponentials) library of waveforms which give localisation in both time and frequency. Similar representations have been developed using wavelets. Approaches that use wavelet thresholding for smoothing the spectra of locally stationary time series have been considered by \cite{vSsch:1996} and \cite{neumvs:1997}. The locally stationary wavelet framework was developed by \cite{nvsk:2000}. Their approach used a set of discrete non-decimated wavelets to replace the Fourier exponentials used in \cite{dahl:1997}, and this in turn offered localisation in both time and scale. Extensions to the multivariate setting appear in the work of \cite{omb:2005}, \cite{sand:2010} and \cite{park:mvlsw}, where a structure for coherence is also embedded in the model.

What is not well accounted for in the literature, is the statistical modelling of second order nonstationarity across a collection of (constant mean) time series arising from the same experiment. In earlier work for stationary time series, \cite{digwas:1997} accounted for the stochastic variation that arises across replicated time series via a subject-specific realisation of a random process. They proposed a generalised linear mixed-effects model to estimate population characteristics where the subject-specific random term in the model captures the cross-subject variation. Extension to the nonstationary case has been addressed by \cite{QGL:2009} whom, through the locally stationary processes framework, developed a time-frequency functional model, where the time-varying log-spectra determine the evolution of the stochastic variation. More recently, \cite{fo:medbp} proposed a model that captures the evolution in the spectral characteristics of dynamic brain processes across a collection of trials in an associative learning experiment. Their methodology, developed for multivariate analysis, captures process evolution within a trial and across trials through a Fourier spectrum. The authors refer to trials as {\em replicates}, a term we will also borrow in our nomenclature.

Our approach acknowledges the work of \cite{fo:medbp}, who developed their time--replicate model using Fourier waveforms, but the limitation of their work is the assumption of uncorrelated replicates. To the best of our knowledge, our proposed methodology is the first to model nonstationary stochastic variation across replicate time series in the {\em wavelet domain}. Major advantages of our model are that (i) it offers the superior time-localisation typical of wavelet constructions, and (ii) it takes into account the correlation of brain signals across trials. We propose to model the replicate process within a locally stationary wavelet process paradigm that builds upon the framework introduced by \cite{nvsk:2000} for a single process (here, replicate or trial). Two new models are developed, first under the constraint of uncorrelated replicates, and then relaxing this assumption and allowing for replicate correlation. This amounts to developing novel evolutionary wavelet quantities and associated estimation theory that encompass variation {\em both} across time and replicate. To obtain well-behaved, consistent spectral estimates we propose to perform local smoothing of  the raw wavelet peridograms across replicates, as opposed to employing the smoothing over time typically undertaken in the locally stationary process context. Replicate-coherence estimation theory is also treated and is shown to provide useful information about the process replicate evolution.

The article proceeds as follows. Section~\ref{sec:model} introduces our proposed model as well as its associated estimation theory, both under the assumption of replicate uncorrelation. Section~\ref{sec:sims} details simulation studies that illustrate the behaviour of the proposed methodology. Section~\ref{sec:modelext} describes the new model that allows for correlation between replicate processes. Corresponding estimation theory is developed in Section~\ref{sec:modelext}. Section~\ref{sec:simsext} illustrates through simulation the advantage of our proposed work, both for within- and across- replicate behaviour characterisation. Section~\ref{sec:rda} details an application of the proposed methodologies to a real data study within neuroscience, and Section~\ref{sec:concl} concludes the paper.

\section{Theoretical model}\label{sec:model}

\subsection*{Brief introduction of locally stationary wavelet processes}
Before we describe the framework for the \textit{replicate locally stationary wavelet} model, we recall some of the defining features of the \textit{locally stationary wavelet} (LSW) framework of \cite{nvsk:2000}. The LSW model provides a \textit{time-scale} representation of nonstationary time series with time-varying second order structure, where the building blocks are the discrete non-decimated wavelets (see \cite{vid:smbw} or \cite{nas:wms} for an extensive introduction to wavelets).
For $T=2^{J(T)}$, a sequence of stochastic processes $\{X_{t;T}\}_{t=0,...,T-1}$ is a LSW process if it admits the representation
\begin{equation*} \label{eq:lsw}
X_{t;T} = \sum_{j=1}^{\infty}\sum_{k\in\mathbb{Z}} \omega_{j,k;T}\psi_{j,k}(t)\xi_{j,k},
\end{equation*}
where for scale $j$ and location $k$, $\omega_{j,k;T}$ is the amplitude corresponding to the discrete non-decimated wavelet $\psi_{j,k}(t)$ and $\{\xi_{j,k}\}$ are a set of orthonormal random variables. Modelling under the concept of local stationarity means that the variation of the amplitudes $\{\omega_{j,k;T}\}_k$, happens slowly over time and this is controlled by a smoothly varying continuous Lipschitz function $W_{j}(k/T)$, that can be thought of as a scale ($j$) and time ($k$) dependent transfer function \citep{frynas:HFisz}. \cite{nvsk:2000} further propose the \textit{evolutionary wavelet spectrum} (EWS) as a means to quantify the contribution to the overall process variance at a scale $j$ and rescaled time $z=k/T$ and define this as $S_{j}(z) = |W_{j}(z)|^2$. The raw wavelet periodogram is used for estimation of the EWS and is defined as $I_{j,k;T}=|d_{j,k;T}|^2$ where $d_{j,k;T}$ denotes the process wavelet coefficient at scale $j$ and location $k$ associated to a discrete non-decimated family of wavelets as defined in \cite{nvsk:2000}.

The original LSW model does not capture the dynamics of time series data recorded for several trials over the course of an entire experiment. This setting presents additional challenges, notably the fact that these signals behave in a way that is nonstationary at multiple scales, (i) within the signal in each trial, and (ii) across trials over the course of the entire experiment.

\subsection{Replicate Locally Stationary Wavelet (RLSW) process}
\begin{definition}\label{def:rlsw}
We define a sequence of stochastic processes $\{X_{t;T}^{r;R}\}$, with time $t=0, \ldots, T-1$ where $T=2^{J(T)}$ and replicate $r=0, \ldots, R-1$ where $R=2^{J(R)}$ to be a \textit{replicate locally stationary wavelet} process if it admits the following representation
\begin{equation} \label{eq:rlsw}
X_{t;T}^{r;R} = \sum_{j=1}^{\infty}\sum_{k\in\mathbb{Z}} \omega_{j,k;T}^{r;R}\psi_{j,k}(t)\xi_{j,k}^{r},
\end{equation}
where within each replicate $r$, for a scale $j$ and time $k$, $\omega_{j,k;T}^{r;R}$ are the amplitudes for the non-decimated wavelets $\psi_{j,k}(t)$ and $\{\xi_{j,k}^{r}\}$ are a set of orthonormal random variables with properties as detailed below. Letting $\nu = r/R$ denote rescaled replicate and $z=k/T$ denote rescaled within-trial time, the quantities in~\eqref{eq:rlsw} possess the following properties:
\begin{enumerate}
\item For all $j$, $k$ and $r$, $\E[\xi_{j,k}^{r}] = 0 \ (\Rightarrow \E[X_{t;T}^{r;R}] = 0)$.
\item $\E[\xi_{j,k}^{r} \xi_{j',k'}^{r'}] = \delta_{j,j'}\delta_{k,k'}\delta_{r,r'} \ (= \text{cov}(\xi_{j,k}^{r}, \xi_{j',k'}^{r'}))$. This amounts to assuming uncorrelated replicates.
\item For each scale $j \geq 1$, there exists a Lipschitz continuous transfer function in both rescaled time ($z$) and rescaled replicate ($\nu$), denoted by $\Wt_{j}(z,\nu)$ with the following properties

\begin{enumerate}
\item {\setlength{\abovedisplayskip}{0pt} \begin{align}\label{eq:unifW}
\sum_{j=1}^{\infty} \left|\Wt_{j}(z,\nu)\right|^{2} < \infty \ \text{uniformly in } z \in (0,1), \nu \in (0,1).
\end{align} }
\item Let $L_j^\nu$ denote the bounded Lipschitz constant corresponding to the time dimension at a particular (rescaled) replicate ($\nu$) and scale $j$. Similarly, denote by $N_j^z$  the bounded Lipschitz constant corresponding to the replicate dimension at a particular (rescaled) time ($z$) and scale $j$. Denote $L_j=\sup_{\substack{\nu\in (0,1)}} L_j^\nu$ and $N_j=\sup_{\substack{z\in (0,1)}} N_j^z$, and assume they are uniformly bounded in $j$. Further assume that
\begin{equation*}\label{eq:condL}
\sum_{j=1}^{\infty} 2^{j}L_{j} < \infty \text{ and }
\sum_{j=1}^{\infty} 2^{j} N_j< \infty.
\end{equation*}
\item \label{enum:seq1} There exist sequences of bounded replicate-specific constants $\{C_{j}^{r}\}_r$ and location-specific constants $\{D_{j}^{k}\}_k$, such that for each $T$ and $R$ respectively, the amplitudes are forced to vary slowly, in the sense that
\begin{align}
\sup_{\substack{k=0: T-1}} \left| \omega_{j,k;T}^{r;R} - \Wt_{j}\left(\frac{k}{T},\frac{r}{R}\right) \right| &\leq \frac{C_{j}^{r}}{T}, \quad \forall j, r,\label{eq:supWr}\\
\sup_{\substack{r=0: R-1}}\left| \omega_{j,k;T}^{r;R} - \Wt_{j}\left(\frac{k}{T},\frac{r}{R}\right)\right| &\leq \frac{D_{j}^{k}}{R}, \quad \forall j, k. \label{eq:supWt}
\end{align}
Denote $C_j=\sup_r C_j^r$ and $D_j=\sup_k D_j^k$ and assume the sequences $\{C_{j}\}$, $\{D_{j}\}$ fulfill $\sum_{j=1}^{\infty} 2^{j}C_{j} < \infty$ and  $\sum_{j=1}^{\infty} 2^{j}D_{j} < \infty$.
\end{enumerate}
\end{enumerate}
\end{definition}

\noindent{\bf Remark (rescaled time and replicate evolution).} Within each scale $j$, the transfer function $\Wt_{j}(z,\nu)$ controls the evolution of the amplitudes, forcing them to vary slowly over {\em both} rescaled time ($z$) and replicate ($\nu$) dimensions. The evolution of the amplitudes over time within each replicate happens in a smooth manner. The evolution across replicates is such that while the spectral properties of different replicates may also be different, however across neighbouring replicates there is a larger degree of commonality. Nevertheless, further apart replicates may display different traits. Such a meta-process evolution appears later in Figure~\ref{fig:wavplots_SIM3} (Section~\ref{sec:sims}).

\subsection{Replicate evolutionary wavelet spectrum}

As is common in spectral domain analysis (both Fourier and wavelet-based), we do not work directly with the time- and replicate-specific multiscale transfer functions $\{\Wt_{j}(\cdotp,\cdotp)\}_{j}$, but instead we define a transformed version that quantifies the contribution to the process variance attributed across scales ($j$) to each time {\em and} replicate.

As noted, the current LSW model and its related quantities are not capable of capturing the multiscale evolution of brain signals. Here, we develop a novel evolutionary wavelet spectrum corresponding to a RLSW process, a quantity that for simplicity we refer to as the {\em replicate evolutionary wavelet spectrum}.

\begin{definition}
The replicate evolutionary wavelet spectrum (REWS) at scale $j$, rescaled replicate $\nu$, rescaled within-trial time $z$ is given by
\begin{equation}\label{eq:rews}
S_{j}(z,\nu) = \left|\Wt_{j}(z,\nu)\right|^2 = \lim_{\substack{T\to\infty \\ R\to\infty}}\left(\left|\omega_{j,\lfloor zT\rfloor;T}^{\lfloor \nu R\rfloor;R}\right|^2\right),
\end{equation}
where $\lfloor zT\rfloor$ and $\lfloor \nu R \rfloor$ denote the largest integer less than or equal to $zT$ and $\nu R $, respectively.
\end{definition}

Note that from equations~\eqref{eq:supWr} and ~\eqref{eq:supWt} we directly obtain that for each $T$ and $R$ we have
\begin{align}\label{eq:approx}
\sup_{\substack{r=0: R-1}}\sup_{\substack{k=0: T-1}}\left|\omega_{j,k;T}^{r;R} - \Wt_{j}\left(\frac{k}{T},\frac{r}{R}\right)\right| &= \ORD(C_{j}{T}^{-1}) + \ORD(D_{j}{R}^{-1}),
\end{align}
hence the right-hand equality in equation~\eqref{eq:rews}.

\noindent{\bf Remark (RLSW versus LSW processes).} An innovation of the proposed RLSW model is to impose within each scale not only a smooth spectral behaviour across each (replicate) time series,  but also to constrain the `meta'-spectral evolution across replicates to happen in a smooth manner, as detailed by the conditions in Definition~\ref{def:rlsw}.  Note that a replicate locally stationary wavelet (RLSW) process is thus {\em not} to be understood only as a collection of locally stationary wavelet (LSW) processes that happen to be observed across several replicates, as this would limit its capacity to represent multiscale behaviour across trials.

\noindent{\bf Remark (bounded variation jumps).} Our theoretical development could of course be extended to encompass bounded variation jumps, but this is outside the scope of this work. Nevertheless, we show through simulation that such behaviour is well handled by the proposed methodology.

\begin{definition}
We define the replicate local autocovariance (RLACV) for the replicated process $\{X_{t;T}^{r;R}\}$ for some rescaled time $z$ and rescaled replicate $\nu$ to be given by
\begin{equation*}
c(z,\nu; \uptau) = \sum_{j=1}^\infty S_{j}(z, \nu)\Psi_{j}(\uptau),
\end{equation*}
where $\uptau$ is an integer time-lag, $z\in (0,1)$, $\nu \in (0,1)$ and $\Psi_{j}(\uptau)=\sum_{k\in\Z}\psi_{j,k}(0)\psi_{j,k}(\uptau)$ denotes the scale $j$ autocorrelation wavelet.
\end{definition}

Note that $|c(z,\nu; \uptau)|<\infty$ follows directly from the uniform bounds in $\uptau$ and $\nu$ for both the limiting amplitudes and autocorrelation wavelets (see equation~\eqref{eq:unifW}). The local autocovariance defined above can be shown to be an approximation of the process autocovariance corresponding to a particular (rescaled) replicate, as follows.

\begin{proposition}\label{prop:autocov}
For a RLSW process $\{X_{t;T}^{r;R}\}$ with properties as in Definition \ref{def:rlsw},\\ $\left| \cov(X_{\lfloor zT\rfloor;T}^{\lfloor \nu R\rfloor;T}, X_{\lfloor zT \rfloor +\uptau;T}^{\lfloor \nu R\rfloor;T} )- c(z,\nu; \uptau) \right|=\ORD(T^{-1})+\ORD(R^{-1})$,
uniformly in $\uptau$ at (rescaled) time $z$ and replicate $\nu$.
\end{proposition}

\begin{proof}
The proof appears in Section~\ref{supp:proofs:autocov} of the Supplementary Material and uses the approximation properties in the definition of the RLSW process.
\end{proof}

\subsection{Spectral estimation}\label{sec:specestRLSW}
We start our proposed estimation procedure for the REWS (and thus also for the RLACV), by first computing the raw wavelet periodogram and exploring its asymptotic properties as an estimator for the true, unknown REWS. We note that the theoretical results in this section are derived under the assumption of Gaussianity.

\begin{definition}\label{def:wp1}
We define the raw wavelet periodogram of a RLSW process $\{X_{t;T}^{r;R}\}$ as
\begin{equation}\label{eq:wp1}
I_{j,k;T}^{r;R} = \left|d_{j,k;T}^{ r;R}\right|^{2},
\end{equation}
where for scale $j$, replicate $r$ and location $k$, $d_{j,k;T}^{r;R}=\sum_{t=0}^{T-1} X_{t;T}^{r;R}\psi_{j,k}(t)$ are the process empirical wavelet coefficients constructed using a family of discrete non-decimated wavelets, $\{\psi_{j,k}(t)\}_{j,k}$.
\end{definition}

We note here that unlike the Fourier periodogram, the wavelet-based raw periodogram is typically not an unbiased estimator of the wavelet spectrum, and this will also turn out to be the case here.

For reasons that will become obvious next, we also define a transformed spectral quantity $\beta_{j}(z,\nu)=\sum_{l=1}^\infty A_{j,l} S_{l} (z,\nu)$, where $A_{j,l}= \langle \Psi_{j},\Psi_{l} \rangle = \sum _{\uptau\in\Z}\Psi_{j}(\uptau)\Psi_{l}(\uptau)$ is the inner product matrix of the autocorrelation wavelets. The invertibility of the matrix $A$ and boundedness of its inverse norm \citep{nvsk:2000} ensure that finding a well-behaved estimator of the REWS is equivalent to finding a well-behaved estimator for the spectral quantity $\beta$. Hence we next focus on estimating $\beta$.

\begin{proposition}\label{prop:EI} For a RLSW process $\{X_{t;T}^{r;R}\}$ with properties as in Definition \ref{def:rlsw}, its associated raw wavelet periodogram in equation~\eqref{eq:wp1} has the following asymptotic properties:

Expectation
\begin{equation} \label{eq:EI}
\E\left[I_{j,\lfloor zT\rfloor;T}^{\lfloor \nu R \rfloor;R}\right] = \beta_{j}(z,\nu)  + \textcolor{black}{\ORD(2^{j}T^{-1}) + \ORD(R^{-1})},
\end{equation}

Variance
\begin{equation*}\label{eq:varI}
\text{var}\left(I_{j,\lfloor zT\rfloor;T}^{\lfloor \nu R \rfloor;R}\right) = 2\beta_{j}^2(z,\nu) + \textcolor{black}{\ORD(2^{2j}T^{-1}) + \ORD(2^{j}R^{-1})}.
\end{equation*}
\end{proposition}

\begin{proof}
Appendix~\ref{app:proofs:rawi} contains the proof.
\end{proof}

From Proposition \ref{prop:EI}, we see that the raw periodogram is asymptotically unbiased for $\beta$, but inconsistent due its asymptotically non-vanishing variance. Thus we next propose to smooth the raw periodogram in order to obtain consistency (and then we can correct for bias to obtain an asymptotically unbiased estimator for $S$).

\begin{definition}
We define a replicate-smoothed estimator for the quantity $\beta_{j}(\frac{k}{T},\frac{r}{R})$ to be
\begin{equation}\label{eq:defItilde}
\It_{j,k;T}^{r;R} = (2M + 1)^{-1} \sum_{s=-M}^{M}I_{j,k;T}^{r+s;R},
\end{equation}
where $(2M+1)$ is the length of the smoothing window and $M$ is an integer such that as $T, R \to \infty$, we have that $M \to \infty$ and $M/R \to 0$.
\end{definition}

\noindent{\bf Remark (replicate smoothing).}
Unlike for the usual locally stationary processes where the periodogram is smoothed over time in order to achieve consistency, here we propose a smoothing procedure that operates over {\em replicates} by locally averaging the spectral estimates across a window of neighbouring replicates. This approach is indeed theoretically justified by the assumption of spectral smoothness across the replicate-dimension. In practice, these assumptions will have to be verified in order to determine some empirically guided choice of $M$, as seen in the simulation study Section~\ref{sec:sims}.

\begin{proposition}\label{prop:ESI} For a RLSW process $\{X_{t;T}^{r;R}\}$ with properties as in Definition \ref{def:rlsw}, the replicate-smoothed wavelet periodogram in equation~\eqref{eq:defItilde} has the following asymptotic properties:

Expectation
\begin{equation*} \label{eq:ESI}
\E\left[\It_{j,\lfloor zT\rfloor;T}^{\lfloor \nu R \rfloor;R}\right] = \beta_{j}(z,\nu)
+\textcolor{black}{\ORD(MR^{-1}) + \ORD(2^{j}T^{-1})},
\end{equation*}

Variance
\begin{equation*}\label{eq:varSI}
\text{var}\left(\It_{j,\lfloor zT\rfloor;T}^{\lfloor \nu R \rfloor;R}\right) = 
\textcolor{black}{\ORD(2^{2j}M^{-1})+\ORD(2^jR^{-1})+\ORD(MR^{-2})}.
\end{equation*}
\end{proposition}

\begin{proof}
Appendix~\ref{app:proofs:itilde} contains the proof which manipulates the amplitude properties across replicates as opposed to those across time.
\end{proof}

Note that as $T$, $R$ and $M \to \infty$ and using the condition $M/R \to 0$, the bias of the smoothed periodogram becomes asymptotically negligible, while its variance tends to zero for any fixed fine enough scale $j$ (with $2^j=o\left(\mbox{min}\left\{T,R,(2M+1)^{1/2}\right\}\right)$). The usual bias--variance trade-off here is manifest through the increase of $M$ resulting in a decrease of the variance at the price of an increase in the bias. As the replicate-smoothed periodogram proposed above is an asymptotically unbiased and consistent estimator for the true $\beta$, the relationship between the true spectral quantities $\beta$ and $S$ suggests a natural way of constructing a well-behaved spectral estimator for the unknown spectrum $S$.
We thus propose to estimate the unknown REWS by means of
\begin{equation}\label{eq:Sestr}
\hat{S}_{j}(z,\nu) = \sum_{l=1}^{J}A_{j,l}^{-1}\It_{l,\lfloor zT \rfloor;T}^{\lfloor \nu R \rfloor;R},
\end{equation}
where $A_{j,l}^{-1}$ is the $(j,l)$ entry of the inverse of the inner product matrix $A$ of the autocorrelation wavelets and $J=\lfloor\alpha J(T)\rfloor$ with $\alpha\in(0,1)$.

\begin{proposition}\label{prop:CS}
For a RLSW process $\{X_{t;T}^{r;R}\}$ with properties as in Definition \ref{def:rlsw}, the spectral estimator proposed in equation~\eqref{eq:Sestr} is asymptotically unbiased and consistent for $S_{j}(z,\nu)$ for each fixed scale $j$, rescaled time and replicate $z$ and $\nu$, respectively, provided that $M/R \to 0$ as $T$, $R$ and $M \to \infty$.
\end{proposition}

\begin{proof}
Appendix~\ref{app:proofs:corrper} contains the proof which hinges on the properties of the replicate-smoothed periodogram shown in Proposition~\ref{prop:ESI} above.
\end{proof}

\noindent{\bf Remark (replicate and time smoothing).}
The results in Proposition \ref{prop:ESI} highlight the small sample dependence of the bias and variance of the smoothed periodogram on the number of replicates $R$, on the time series length $T$ and on the smoothing window $(2M+1)$, as well as well as on the ratio of (replicate) smoothing window to the total number of replicates. While still having a bias--variance tradeoff, the variance can be further improved by additionally smoothing across the time-dimension.

Specifically, using a time-smoothing window of length $(2M_T+1)$ such that $M_T \to \infty$ and $M_T/T \to 0$ (the reader may also refer to \cite{park:mvlsw}) and chosen as usual under LSW modelling (see e.g. \cite{nas:station}), and preserving the previous notation of $(2M+1)$ for the replicate-smoothing window, we define the replicate- and time-smoothed periodogram
\begin{equation}\label{eq:defItildetilde}
\tilde{\It}_{j,k;T}^{r;R} = (2M + 1)^{-1} (2M_T + 1)^{-1} \sum_{s=-M}^{M} \sum_{t=-M_T}^{M_T} I_{j,k+t;T}^{r+s;R},
\end{equation}
to act as an estimator for the transformed spectral quantity $\beta_{j}(\frac{k}{T},\frac{r}{R})$. Below we show that this estimator has desirable asymptotic properties, leading to faster convergence than its counterpart involving only replicate-smoothing.

\begin{proposition}\label{prop:ESIts} For a RLSW process as in Definition \ref{def:rlsw} and satisfying the additional assumption of autocovariance summability, $\sup_{z,\nu}\sum_{n\in\Z}\left|c(z,\nu;n\right|)=\ORD(1)$, the smoothed time- and replicate-specific wavelet periodogram defined in equation~\eqref{eq:defItildetilde} has the following asymptotic properties:

Expectation
\begin{equation*} \label{eq:ESIts}
\E\left[\tilde{\It}_{j,\lfloor zT\rfloor;T}^{\lfloor \nu R \rfloor;R}\right] = \beta_{j}(z,\nu)+ \ORD(M_T T^{-1})
+\textcolor{black}{\ORD(MR^{-1}) + \ORD(2^{j}T^{-1})},
\end{equation*}

Variance
\begin{equation*}\label{eq:varSIts}
\text{var}\left(\tilde{\It}_{j,k;T}^{r;R}\right)
=\ORD(2^{2j}(M_TM)^{-1})+\ORD(2^{j}(MR)^{-1}).
\end{equation*}
\end{proposition}

\begin{proof}
Appendix~\ref{app:proofs:itildetilde} contains the proof which makes use of the smoothing in both directions.
\end{proof}

The replicate- and time- smoothed periodogram can then be used to further build a well-behaved estimator of the unknown REWS by means of
\begin{equation*}\label{eq:Sestrt}
\hat{\hat{S}}_{j}(z,\nu) = \sum_{l=1}^{J}A_{j,l}^{-1}\tilde{\It}_{l,\lfloor zT \rfloor;T}^{\lfloor \nu R \rfloor;R}.
\end{equation*}
It is straightforward to show that this is also asymptotically unbiased and consistent for $S_j(z,\nu)$, in the same manner as in the proof of Proposition \ref{prop:CS}.

\section{Simulation Study}\label{sec:sims}

Here we aim to assess the behaviour of our proposed RLSW methodology as well as compare it to a classical approach involving the LSW methodology \citep{nvsk:2000}. Specifically, we evaluate (i) the classical approach where one would independently estimate the spectrum for each replicate using a localised time smoother and then average over all replicates (`LSW'), (ii) our proposed methodology involving localised smoothing over replicates (`RLSW$_1$'), and (iii) our proposed methodology involving localised smoothing over time and replicates (`RLSW$_2$'). In order to match the current practice for LSW estimation, e.g. \cite{nvsk:2000, park:mvlsw}, we have set $J=J(T)$ (corresponding to $\alpha=1$), although in a bivariate spectral estimation context \cite{sand:2010} set a similar measure to $\alpha=0.7$ and remark on its improved results when compared to $\alpha=1$. We carry out simulations over $N=100$ runs and explore performance across a range of time series lengths $T$ from $128$ to $1024$, number of replicates $R$ from $64$ to $512$ and smoothing windows $(2M + 1)$ from $9$ to $25$. We report the mean squared errors (MSE) and squared bias results.

Overall, based on our findings, we recommend the use of the method involving both time and replicate smoothing (RLSW$_2$) with a window length choice of $(2M+1)=0.15 R$ as a rule of thumb.

\noindent{\bf Illustrative example}. We choose to present here the behaviour of our proposed methodology on a process with a challenging spectral structure, as shown in Figure~\ref{fig:wavplots_SIM3} and mathematically defined in Section~\ref{supp:extradetail} of the Supplementary Material. Further simulation findings are detailed in Section~\ref{supp:furthersims} of the Supplementary Material, including MSE and squared bias tables.

For $R = 256$ replicates each of length $T = 256$, the process places spectral content at level $j = 5$, manifest through a decreasing amplitude of the cosine across the last 192 replicates, and at level $j = 6$, where the periodicity of the cosine increases across the first 128 replicates.

\begin{figure}[h!]
\begin{center}
\resizebox{!}{0.4\textwidth}{\includegraphics{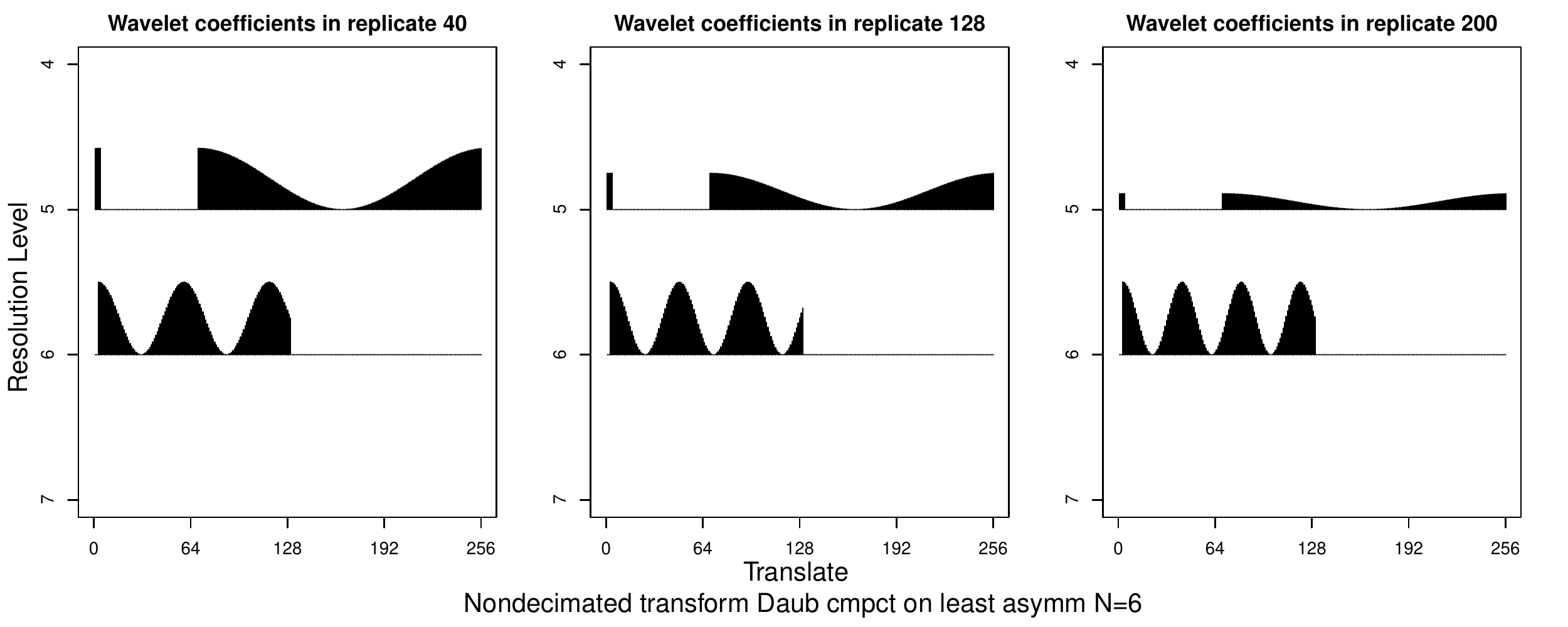}}
\end{center}
\vspace{-1em}
\caption{\small True evolutionary wavelet spectra for replicates 40, 128 and 200.}
\label{fig:wavplots_SIM3}
\end{figure}

A concatenated realisation of this process is shown in Figure~\ref{fig:reptsBY7_SIM3proc}. Note however that this is an abuse of representation, since each replicate is a time series of its own, and the sole purpose of this visualisation is to highlight the evolution of the meta-process. Furthermore this process departs somewhat from the requirement that the amplitudes evolve slowly over both rescaled time $(z)$ and replicate $(\nu)$ dimensions, however we show that despite this the methodology still performs well.

\begin{figure}[h!]
\begin{center}
\resizebox{!}{0.6\textwidth}{\includegraphics{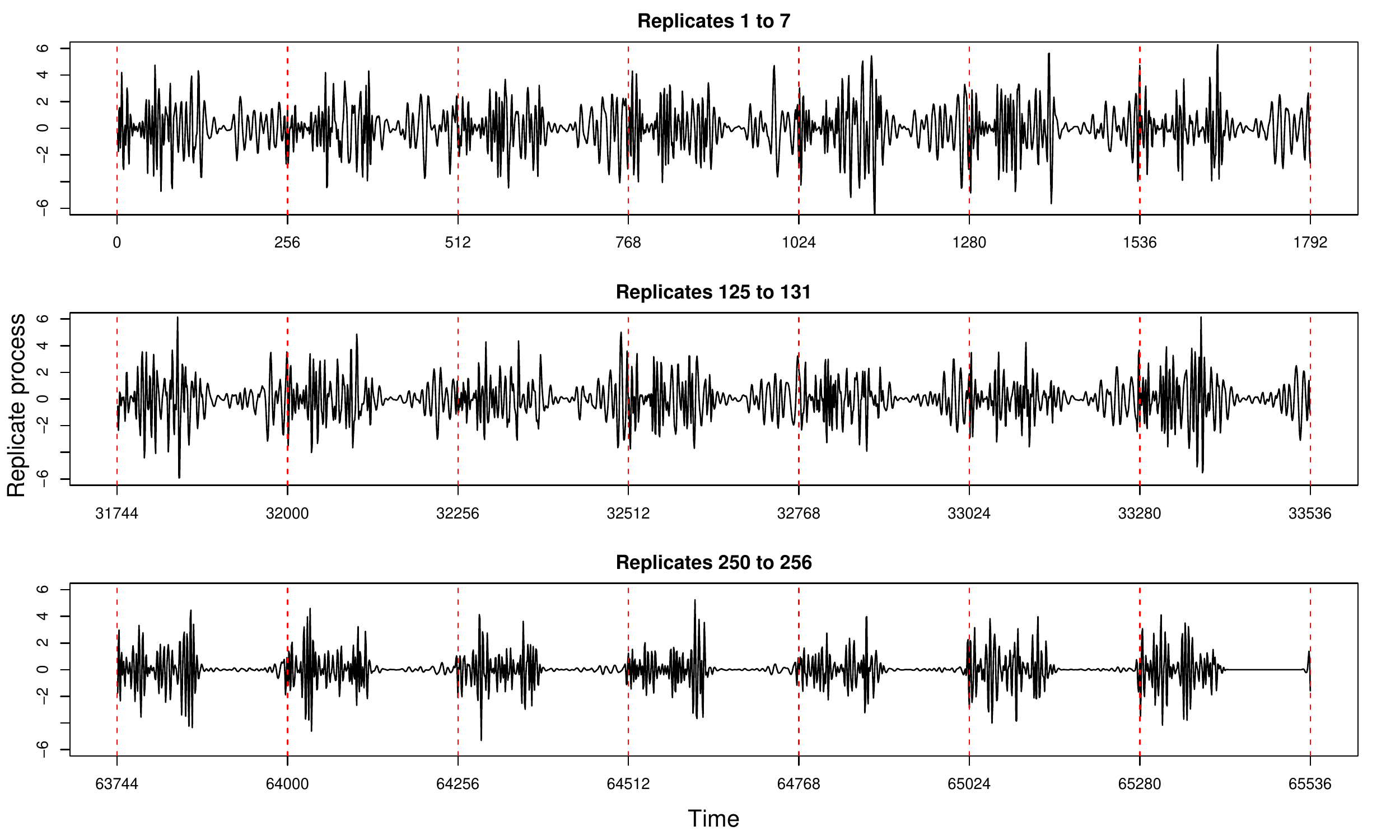}}
\end{center}
\vspace{-2em}
\caption{\small Realisation of a RLSW process with spectra as in Figure~\ref{fig:wavplots_SIM3}.}
\label{fig:reptsBY7_SIM3proc}
\end{figure}

Spectral estimates have been computed using discrete non-decimated wavelets built by means of Daubechies Least Asymmetric family with 6 vanishing moments (see \cite{daub:ten} for an understanding of Daubechies compactly supported wavelets). For the RLSW method, local averaging involved windows of 9 replicates corresponding to $M = 4$ and we note that numerical MSE results in Appendix~\ref{app:tables:sims} highlight that we chose to visually present here some of our least performant results. The LSW and RLSW($_1$) spectral estimates appear in Figure~\ref{fig:coef_SIM3proc_R256_T256}, along with the truth.

From the figures we get a visual clarification that the RLSW model is doing a good job at capturing the evolving characteristics of the spectra across replicates and the leakage across the neighbouring levels $j = 5$ and 6 is minor. They also highlight that when neglecting the possibility of evolutionary behaviour over replicates, when it is in fact present as seen for levels $5$ and $6$ in the top row plots of the true spectrum, the LSW model struggles to reflect this and either under or over-estimates, as seen in the middle row plots. The bottom row plots show that the RLSW$_1$ estimates do indeed pick up the evolution over replicates. Figures~\ref{fig:lev2_SIM3proc_R256_T256} and~\ref{fig:lev3_SIM3proc_R256_T256} further support the evolutionary behaviour of the spectral quantities over time and replicates in levels 5 and 6, respectively.

\begin{figure}[h!]
\begin{center}
\resizebox{!}{0.6\textwidth}{\includegraphics{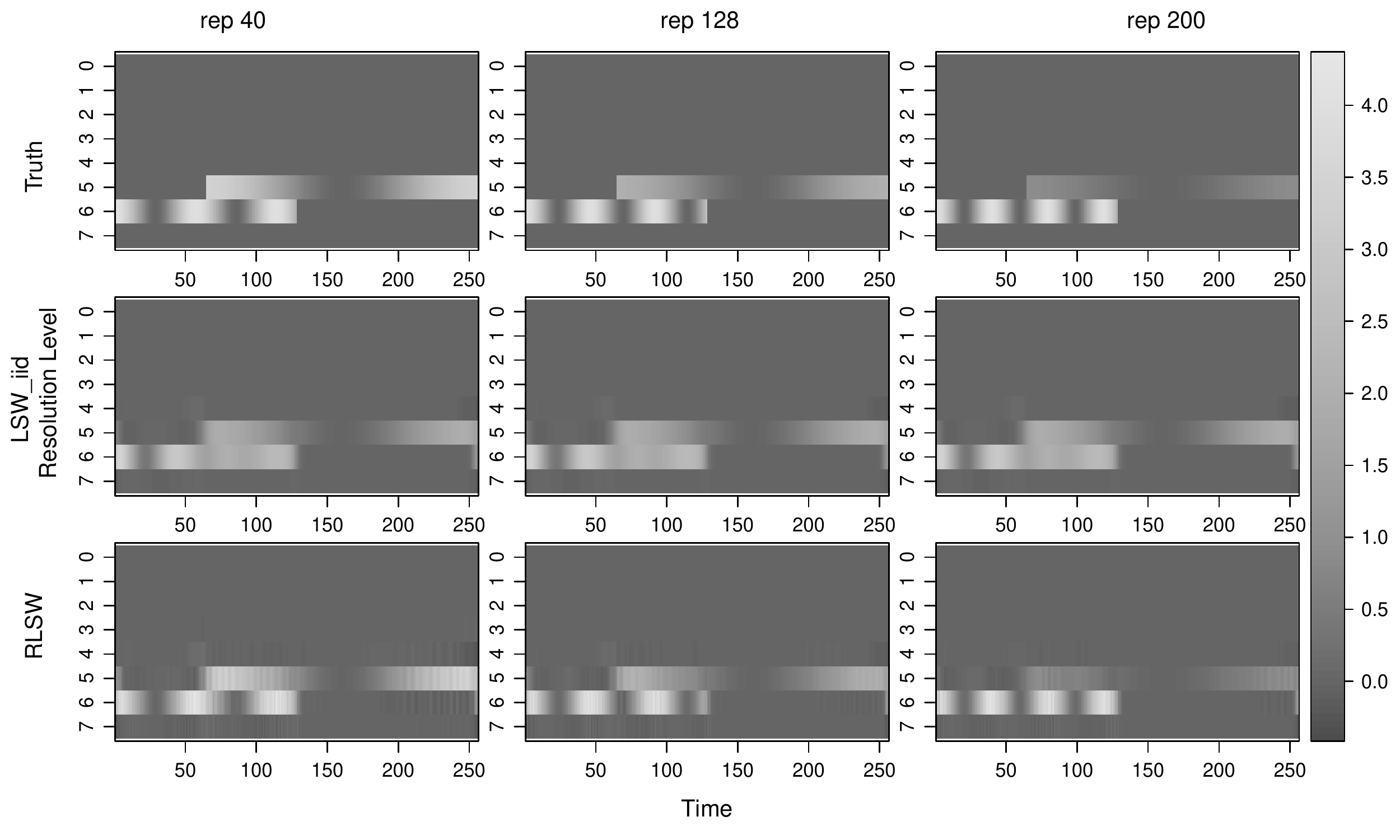}}
\end{center}
\vspace{-2em}
\caption{\small Time-scale plots for replicates 40, 128 and 200, respectively the first, second and third columns. Estimates are averaged  over 100 realisations. \textit{Top}: true spectra; \textit{Middle}: estimates from the LSW method averaged over all replicates; \textit{Bottom}: estimates using RLSW($_1$). }
\label{fig:coef_SIM3proc_R256_T256}
\end{figure}

\begin{figure}[htbp!]
\begin{center}
\resizebox{!}{0.6\textwidth}{\includegraphics{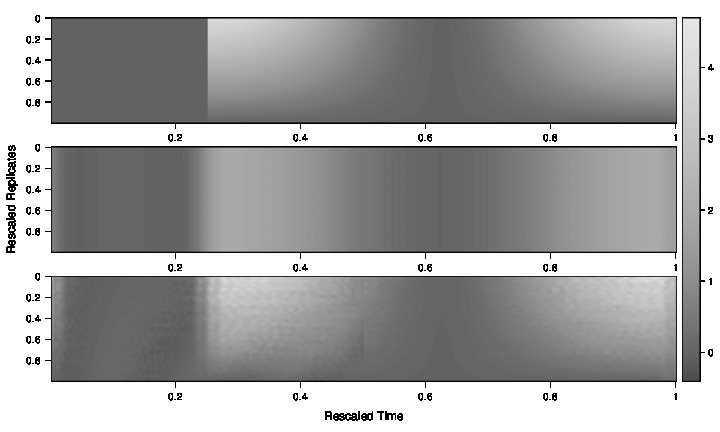}}
\end{center}
\vspace{-2em}
\caption{\small Time-replicate spectral plots for level 5. Estimates are averaged  over 100 realisations. \textit{Top}: true spectra; \textit{Middle}:  estimated spectra from the LSW method averaged over all replicates; \textit{Bottom}: estimates using RLSW($_1$). }
\label{fig:lev2_SIM3proc_R256_T256}
\end{figure}

\begin{figure}[htbp!]
\begin{center}
\resizebox{!}{0.6\textwidth}{\includegraphics{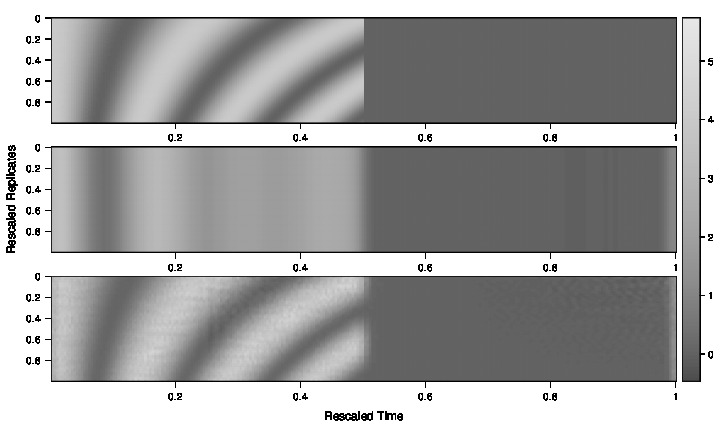}}
\end{center}
\vspace{-2em}
\caption{\small Time-replicate spectral plots for level 6. Estimates are averaged  over 100 realisations. \textit{Top}: true spectra; \textit{Middle}:  estimated spectra from the LSW method averaged over all replicates; \textit{Bottom}: estimates using RLSW($_1$).}
\label{fig:lev3_SIM3proc_R256_T256}
\end{figure}

Histograms of the MSEs over the 100 simulations are shown in Section~\ref{supp:extradetail} (Supplementary Material) and highlight not only how the increase in $M$ improves performance but also how the increase in $R$ and $T$ reduces the MSEs, thus demonstrating the expected asymptotic behaviour of our smoothed estimator. To numerically strengthen our visual inference, we examine the MSEs and squared bias results in Table~\ref{tab:msesim3} (Appendix~\ref{app:tables:sims}). The best results in terms of lowest MSEs are obtained by RLSW$_2$. Specifically, RLSW$_2$ does incur a somewhat higher bias than RLSW$_1$, but both our methods have a substantially lower bias than the LSW.
The higher bias of the LSW estimates is unsurprisingly resulting from averaging over all the replicates and thus failing to account for the evolutionary behaviour through replicates. The benefit of taking a local smoothing approach over both time and replicates is that it always results in spectral estimates with lower bias and MSE when compared to LSW, although it is worth pointing out that taking a local smoothing approach over replicates only, while yielding lower bias, might increase the MSE for inappropriately small windows.

Asymptotically, the MSEs associated to our methods decay much faster than for the LSW and as we increase the local averaging window length ($2M + 1$), the performance of our the RLSW methodology improves. A replicate window length choice of $(2M+1)$ equal to about 15\% of $R$ appears to work well across all our investigations.

\section{RLSW model embedding replicate coherence}\label{sec:modelext}
So far, a serious limitation is the assumption that the replicate time series are uncorrelated. We now develop the theory to allow for cross-trial dependence by means of a coherence structure along the replicate dimension. This is a major innovation of this work, as to the best of our knowledge this is the first paper that accounts for correlation across trials.

\begin{definition}\label{def:rlswc}
Let the sequence of stochastic processes $\{X_{t;T}^{r;R}\}_{t,r}$ be a \textit{replicate locally stationary wavelet} process as defined in Definition \ref{def:rlsw} with the following amendments to the properties:
\begin{enumerate}
\setcounter{enumi}{1}
\item (\textit{replacing property} 2) Additional to $\{\xi_{j,k}^{r}\}_{j,k}$ being orthonormal within replicate $r$, we have $\E[\xi_{j,k}^{r} \xi_{j',k'}^{r'}] = \delta_{j,j'}\delta_{k,k'}\rho_{j,k;T}^{r,r';R}$, where $\{\rho_{j,k;T}^{r,r';R}\}_k$ determine the innovation dependence structure between replicates $r$ and $r'$, at each scale $j$. Note that the within-replicate orthonormality induces $\left|\rho_{j,k;T}^{r,r';R}\right| \leq1$ for all $j$, $k$ and $r$, $r'$, with equality when $r=r'$.
\item (\textit{replacing property} 3(a))
\begin{align}\label{eq:unifWc}
\sum_{j=1}^{\infty} \left|\Wt_{j}(z,\nu)\Wt_{j}(z,\nu^\prime)\right| < \infty \ \text{uniformly in } z \in (0,1), \nu , \nu^\prime \in (0,1).
\end{align}
\item (\textit{additional property}) For each scale $j \geq 1$, there exists a Lipschitz continuous function in rescaled time ($z$) and rescaled replicate arguments $\nu$ and $\nu'$, denoted by $\rho_{j}(z, \nu, \nu')$, which constrains the covariance structure and fulfills the assumptions below, as follows
\begin{enumerate}
\item Let $Q_j^{\nu,\nu'}$ denote the bounded Lipschitz constant corresponding to the time dimension at particular (rescaled) replicates $\nu$ and $\nu'$, at scale $j$. Similarly, denote by $P_j^z$  the bounded Lipschitz constant corresponding to the replicate dimension at a particular (rescaled) time ($z$), at scale $j$. Denote $Q_j=\sup_{\substack{\nu,\nu'\in (0,1)}} Q_j^{\nu,\nu'}$, $P_j=\sup_{\substack{z\in (0,1)}} P_j^z$ and assume they are uniformly bounded in $j$. Further assume that
\begin{equation*}\label{eq:condLrho}
\sum_{j=1}^{\infty} 2^{j}Q_{j} < \infty \text{ and }
\sum_{j=1}^{\infty} 2^{j} P_j< \infty.
\end{equation*}
\item \label{enum:seq2} There exist sequences of bounded replicate-specific constants $\{\tilde{C}_{j}^{r,r'}\}_{r,r'}$ and location-specific constants $\{\tilde{D}_{j}^{k}\}_k$, such that for each $T$ and $R$ respectively, the covariances are forced to vary slowly, in the sense that
\begin{align}
\sup_{\substack{k=0: T-1}} \left| \rho_{j,k;T}^{r,r';R} - \rho_{j}\left(\frac{k}{T},\frac{r}{R},\frac{r'}{R}\right) \right| &\leq \frac{\tilde{C}_{j}^{r,r'}}{T}, \quad \forall j, r, r'\label{eq:supcohr}\\
\sup_{\substack{r,r'=0: R-1}}\left| \rho_{j,k;T}^{r,r';R} - \rho_{j}\left(\frac{k}{T},\frac{r}{R},\frac{r'}{R}\right)\right| &\leq \frac{\tilde{D}_{j}^{k}}{R}, \quad \forall j, k. \label{eq:supcoht}
\end{align}
Denote $\tilde{C}_j=\sup_{r,r'} \tilde{C}_j^{r,r'}$ and $\tilde{D}_j=\sup_k \tilde{D}_j^k$ and assume the sequences $\{\tilde{C}_{j}\}$, $\{\tilde{C}_{j}\}$ fulfill $\sum_{j=1}^{\infty} 2^{j}\tilde{C}_{j} < \infty$ and $\sum_{j=1}^{\infty} 2^{j}\tilde{D}_{j} < \infty$.
\end{enumerate}
\end{enumerate}
\end{definition}

\noindent{\bf Remark (rescaled replicate dependence).}
Note that from equations~\eqref{eq:supcohr} and ~\eqref{eq:supcoht} we directly obtain for each $T$ and $R$
\begin{align*}\label{eq:approxc}
\sup_{\substack{r,r'=0: R-1}}\sup_{\substack{k=0: T-1}}\left|\rho_{j,k;T}^{r,r';R} - \rho_{j}\left(\frac{k}{T},\frac{r}{R},\frac{r'}{R}\right) \right| &= \ORD(\tilde{C}_{j}{T}^{-1}) + \ORD(\tilde{D}_{j}{R}^{-1}).
\end{align*}

Hence for some rescaled time $z$ and rescaled replicates $\nu$ and $\nu'$, we have in the limit
\begin{equation*}
\rho_{j}(z, \nu,\nu') = \lim_{\substack{T\to\infty \\ R\to\infty}}\left(\rho_{j,\lfloor zT\rfloor;T}^{\lfloor \nu R\rfloor, \lfloor \nu' R\rfloor;R}\right).
\end{equation*}

For a scale $j$ and time location $k$, the quantity $\rho_{j}(\frac{k}{T}, \frac{r}{R},\frac{r'}{R})$ thus gives a measure of the dependence between replicates $r$ and $r'$.

\subsection{Evolutionary wavelet replicate--(cross-)spectrum and coherence}
We next define a measure for the scale, time and cross-replicate contribution to the overall process variance.

\begin{definition}
The {\em evolutionary wavelet replicate--cross-spectrum} defined at scale $j$, rescaled time $z$ within rescaled replicates $\nu$ and $\nu'$ is given by
\begin{equation*}
S_{j}(z,\nu,\nu') = \Wt_{j}(z,\nu)\Wt_{j}(z,\nu')\rho_{j}(z, \nu, \nu').
\end{equation*}
\end{definition}

In order to simplify the terminology, we also refer to this quantity as the \textit{replicate evolutionary wavelet cross-spectrum} (REWCS). Note in the above definition that the replicate--cross-spectrum corresponding to any rescaled replicate $\nu$, is nothing else but the spectrum corresponding to that replicate, \i.e. $S_{j}(z,\nu,\nu) =S_{j}(z,\nu)$. Alternatively, $|\tilde{W}_{j}(z,\nu)|=\left(S_{j}(z,\nu,\nu)\right)^{1/2}$.

\begin{definition}
For a replicate locally stationary wavelet process $\{X_{t;T}^{r;R}\}$, we define the replicate local cross-covariance (RLCCV) at rescaled time $z\in (0,1)$ within rescaled replicates $\nu$ and $\nu'$, both in $(0,1)$, to be given by
\begin{equation*}
c(z,\nu, \nu'; \uptau) = \sum_{j=1}^{\infty}S_{j}(z, \nu, \nu')\Psi_{j}(\uptau),
\end{equation*}
where $\uptau \in \mathbb{Z}$ is an integer time-lag and we recall that $\Psi_{j}(\uptau)=\sum_{k\in\Z}\psi_{j,k}(0)\psi_{j,k}(\uptau)$ is the scale $j$ autocorrelation wavelet.
\end{definition}

Akin to established time series literature, we can next rephrase the dependence measure as
\begin{equation} \label{eq:lsrc}
\rho_{j}(z, \nu, \nu') = \frac{S_{j}(z,\nu,\nu')}{\left\{S_{j}(z,\nu)S_{j}(z,\nu')\right\}^{1/2}},
\end{equation}
and we shall refer to it as the \textit{locally stationary replicate--coherence}, with values ranging from $-1$, indicating an absolute negative correlation, to $1$ indicating an absolute positive correlation.

Note that $|c(z,\nu, \nu'; \uptau)|<\infty$ follows directly from the coherence range between $-1$ and $1$, and from the uniform bounds in lag ($\uptau$) and rescaled-replicate time ($\nu$, $\nu'$) for both the limiting amplitudes and the autocorrelation wavelets (see equation~\eqref{eq:unifWc}).

Unsurprisingly, the local autocovariance can be shown to be an approximation of the process autocovariance, as follows.

\begin{proposition}\label{prop:autocovc}
For a RLSW process $\{X_{t;T}^{r;R}\}$ with properties as in Definition \ref{def:rlswc},\\ $\left| \cov(X_{\lfloor zT\rfloor;T}^{\lfloor \nu R\rfloor;T}, X_{\lfloor zT \rfloor +\uptau ;T}^{\lfloor \nu' R\rfloor;T} )- c(z,\nu,\nu'; \uptau) \right|=\ORD(T^{-1})+\ORD(R^{-1})$,
uniformly in $\uptau$ at (rescaled) time $z$ and replicates $\nu$, $\nu'$.
\end{proposition}

\begin{proof}
The proof appears in Section~\ref{supp:proofs:autocovc} of the Supplementary Material.
\end{proof}

\subsection{Estimation theory}\label{sec:modelext_est}
Recall, in the absence of cross-trial dependence, the raw wavelet periodogram is given by $I_{j,k;T}^{r;R} = |d_{j,k;T}^{ r;R}|^{2}$. We next introduce a cross-replicate version.

\begin{definition}
For a scale $j$ and location $k$, we define the raw wavelet cross-periodogram between replicates $r$ and $r'$ of a RLSW process to be
\begin{equation}\label{eq:rawIc}
I_{j,k;T}^{(r,r');R} = d_{j,k;T}^{ r;R}d_{j,k;T}^{ r';R}.
\end{equation}
\end{definition}

\noindent The theoretical results below are derived under the Gaussianity assumption.

\begin{proposition}\label{prop:EIc}
For a RLSW process $\{X_{t;T}^{r;R}\}$ as in Definition~\ref{def:rlswc}, the wavelet cross-periodogram has the following asymptotic properties:

Expectation
\begin{equation*} 
\E\left[I_{j,\lfloor kT \rfloor;T}^{(\lfloor \nu R \rfloor,\lfloor \nu' R \rfloor);R}\right] = \sum_{l=1}^{\infty}A_{j,l} S_{l}\left(z,\nu,\nu'\right) + \ORD(2^jT^{-1}) + \ORD(R^{-1}),
\end{equation*}

Variance
\begin{align*}\nonumber
\text{var}\left(I_{j,\lfloor kT \rfloor;T}^{\lfloor \nu R \rfloor,\lfloor \nu' R \rfloor;R}\right) &= \left(\sum_{l=1}^{\infty}A_{j,l} S_{l}\left(z,\nu\right)\right)
\left(\sum_{l=1}^{\infty}A_{j,l} S_{l}\left(z,\nu'\right)\right) \nonumber \\
& \qquad + \left(\sum_{l=1}^{\infty}A_{j,l} S_{l}\left(z,\nu,\nu'\right)\right)^{2} + \ORD(2^{2j}T^{-1}) + \ORD(2^{j}R^{-1}).\nonumber
\end{align*}
\end{proposition}

\begin{proof}
The proofs follow similar steps to their non-coherence counterpart in Appendix~\ref{app:proofs:rawi} and are thus omitted here.
\end{proof}
\noindent{\bf Remark (replicate smoothing).} In line with the spectral `similarity' amongst neighbouring replicates, we proceed by smoothing the cross-periodograms across replicates for consistency and then correcting them for bias.

\begin{definition}
We define a smoothed estimator across the replicate dimension to be
\begin{equation}\label{eq:reppersmooth}
\It_{j,k;T}^{(r,r');R} = (2M + 1)^{-1} \sum_{s=-M}^{M}I_{j,k;T}^{(r+s,r'+s);R},
\end{equation}
where $(2M+1)$ is the length of the smoothing window and $M$ is an integer such that as $T, R \to \infty$, we have that $M \to \infty$ and $M/R \to 0$.
\end{definition}

\begin{proposition}\label{prop:ESIc} Under the properties of Definition \ref{def:rlswc} and the additional assumption\\ $\sup_{z,\nu\in (0,1)}\sum_{\eta\in\Z} \left| c(z,\nu,\nu+\frac{\eta}{R};\uptau) \right|=\ORD(1)$ for any time lag $\uptau$, the {\em replicate-smoothed} wavelet cross-periodogram defined above has the following asymptotic properties:

Expectation
\begin{equation*}\label{eq:ESIc}
\E\left[\It_{j,\lfloor kT \rfloor;T}^{(\lfloor \nu R \rfloor,\lfloor \nu' R \rfloor);R}\right] = \sum_{l=1}^{\infty}A_{j,l} S_{l}\left(z,\nu,\nu'\right)+ \ORD(MR^{-1}) + \ORD(2^{j}T^{-1}),
\end{equation*}

Variance
\begin{equation*}\nonumber
\text{var}\left(\It_{j,\lfloor kT \rfloor;T}^{(\lfloor \nu R \rfloor,\lfloor \nu' R \rfloor);R}\right)=\ORD(2^{2j}M^{-1})+ \ORD(2^j R^{-1})+ \ORD(M^2R^{-2}).
\end{equation*}
\end{proposition}

\begin{proof}
Appendix~\ref{app:proofs:itildec} contains the proof.
\end{proof}

The bias of the replicate-smoothed wavelet cross-periodogram becomes asymptotically negligible and $\text{var}\left(\It_{j,k;T}^{(r,r');R}\right) \to 0$ for any scale $j$ with $2^j=o\left( \mbox{min}\{ T, R, (2M+1)^{1/2}\}\right)$ as $T$, $R$, $M \to \infty$ and $M/R \to 0$. Then correcting for the bias will yield a desirable spectral estimator, as follows.

\begin{proposition}\label{prop:CSC}
The following is an asymptotically unbiased and consistent replicate-smoothed estimator for the REWCS
\begin{equation}\label{eq:corrspecc}
\hat{S}_{j}(z,\nu, \nu') = \sum_{l=1}^{J}A_{j,l}^{-1}\It_{l,\lfloor zT \rfloor;T}^{(\lfloor \nu R \rfloor,\lfloor \nu' R \rfloor);R},
\end{equation}
where $A_{j,l}^{-1}$ is the $(j,l)$ entry of the inverse of the inner product matrix $A$ of the autocorrelation wavelets and $J=\lfloor \alpha J(T) \rfloor$ with $\alpha \in (0,1)$.
\end{proposition}

\begin{proof}
The proof follows the same steps as for Proposition~\ref{prop:CS}.
\end{proof}

This paves the way towards proposing the replicate--coherence estimator
\begin{equation}\label{eq:rhoc}
\hat{\rho}_{j}(z,\nu, \nu') = \frac{\hat{S}_{j}(z,\nu,\nu')}{\left\{\hat{S}_{j}(z,\nu)\hat{S}_{j}(z,\nu')\right\}^{1/2}},
\end{equation}
where the involved spectral quantities are consistently estimated as proposed in Proposition~\ref{prop:CSC} (see also equation~\eqref{eq:reppersmooth}) and the use of the same smoothing windows guarantees that the values of the resulting coherence estimator are indeed quantities between $-1$ and $1$ (a proof appears in Section~\ref{supp:proofs:rhoclim} of the Supplementary Material).

The following proposition shows that the step of examining $\hat{\rho}_{j}(z,\nu, \nu')$ is theoretically justified.
\begin{proposition}\label{prop:rhoc}
Under the assumptions of Proposition~\ref{prop:ESIc}, the coherence estimator in~\eqref{eq:rhoc} is asymptotically consistent for the true coherence $\rho_{j}(z,\nu, \nu')$.
\end{proposition}

\begin{proof}
See Appendix~\ref{app:proofs:rhoc}.
\end{proof}

\section{Coherence illustration via simulation} \label{sec:simsext}

We shall now investigate through simulation the performance of our proposed methodology for coherence estimation. We use the mean squared error (MSE) and squared bias of the estimates $\hat \rho$, averaged over all time-scale points and replicates, as detailed in Appendix~\ref{app:tables:simsext}. We display the behaviour of our estimators on a simulated example below, and provide a further simulation study in Section~\ref{supp:furthersimscoh} of the Supplementary Material.

\noindent{\bf Illustrative example}. We simulate a replicate locally stationary wavelet process with $R=256$ replicates that feature dependence, measured at $T=512=2^{9}$ time points. The locally stationary wavelet autospectra are defined by a sine wave whose periodicity and magnitude evolve slowly over the replicates in such a way that the spectral characteristics of neighbouring replicates do not look too dissimilar whilst there is a noticeable difference between replicates further apart (for their mathematical expression, see Simulation 1, Section~\ref{supp:furthersims} of the Supplementary Material). Here we have $J(T)=9$ (in short, $J$) and the spectral characteristics are placed in level $j=J(T)-4=5$. In addition to the autospectral characteristics, we also define a challenging cross-replicate spectral structure by means of defining their (true) coherence at each level $j$ and location $k$. Specifically, we set the coherence to be zero over the last $256$ locations, yielding an $R \times R$ coherence identity matrix. For level $j=5$ and time $k = 1,\ldots, 256$ we define the non-zero replicate coherence matrices as follows: the first $128$ replicates have a strong positive coherence (0.99) with one another, however this coherence becomes negative (-0.71) with the last $128$ replicates. A (weaker) positive coherence (0.5) also exists between the last $128$ replicates.  The expressions of the non-zero coherence matrices can be found in Section~\ref{supp:furthersims} of the Supplementary Material and the illustrative true coherence structures for replicates 50 (top row) and 200 (bottom row) can be visualised in Figure~\ref{fig:coh_sim2} (left panels).

Coherence estimates obtained using the methodology proposed in Section~\ref{sec:modelext_est} are represented in Figure~\ref{fig:coh_sim2} (right panels) for replicates 50 (top row) and 200 (bottom row). Non-decimated discrete wavelets built using Daubechies least asymmetric family with 10 vanishing moments and local averaging over a window of 9 replicates were employed.

\begin{figure}[htbp!]
\begin{center}
\resizebox{!}{0.6\textwidth}{\includegraphics{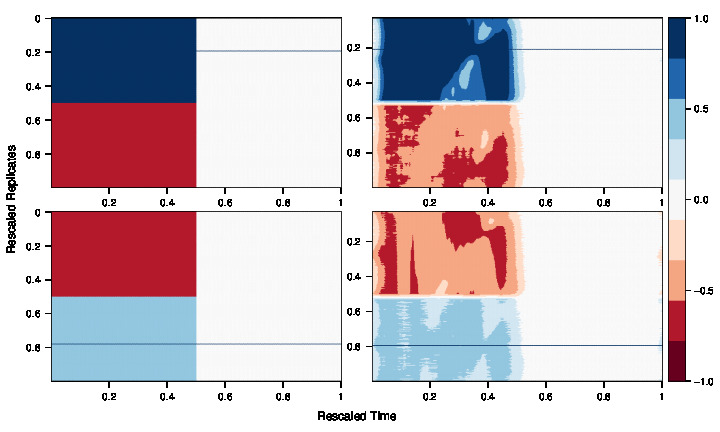}}
\end{center}
\vspace{-2em}
\caption{\small Coherence plots for replicates 50 (\textit{top row}) and 200 (\textit{bottom row}) over rescaled time and replicates in level $5$. \textit{Left}: true coherence; \textit{Right}: coherence estimates averaged over 100 simulations.}
\label{fig:coh_sim2}
\end{figure}

In terms of correctly estimating the coherence structure switch over times and replicates, as well as identifying the positive or negative character of the coherence, the proposed estimation procedure does a good job. We do however note that the estimated coherence intensity does exhibit some bias, which we may attribute to the smoothing performed in order to address the practical computation considerations (see the remark in Appendix~\ref{app:tables:simsext}). Nevertheless, we could argue that the model does give a good indication for the degree of the positiveness of the coherence, approximately $0.99$ and $0.5$ for replicates 50 and 200 respectively (right panels of Figure~\ref{fig:coh_sim2}).

Tables~\ref{tab:rhomsesim2} (Appendix~\ref{app:tables:simsext}) and~\ref{tab:rhomsesim1} (Section~\ref{supp:furthersimscoh} in the Supplementary Material) illustrate MSE results for two smoothing approaches: the first involves smoothing only over a window of replicates; the second involves local averaging through time and through replicates.

For both simulations, the results paint the same picture. As we increase the replicate smoothing window (such that $M/R \to 0$) the performance of our models improves in terms of MSEs, and the double smoothing over time and replicates further reduces the MSEs. The price to pay for double smoothing as usual is a slightly higher bias (than when using averaging over replicates only). In order to ensure that our spectral estimates are positive, our correction procedure uses the correction matrix $A^{-1}$ truncated at zero. Inevitably this introduces bias, evident through the increasing MSEs as $R$ and $T$ increase. In their work on bivariate coherence estimation, \cite{sand:2010} reported better results when additionally employing smoothing over scales. We conjecture that this is also applicable for our work, but leave the further numerical treatment for future research.

\section{Analysis of Macaque Local Field Potentials} \label{sec:rda}

We perform our analysis on the dataset of local field potentials (LFPs) from the hippocampus (Hc) and nucleus accumbens (NAc) of a macaque over the course of an associative learning experiment. Due to their roles in the consolidation of memory information and the processing of rewarding stimuli, the Hc and NAc have been studied in relation to learning tasks for monkeys, rats and humans \citep{wirth:2003, abela:2015, segcin:dynamics}. Our RLSW-based spectral analysis offers not only confirmations to the results of previous studies, but also provides additional insights to the understanding of the dynamics of the LFPs through capturing the evolutionary characteristics of brain processes {\em within} and {\em across} the trials of the experiment, in a {\em scale-dependent} manner, through the use of the wavelet transform. One notable benefit of our model in contrast to previous Fourier-based methodology \citep{fo:medbp}, is its superior time-localisation, as we will see next.

Each trial consists of $T = 2048$ time points, corresponding to approximately 2 seconds of data. The design splits each trial into four time blocks of 512 milliseconds each and it is in the final time block that the macaque was tasked with associating one of four doors (appearing on a screen) with the picture visual presented in the second time block. The macaque had to learn the associations through repeated trials and for each correct association made the macaque was given a juice reward. Further details are given in Section~\ref{supp:exper} of the Supplementary Material. The data has been grouped into sets of `correct' and `incorrect' responses, in order to investigate how the contributions of the Hc and NAc to the learning process differ between groups \citep{Gorrombao:monkeydepend}. The groups containing the correct and incorrect responses consist of 241 and 264 trials, respectively.

We opt to carry out the analysis on $R = 256$ replicates (trials). To obtain a dyadic number of replicates necessary for estimation (here, 256) for the correct response group, we mirror the last 15 trials. This is for computational purposes only, and we naturally discard the corresponding estimates from our discussions and plots. To ensure comparability across trials, each trial is standardised to have mean zero and unit variance. Plots for the correct responses appeared in Figures~\ref{fig:rb7_Ct_Hc} and~\ref{fig:rb7_Ct_NAc} (Section~\ref{sec:intro}) and for the incorrect responses in Figures~\ref{fig:rb7_It_Hc} and~\ref{fig:rb7_It_NAc} of the Supplementary Material Section~\ref{supp:exper}, which also details the implemented methodology.

\subsection{Results}

\textit{Hippocampus}.
The proposed spectral estimates for the correct and incorrect sets of trials appear in Figures~\ref{fig:coef_CtHc_M10D10} (below) and~\ref{fig:coef_ItHc_M10D10} (Section~\ref{supp:exper}, Supplementary Material), respectively. The bottom row plots show the RLSW($_2$)-spectral estimates averaged across 30 replicates to illustrate the process behaviour in the beginning, middle and end of the experiment. These demonstrate (i) the sequenced activation of within-trial time blocks and (ii) the evolutionary behaviour of the wavelet spectrum along the course of the experiment. For both correct and incorrect datasets, the `activity' is primarily captured within the coarser levels of the wavelet periodograms, approximately corresponding to frequencies 2-8Hz. Of these, the theta band frequencies 4-8Hz are typical of slow activity, known for their association to hippocampal activity in mammals and to promote memory \citep{buzs06}. \cite{fo:medbp} report the low frequency range 1-12Hz to account for most variability in
the Hc data. Our analysis offers a finer characterisation that does not fully support activation of low delta waves (under 2Hz), known to be typical of deep sleep, and shows weak alpha band (8-12Hz) and low beta band (12-16Hz) alertness at certain time blocks within each trial. Due to its construction, the RLSW model has been able to capture the process evolutionary behaviour across trials, which naturally cannot be achieved by LSW alone. The LSW-based estimates do capture some of the time-dependencies but remained constant across replicates. The RLSW estimates have the capacity to highlight the individual time blocks as they activate through the course of the experiment, an insight invisible to LSW and weakly represented in the Fourier approach of \cite{fo:medbp}.

\begin{figure}[h!]
\centering
\resizebox{!}{0.6\textwidth}{\includegraphics{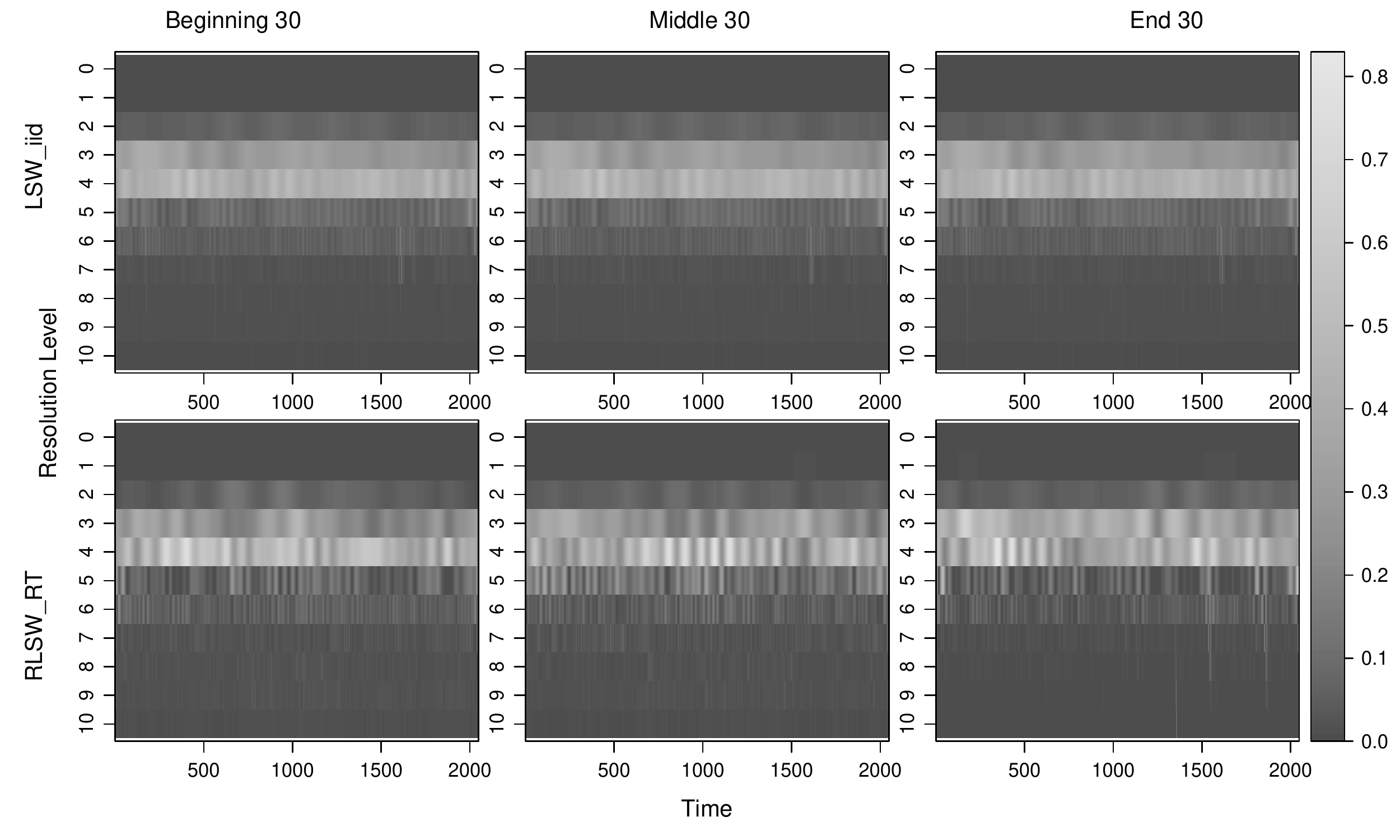}}
\caption{\small Time-scale hippocampus (Hc) plots for the correct trials. Spectral estimates are shown for the average over 30 replicates in the beginning, middle and end of the experiment. \textit{Top}: estimates from the LSW method averaged over all replicates; \textit{Bottom}: RLSW method with smoothing over time and replicates.}
\label{fig:coef_CtHc_M10D10}
\end{figure}

For the correct Hc trials, the RLSW models capture the bulk of `activity' in levels 3 and 4. In early replicates this is fairly even through time, while for middle replicates the bursts of activity shift centrally within time, thus coinciding with the second block of the macaque being shown the visual stimulus and possibly with the expectation of the picture to continue being shown. For the final trials, the activity is clearly localised around time-point 500 (corresponding to the visual exposure) and towards the final quarter of time (corresponding to the time when the macaque made the correct association). When compared to a Fourier approach, our wavelet-based analysis thus brings to the fore novel information that links the experimental time blocks to Hc activation. Specifically, as the correct trials progress, the activity in the Hc is evident at the visual cue time and also at the selection task time, thus suggesting learning of the picture associations.

Although we cannot compare the correct and incorrect trials like-for-like, we are still able to see evidence of evolutionary behaviour across the incorrect trials. As the experiment progresses, there is evidence of less spectral activity in the incorrect trials, with a brief Hc activation in the visual exposure block for the middling trials, and a burst of Hc activity localised in the last time block, when the task is carried out, for the end trials. The spectrum suggests that whereas the Hc displays prolonged activity in the second time block for the correct trials (corresponding to the picture being presented), this feature is not as marked in the incorrect trials and thus the macaque is not making the association between the picture presented and the selection task.
Scientific literature has shown \citep{segcin:dynamics} that during a learning experiment, activity in the Hc decreases as associations/rules are learned but would spike upon the application of an association/rule. The capacity of our model to extract time-localised information thus highlights novel traits that suggest that the macaque in this experiment has not yet fully learned the associations, but evidence of learning is indeed present.

Figure \ref{fig:Hc_avlevplots} further illustrates spectral evolution through time is captured by both the LSW and RLSW models. As the experiment progresses, the RLSW model identifies that the activity in the Hc increases across the correct trials with higher activity along the middle and end replicates, and (within-replicate) time dependent peaks gradually spanning the course of the experiment. The final trials display more Hc activity towards the task time than starting trials, again indicating a learning process. In contrast, the incorrect trials  display a much less structured behaviour, with time-dependent activity distributed more evenly across the replicates.

\begin{figure}[h!]
\centering
\resizebox{!}{0.6\textwidth}{\includegraphics{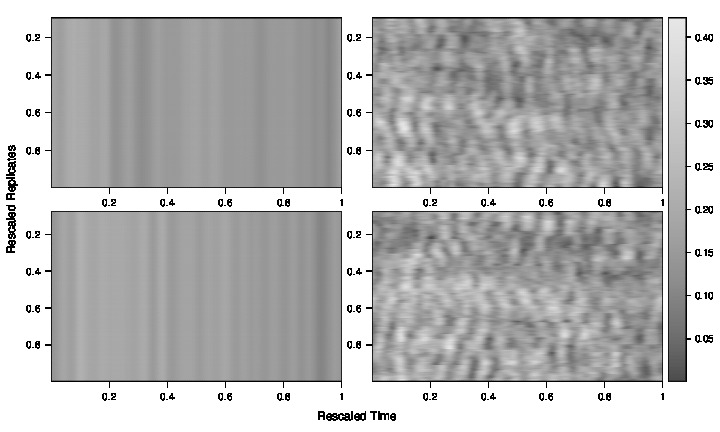}}
\caption{\small Time-replicate plots of the hippocampus (Hc) average spectra in levels 3 and 4 for the correct (\textit{Top}) and incorrect (\textit{Bottom}) trials. \textit{Left}: estimates from the LSW method averaged over all replicates; \textit{Right}: estimates from the RLSW method with smoothing over time and replicates.}
\label{fig:Hc_avlevplots}
\end{figure}

\textit{Nucleus accumbens}. The resulting (R)LSW spectral estimates appear in Figures \ref{fig:coef_CtNAc_M10D6} (below) and \ref{fig:coef_ItNAc_M10D6} (Section~\ref{supp:exper}, Supplementary Material) for the correct and incorrect sets of trials, respectively. These plots are to be understood in the same manner as those for the hippocampus. \cite{fo:medbp} find that the bulk of variability in the NAc is accounted for by (high) beta band frequencies (20-30Hz), while we place this in the wider range of beta band waves 16-30Hz, associated to focussed activity. Our analysis also offers evidence for low gamma frequency waves (31-60Hz), typical of working memory activation \citep{iacc:2016}. Additional to Fourier analysis, the RLSW model also shows that nonstationarity across time is clearly present, as well as some spectral evolution across the replicates. Although not as obvious as for the Hc, for the beginning and middle replicates of the correct group, NAc activity is manifest towards the trial start and end, while for the final replicates activity is captured in the final quarter of time. A similar pattern of behaviour is displayed by the incorrect group of trials. Also note that the NAc activity decreases in intensity from the beginning to end replicates for both groups.

\begin{figure}[h!]
\centering
\resizebox{!}{0.6\textwidth}{\includegraphics{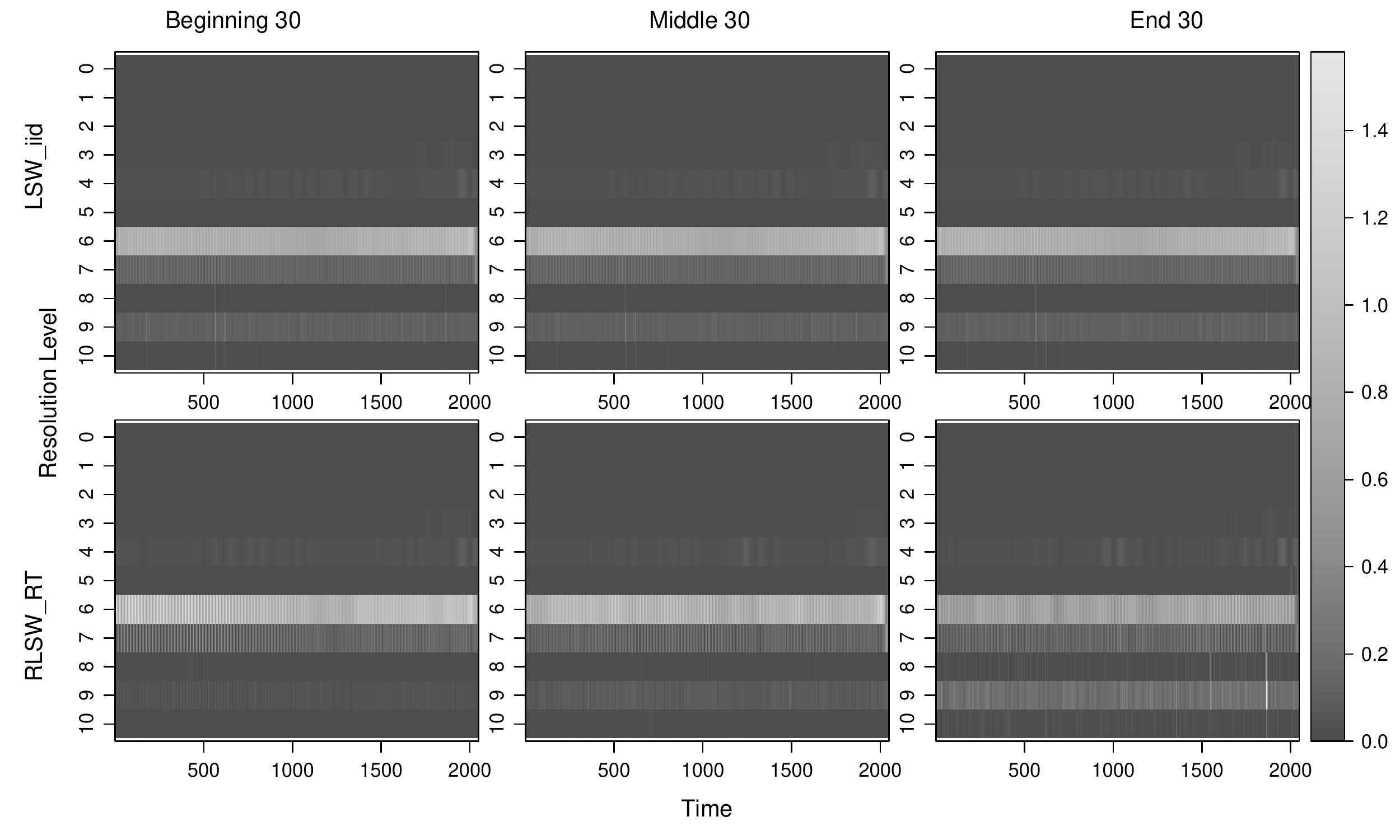}}
\caption{\small Time-scale nucleus accumbens (NAc) plots for the correct trials. Spectral estimates are shown for the average over 30 replicates in the beginning, middle and end of the experiment. \textit{Top}: estimates from the LSW method averaged over all replicates; \textit{Bottom}: RLSW method with smoothing over time and replicates.}
\label{fig:coef_CtNAc_M10D6}
\end{figure}

The NAc is part of the ventral striatum and plays a role in the processing of rewarding stimuli. The activity seen in the final 512 milliseconds can be attributed to the macaque expecting and receiving the juice reward in the correct trials, or expectation of reward in the incorrect trials. The impact of reward expectation \citep{schultz:rewardexp, Hollerman:1998, mulder:2005} could also explain the activity we see at the trial start for the beginning replicates and its observed periodicity across the experiment. Upon receiving no reward in an incorrect trial, the NAc activity decreases and with it the reward expectation falls for the next trial. The opposite holds for a correct trial when the reward is given. Our analysis reflects the results of other studies on learning experiments \citep{Hollerman:1998, fo:medbp} that highlight that the activity in the ventral striatum decreases as the stimuli are learned.

Plots for the average RLSW spectral estimates of the NAc in the finer levels 6 and 7 appear in Figure \ref{fig:NAc_avlevplots}. Our previous comments are also reflected in these plots. Evolution in the spectra across time is captured by both models, with the NAc activity displaying periodic patterns, while the RLSW-based estimation highlights a decrease in NAc activity along the trials of the experiment.

\begin{figure}[h!]
\centering
\resizebox{!}{0.6\textwidth}{\includegraphics{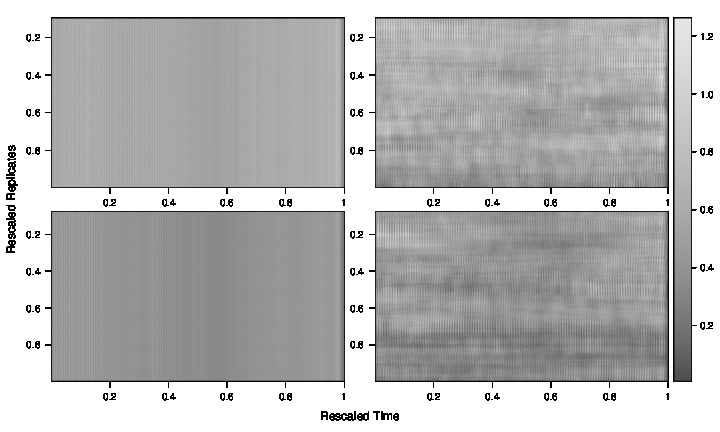}}
\caption{\small Time-replicate plots of the nucleus accumbens (NAc) average spectra in levels 6 and 7 for the correct (\textit{Top}) and incorrect (\textit{Bottom}) trials. \textit{Left}: estimates from the LSW method averaged over all replicates; \textit{Right}: estimates from the RLSW method with smoothing over time and replicates.}
\label{fig:NAc_avlevplots}
\end{figure}

\textit{Remarks}. Our analysis demonstrates how the simplifying assumption of trials that are identical realisations of the same process, leading one to draw conclusions solely based on averaging across all replicates, could cause an important understanding in the process evolution through the experiment to be missed. Our proposed RLSW methodology has captured the spectral {\em time-} and {\em replicate-}evolutionary behaviour, thus yielding new scale-based results and refining the findings of \cite{fo:medbp} in the Fourier domain. Nevertheless, these results are still underpinned by the assumption of uncorrelated trials. We next explore whether this assumption is tenable.

\subsection{Allowing for correlation between trials}
We relax the assumption of uncorrelated trials and allow for (potential) dependencies across trials. Hence we now estimate the locally stationary replicate-coherence, a quantity defined in~\eqref{eq:lsrc}.

The aim of our analysis is to investigate whether dependencies across trials exist.
For the Hc data, we were unable to capture any substantial coherence structure over the trials (correct and incorrect groups). However, our analysis did find evidence of a moderate dependence across neighbouring trials in the NAc at the beta band frequencies (16-30Hz), known to be responsible for brain activity related to reward feedback mechanisms.
The estimated NAc replicate-coherence (absolute value) is shown in Figure~\ref{fig:NAc_cohplots} at level 6 for trials 20, 100 and 200 in the incorrect group and for correct trial 200, depicting typical behaviour. Some burst areas are present, indicating a moderate neighbouring replicate coherence, with most meaningful values either side of 0.4 and a few above 0.5. For the beginning and middling incorrect trials, this is apparent in the time periods leading up to and inclusive of the trial task phase, upon which the macaque would receive a juice reward if the task was done correctly.

\begin{figure}[htbp!]
\centering
\resizebox{!}{0.6\textwidth}{\includegraphics{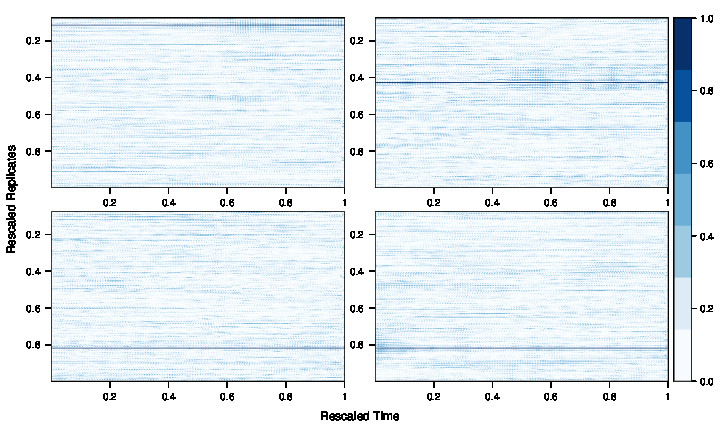}}
\caption{\small Level 6 absolute value coherence for NAc incorrect trials: 20 (\textit{top left}), 100 (\textit{top right}) and 200 (\textit{bottom left}); and for NAc correct trial 200 (\textit{bottom right}).}
\label{fig:NAc_cohplots}
\end{figure}

\textit{Remarks}. Our analysis provides novel evidence in the temporal and scale (frequency)-dimensions that mild to moderate dependence is exhibited in both the correct and incorrect trials. This is primarily evident in the final correct trials, potentially as the manifest result of learning, and at the onset of the incorrect trials as the likely result to the expectation of reward.
Within the neuroscience literature (e.g.  \cite{Gorrombao:monkeydepend}), dependence between brain regions is also of interest, with coherence measures setup between channels of interest (here, Hc and NAc). Such measures have not formed the scope of our work here, but in relating our replicate-coherence results to the reported evolutionary coherence between the Hc and NAc, the dependence we observed at the beginning of NAc correct trial 200 (approximately rescaled replicate 0.8) is reminiscent of the dependence between Hc and NAc captured at rescaled replicate-time 0.8 by \cite{fo:medbp}.

\section{Concluding remarks}\label{sec:concl}
Our work proposed a novel wavelet-based methodology that successfully captures nonstationary process characteristics across a replicated time series, often encountered in experimental studies. Its desirable properties were evidenced by simulation studies and through a real data application from the neurosciences. However, the methodology itself is not restricted to use within this field, and the authors envisage its utility in other experimental areas where wavelet spectral analysis has proved to be ideally suited, e.g. circadian biology \citep{hargreaves19:wavelet}. This work demonstrates the dangers of approaching replicated time series as identical process realisations and the misleading results this can yield when studying the process dynamics along the course of the experiment.  Our work proposes statistical models and associated estimation theory for processes that embed not only uncorrelated replicates, but also allow for dependence between replicates. A next natural step would be to extend the RLSW methodology to a multivariate setting (\cite{sand:2010,park:mvlsw}) and to additionally define and investigate the variate replicate coherence.

\addtocounter{section}{0}
\appendixtitleon

\begin{appendices}
\section{Supporting evidence for simulation studies}\label{app:tables}
\subsection{MSE and squared bias tables for Section~\ref{sec:sims}}\label{app:tables:sims}
As means to quantify the performance of the models, we employ the {mean squared error} (MSE) and {squared bias}, calculated as the average over all time-scale points and replicates as follows
\begin{align*}
MSE(\hat{S}) &=
{(R  \, J \, T)}^{-1} \sum_{r, j, k}\left[ \frac{1}{N}\sum_{n=1}^{N} \left( \hat{S}_{j}^{(n)}\left(\frac{k}{T},\frac{r}{R}\right) - S_{j}\left(\frac{k}{T},\frac{r}{R}\right) \right)^{2}\right],\\
Bias^2(\hat{S}) &=
{(R \, J \, T)}^{-1} \sum_{r, j, k}\left[ \frac{1}{N}\sum_{n=1}^{N} \hat{S}_{j}^{(n)}\left(\frac{k}{T},\frac{r}{R}\right) - S_{j}\left(\frac{k}{T},\frac{r}{R}\right) \right]^{2}.
\end{align*}

\noindent{\em Remarks on estimates at the boundaries}. We do not assess edges that involve local averaging over the first and last $(M-1)$ replicates. This has also been accounted for when calculating the MSE and squared bias. As a result, the reported measures whose values correspond to modelling via LSW will appear to change (in a very minor way), when in fact they should be the same for all choices of $M$ at fixed $R$ and $T$.

\begin{table}[htbp!]
\scriptsize
\hspace{-6em}
\begin{tabular}{@{\extracolsep{3pt}}ccccccccc@{}}
\multicolumn{1}{c}{} & \multicolumn{7}{c}{Mean squared errors $\left( \times 100 \right)$} & \multicolumn{1}{c}{}
\\
\cline{1-9}\noalign{\smallskip}
\multicolumn{3}{c}{} & \multicolumn{2}{c}{LSW} & \multicolumn{2}{c}{RLSW\tsbs{1}} & \multicolumn{2}{c}{RLSW\tsbs{2}}\\
\cline{4-5}\cline{6-7}\cline{8-9}\noalign{\smallskip}
R & T & M & mse & bias\tsps{2} & mse & bias\tsps{2} & mse & bias\tsps{2} \\
\firsthline\noalign{\smallskip}
256 & 128 & 4 & 17.55 & 17.27 & 22.31 & 2.53 & 11.52 & 3.88
\smallskip \\
 & & 12 & 16.74 & 16.46 & 9.45 & 2.40 & 6.54 & 3.82
\smallskip \\
 & 256 & 4 & 14.21 & 13.95 & 19.62 & 1.04 & 8.46 & 1.29
\smallskip \\
 & & 12 & 13.48 & 13.22 & 7.58 & 0.92 & 3.83 & 1.23
\smallskip \\
 & 512 & 4 & 12.25 & 12.01 & 17.59 & 0.53 & 7.15 & 0.52
\smallskip \\
 & & 12 & 11.59 & 11.35 & 6.51 & 0.42 & 2.85 & 0.48
\smallskip \\
\lasthline
\end{tabular}
\quad
\begin{tabular}{@{\extracolsep{3pt}}ccccccccc@{}}
\multicolumn{1}{c}{} & \multicolumn{7}{c}{Mean squared errors $\left( \times 100 \right)$} & \multicolumn{1}{c}{}
\\
\cline{1-9}\noalign{\smallskip}
\multicolumn{3}{c}{} & \multicolumn{2}{c}{LSW} & \multicolumn{2}{c}{RLSW\tsbs{1}} & \multicolumn{2}{c}{RLSW\tsbs{2}}\\
\cline{4-5}\cline{6-7}\cline{8-9}\noalign{\smallskip}
R & T & M & mse & bias\tsps{2} & mse & bias\tsps{2} & mse & bias\tsps{2} \\
\firsthline\noalign{\smallskip}
512 & 256 & 4 & 14.29 & 14.14 & 19.68 & 1.05 & 8.49 & 1.28
\smallskip \\
 & & 12 & 13.89 & 13.76 & 7.62 & 0.92 & 3.81 & 1.22
\smallskip \\
 & 512 & 4 & 12.30 & 12.18 & 17.61 & 0.53 & 7.15 & 0.52
\smallskip \\
 & & 12 & 11.96 & 11.84 & 6.54 & 0.42 & 2.85 & 0.48
\smallskip \\
 & 1024 & 4 & 10.92 & 10.81 & 15.83 & 0.32 & 6.29 & 0.25
\smallskip \\
 & & 12 & 10.61 & 10.50 & 5.78 & 0.22 & 2.38 & 0.21
\smallskip \\
\lasthline
\end{tabular}
\caption{\small MSE and squared bias results, averaged over all time-scale points and replicates for the simulation in Section~\ref{sec:sims}. `LSW' denotes the classical approach of averaging over the replicates. `RLSW$_1$' denotes our proposed approach using localised replicate smoothing. `RLSW$_2$' denotes our proposed approach using localised time and replicate smoothing. Our proposed methods `RLSW$_1$'and `RLSW$_2$' use a replicate smoothing window of length $(2M+1)$, while the time-smoothing window for `LSW' and `RLSW$_2$' is automatically chosen.}
\label{tab:msesim3}
\end{table}

\subsection{MSE and squared bias tables for Section~\ref{sec:simsext}}\label{app:tables:simsext}

As numerical tools that quantify the performance of an estimate $\hat{\rho}_{j}(z,\nu)$ across all time points $z=k/T$ with $k=0, \ldots, T-1$, replicates $\nu=r/R$ with $r=0, \ldots, R-1$ and scales $j=1, \ldots, J$, we will use the mean squared error and squared bias, defined in this context as

\begin{equation}\nonumber
MSE(\hat{\rho})=
( R^{*} \, J \, T )^{-1} \sum_{r, r^{*}, j, k}\left[ \frac{1}{N}\sum_{n=1}^{N} \left( \hat{\rho}_{j}^{(n)}\left(\frac{k}{T},\frac{r}{R},\frac{r^{*}}{R}\right) - \rho_{j}\left(\frac{k}{T},\frac{r}{R},\frac{r^{*}}{R}\right) \right)^{2}\right],
\end{equation}

\begin{equation}\nonumber
Bias^2(\hat{\rho}) =
(R^{*} \, J \, T)^{-1} \sum_{r, r^{*}, j, k}\left( \frac{1}{N}\sum_{n=1}^{N} \hat{\rho}_{j}^{(n)}\left(\frac{k}{T},\frac{r}{R},\frac{r^{*}}{R}\right) - \rho_{j}\left(\frac{k}{T},\frac{r}{R},\frac{r^{*}}{R}\right) \right)^{2},
\end{equation}
where due to the symmetry of the coherence matrix, we have used $R^{*} = R(R+1)/2$ and $r^{*} = r' \geq r$, and $N$ denotes the number of simulation runs. As in Section~\ref{sec:sims}, we also adopt here $J=J(T)$.

\noindent{\em Remarks on implementation}. From the theoretical model construction, the autospectra are positive quantities. However, our spectral estimates may take values that are negative or close to zero after correction, and this in turn can cause problems when normalising for coherence estimation. In order to bypass this issue, we choose to correct our raw wavelet periodogram estimates before smoothing. The theoretical properties of the coherence estimator show that using replicate-smoothing does yield an estimator with good properties, albeit its rate of convergence is heavily dependent on the smoothing window width $(2M+1)$. A local averaging window over time for smoothing each replicate before applying smoothing over replicates could also be employed, just as proposed for spectral estimation in Section~\ref{sec:specestRLSW}. A possible further step, to be done after smoothing through time, is to smooth over scales as proposed by \cite{sand:2010} for estimating the linear dependence between bivariate LSW time series. However, we do not pursue this approach here.

\begin{table}[htbp!]
\small
\hspace{1em}
\centering
\begin{tabular}{cccccccc}
\cline{1-8}\noalign{\smallskip}
R & T & M & & mse\tsbs{1} & bias\tsps{2}\tsbs{1} & mse\tsbs{2} & bias\tsps{2}\tsbs{2} \\
\firsthline\noalign{\smallskip}
128 & 256 & 7 & & 18.10 & 9.02 & 15.50 & 10.71
\smallskip \\
 & & 12 & & 14.53 & 8.22 & 13.05 & 9.64
\smallskip \\
128 & 512 & 7 & & 18.91 & 10.56 & 16.72 & 12.37
\smallskip \\
 &  & 12 & & 15.10 & 9.40 & 13.95 & 10.92
\smallskip \\
256 & 512 & 7 & & 20.82 & 11.58 & 18.40 & 13.60
\smallskip \\
 & & 12 & & 17.78 & 10.93 & 16.43 & 12.81
\smallskip \\
\lasthline
\end{tabular}
\caption{\small MSE and squared bias $\left( \times 100 \right)$, averaged over all time-scale point and replicates for the simulation in Section~\ref{sec:simsext}. Subscripts 1 and 2 denote the models with smoothing over replicates only and time-replicate smoothing, respectively.}\label{tab:rhomsesim2}
\end{table}

\section{Proofs of results on the asymptotic behaviour of proposed estimators under the uncorrelated replicates assumption}\label{app:proofs}

In this section, we give details of the proofs in Section~\ref{sec:model} using the notation described therein.

For ease we recall here that the auto- and cross-correlation wavelets are defined  for $\uptau\in\Z$ as $\Psi_j(\uptau)=\sum_{k\in \Z}\psi_{j,k}(0)\psi_{j,k}(\uptau)$ and $\Psi_{j,l}(\uptau)=\sum_{k\in \Z}\psi_{j,k}(0)\psi_{l,k}(\uptau)$ respectively, where in general $\psi_{j,k}(\uptau)=\psi_{j,k-\uptau}$ are compactly supported discrete wavelets as defined in \cite{nvsk:2000}. Note that from their construction, $\Psi_j(-\uptau)=\Psi_j(\uptau)$ and $\Psi_{j,l}(-\uptau)=\Psi_{l,j}(\uptau)$, and both have compact support of order $2^j$, $2^j+2^l$ respectively \citep{nvsk:2000, sand:2010}.

In what follows, wherever the summation domain is not specified, it is to be understood as $\Z$ for time indices (e.g. $k, \, n, \, \uptau$) and as $\Z^\star_+$ (strictly positive integers) for scale indices (e.g. $j, \, j', \, l$).

We also recall the autocorrelation wavelet inner product matrices, defined as
\beqann
A_{j,l}&=&\sum_{\uptau} \Psi_j(\uptau)\Psi_l(\uptau)=\sum_{\uptau} \left| \Psi_{j,l}(\uptau)\right|^2,\\
A_{j,l}^{\uptau}&=&\sum_{n} \Psi_{j}(n) \Psi_{l}(n+\uptau)=\sum_{n}\Psi_{j,l}(n)\Psi_{j,l}(n+\uptau) ,
\eeqann
with properties $A_{j,l}\geq 0$, $A_{j,l}=A_{l,j}$, $\sum_j 2^{-j} A_{j,l}=1$ \citep{fryz03:forecasting}.

In the proofs that follow, we make use of the results in the following lemmas, whose proofs appear in Section~\ref{supp:proofs:sec3} of the Supplementary Material.

\begin{lemma}\label{lem:ASord}
Under the assumptions of Definition~\ref{def:rlsw}, we have for every rescaled time and replicate, $z, \, \nu$ respectively, and scale $j$,
 $\sum_{l=1}^{\infty}A_{j,l} S_{l}\left(z,\nu\right)=\ORD(2^{j})$.
\end{lemma}

\begin{lemma}\label{lem:ASordtau}
Under the assumptions of Definition~\ref{def:rlsw}, we have for every rescaled time and replicate, $z, \, \nu$ respectively, lag $\uptau$ and scale $j$,
$\sum_{l=1}^{\infty}A_{l,j}^{\uptau} S_{l}\left(z,\nu\right)=\ORD(2^{j})$.
\end{lemma}

\begin{lemma}\label{lem:ASordjj}
Under the assumptions of Definition~\ref{def:rlsw}, we have  for every rescaled time and replicate, $z, \, \nu$ respectively, lag $\uptau$ and scales $j, j'$,
$$\sum_{l=1}^{\infty} \left|\sum _{n\in \Z} \Psi_{j,j'}(n) \Psi_{l}(n+\uptau) S_{l}\left(z,\nu\right)\right|=\ORD(2^{(j+j')/2}).$$
\end{lemma}

\subsection{Proof of Proposition~\ref{prop:EI} }\label{app:proofs:rawi}

\subsubsection*{Proof of Proposition \ref{prop:EI} (Expectation)}
\begin{proof}[\unskip\nopunct]
For a fuller understanding of the process behaviour over replicates, we start by defining the cross-replicate-periodogram as $I_{j,k;T}^{(r,r');R} = d_{j,k;T}^{ r;R} d_{j,k;T}^{ r';R}$. Taking the expectation, it follows that
\begin{align} 
\E\left[I_{j,k;T}^{(r,r');R}\right] &= \E\left[d_{j,k;T}^{ r;R} d_{j,k;T}^{ r';R}\right] \nonumber \\
&= \E\left[\left\{ \sum_{t} X_{t;T}^{r;R}\psi_{j,k}(t) \right\} \left\{ \sum_{t'} X_{t';T}^{r;R}\psi_{j,k}(t') \right\}\right] \nonumber \\
&= \E\left[\left\{ \sum_{t} \sum_{l=1}^{\infty}\sum_{m\in\mathbb{Z}} \omega_{l,m;T}^{r;R}\psi_{l,m}(t)\xi_{l,m}^{r}\psi_{j,k}(t) \right\}\right. \nonumber \\ & \quad\qquad\qquad \times \left.\left\{ \sum_{t'} \sum_{l'=1}^{\infty}\sum_{m'\in\mathbb{Z}} \omega_{l',m';T}^{r;R}\psi_{l',m'}(t')\xi_{l',m'}^{r}\psi_{j,k}(t') \right\}\right] \nonumber \nonumber \\
&= \sum_{t}\sum_{t'} \sum_{l=1}^{\infty}\sum_{l'=1}^{\infty}\sum_{m\in\mathbb{Z}}\sum_{m'\in\mathbb{Z}}\omega_{l,m;T}^{r;R}\omega_{l',m';T}^{r';R}\psi_{l,m}(t)\psi_{l',m'}(t')\psi_{j,k}(t)\psi_{j,k}(t')\E\left[\xi_{l,m}^{r}\xi_{l',m'}^{r'}\right] \nonumber \\
&= \sum_{l=1}^{\infty}\sum_{m\in\mathbb{Z}} (\omega_{l,m;T}^{r;R})^{2} \sum_{t}\psi_{l,m}(t)\psi_{j,k}(t)\sum_{t'}\psi_{l,m}(t')\psi_{j,k}(t') \nonumber \\
&= \E\left[I_{j,k;T}^{r;R}\right],\nonumber
\end{align}
since for the orthogonal increments, $\E\left[\xi_{l,m}^{r}\xi_{l',m'}^{r'}\right] = \delta_{l,l'}\delta_{m,m'}\delta_{r,r'}$ and so $l=l'$, $m=m'$ and $r=r'$.

Using the properties in equation~\eqref{eq:approx} and the definition of the cross-correlation wavelet functions ($\Psi_{j,l}$), we obtain by replacing $n: =m-k$

\beqann
\E\left[I_{j,k;T}^{r;R}\right]&=& \sum_{l=1}^{\infty}\sum_{n\in\mathbb{Z}} (\omega_{l,n+k;T}^{r;R})^{2} \sum_{t}\psi_{l,n+k}(t)\psi_{j,k}(t)\sum_{t'}\psi_{l,n+k}(t')\psi_{j,k}(t') ,\\
&=& \sum_{l=1}^{\infty}\sum_{n\in\mathbb{Z}} \left(S_{l}\left(\frac{n+k}{T},\frac{r}{R}\right) + \ORD(D_{l}R^{-1})+\ORD(C_{l}T^{-1})\right)
\left\{ \sum_{t}\psi_{l,n+k}(t)\psi_{j,k}(t) \right\}^2, \\
&=& \sum_{l=1}^{\infty}\sum_{n\in\mathbb{Z}} \left(S_{l}\left(\frac{k}{T},\frac{r}{R}\right) + \ORD(L_{l} |n| T^{-1})+\ORD(D_{l}R^{-1})+\ORD(C_{l}T^{-1})\right) \left(\Psi_{j,l}(n)\right)^2,
\eeqann
where in the last equality we used the Lipschitz continuity of the spectrum in the (rescaled) time argument and a substitution $n:= -n$.

Let us now bound the order terms, as follows.

Since the number of terms in the wavelet cross-correlation is finite and bounded as a function of $n$ \citep{nvsk:2000} and the compact support of $\Psi_{j,l}^2$ is bounded by $K(2^{j}+2^{l})$ for some constant $K$ \citep{sand:2010}, we have
\beqann
\sum_{l}\sum_{n} L_{l} |n| \ORD(T^{-1})\Psi_{j,l}^2(n)
&\leq& \ORD(T^{-1}) \sum_{l} (2^{j}+2^{l}) L_l A_{j,l},\\
&\leq& \ORD(T^{-1} 2^j) \sum_{l} L_l A_{j,l} + \ORD(T^{-1}) \sum_{l} 2^{l}L_l A_{j,l} , \\
&=& \ORD(2^j T^{-1}),
\eeqann
where we used
$\sum_{l} L_l A_{j,l}=\sum_{l} 2^l L_l 2^{-l}A_{j,l}\leq \sum_{l^\prime} 2^{l^\prime} L_{l^\prime} \sum_{l} 2^{-l} A_{j,l} < \infty$ since $\sum_{l} 2^{l}L_l<\infty$ and $\sum_{l} 2^{-l} A_{j,l}=1$ \citep{fryz03:forecasting}, and
\beqann
T^{-1}\sum_{l} 2^l L_l A_{j,l}&=&2^j T^{-1}
\sum_{l} 2^l L_l 2^{-j}A_{j,l},\\
&\leq& 2^j T^{-1}  \sum_{l} \left( 2^{l} L_{l} \sum_{j^\prime} 2^{-j^\prime} A_{j^\prime,l} \right),\\
&=& \ORD(2^j T^{-1}),
\eeqann
again as $\sum_{l} 2^{l}L_l<\infty$ and $\sum_{j} 2^{-j} A_{j,l}=1$.

Using the definition of the $A$ matrix, $A_{j,l}=\sum_{n} \Psi_{j,l}^2(n)$, in the next order term, we obtain
\beqann
\sum_{l}\sum_{n} D_{l} \ORD(R^{-1})\Psi_{j,l}^2(n)
&=& \ORD (R^{-1}) \sum_{l} 2^l D_l 2^{-l} A_{j,l}, \\
&\leq& \ORD(R^{-1}) \sum_{l^\prime} 2^{l^\prime}D_{l^\prime} \sum_{l} 2^{-l}A_{j,l},\\
&=& \ORD(R^{-1})
\eeqann
where we used the condition $\sum_{l} 2^l D_l<\infty$ and $\sum_{l} 2^{-l} A_{j,l}=1$ \citep{fryz03:forecasting}.
Using the same set of arguments and the condition $\sum_{l} 2^l C_l<\infty$, we also have $\sum_{l}\sum_{n} C_{l} \ORD(T^{-1})\Psi_{j,l}^2(n)=\ORD(T^{-1})$.

Hence, retaining the maximum order terms, we obtain

\beqann
\E\left[I_{j,k;T}^{r;R}\right]
&=& \sum_{l=1}^{\infty}\sum_{n\in\mathbb{Z}} S_{l}\left(\frac{k}{T},\frac{r}{R}\right) \left(\Psi_{j,l}(n)\right)^2 + \ORD(2^j T^{-1})+ \ORD(R^{-1}),\\
&=& \sum_{l=1}^{\infty} A_{j,l}S_{l}\left(\frac{k}{T},\frac{r}{R}\right) + \ORD(2^j T^{-1})+  \ORD(R^{-1}).
\eeqann
\end{proof}

\subsubsection*{Proof of Proposition \ref{prop:EI} (Variance)}
\begin{proof}[\unskip\nopunct]
For ease of notation let $X_{t;T}^{r;R} = X_{t}^{r}$ and $d_{j,k;T}^{ r;R}=d_{j,k}^{ r}$.
\begin{align}\label{eq:EII1}
\E\left[I_{j,k;T}^{r;R}I_{j',k';T}^{r';R}\right] &= \E\left[d_{j,k}^{ r} d_{j,k}^{ r}d_{j',k'}^{ r'} d_{j',k'}^{ r'}\right] \nonumber \\
&= \E\left[ \sum_{t} X_{t}^{r}\psi_{j,k}(t) \sum_{t'} X_{t'}^{r}\psi_{j,k}(t')\right. \nonumber \\
& \qquad\qquad\qquad\qquad \times \left.\sum_{h} X_{h}^{r'}\psi_{j',k'}(h)\sum_{h'} X_{h'}^{r'}\psi_{j',k'}(h')\right] \nonumber \\
&= \sum_{t}\sum_{t'}\sum_{h}\sum_{h'}\psi_{j,k}(t)\psi_{j,k}(t')\psi_{j',k'}(h)\psi_{j',k'}(h')\E\left[ X_{t;T}^{r;R}X_{t';T}^{r;R} X_{h;T}^{r';R}X_{h';T}^{r';R}\right].
\end{align}
Using the result of \cite{isserlis:1918} for \textcolor{black}{zero-mean Gaussian random variables}, we can write that
\begin{equation} \label{eq:iss}
\E\left[ X_{t}^{r}X_{t'}^{r} X_{h}^{r'}X_{h'}^{r'}\right] = \E\Big[ X_{t}^{r}X_{t'}^{r}\Big]\E\left[ X_{h}^{r'}X_{h'}^{r'}\right] + \E\left[ X_{t}^{r}X_{h}^{r'}\right]\E\left[ X_{t'}^{r}X_{h'}^{r'}\right] + \E\left[ X_{t}^{r}X_{h'}^{r'}\right]\E\left[ X_{t}^{r}X_{h}^{r'}\right],
\end{equation}
and substituting \eqref{eq:iss} into \eqref{eq:EII1} re-write
\beq\label{eq:EII}
\E\left[I_{j,k;T}^{r;R}I_{j',k';T}^{r';R}\right] = \alpha + \beta + \gamma.
\eeq

Let us now consider
\begin{align} \label{eq:EXX}
\E\Big[ X_{t}^{r}X_{t'}^{r}\Big] &= \E\left[\left\{ \sum_{l=1}^{\infty}\sum_{m\in\mathbb{Z}} \omega_{l,m;T}^{r;R}\psi_{l,m}(t)\xi_{l,m}^{r} \right\}\left\{ \sum_{l'=1}^{\infty}\sum_{m'\in\mathbb{Z}} \omega_{l',m';T}^{r;R}\psi_{l',m'}(t')\xi_{l',m'}^{r}\right\}\right] \nonumber \\
&= \sum_{l=1}^{\infty}\sum_{m\in\mathbb{Z}} \omega_{l,m;T}^{r;R}\psi_{l,m}(t)\sum_{l'=1}^{\infty}\sum_{m'\in\mathbb{Z}} \omega_{l',m';T}^{r;R}\psi_{l',m'}(t')\E\left[\xi_{l,m}^{r}\xi_{l',m'}^{r}\right] \nonumber\\
&= \sum_{l=1}^{\infty}\sum_{m\in\mathbb{Z}} (\omega_{l,m;T}^{r;R})^{2}\psi_{l,m}(t)\psi_{l,m}(t'),
\end{align}
since for the orthogonal increments, $\E\left[\xi_{l,m}^{r}\xi_{l',m'}^{r}\right] = \delta_{l,l'}\delta_{m,m'}$ and so \eqref{eq:EXX} is non-zero only when $l=l'$ and $m=m'$. Then for the first term in equation \eqref{eq:iss}, we have
\begin{equation} \label{eq:EEXX}
\E\Big[ X_{t}^{r}X_{t'}^{r}\Big]\E\left[ X_{h}^{r'}X_{h'}^{r'}\right] = \sum_{l=1}^{\infty}\sum_{m\in\mathbb{Z}} (\omega_{l,m;T}^{r;R})^{2}\psi_{l,m}(t)\psi_{l,m}(t')\sum_{e=1}^{\infty}\sum_{f\in\mathbb{Z}} (\omega_{e,f;T}^{r';R})^{2}\psi_{e,f}(h)\psi_{e,f}(h').
\end{equation}
Thus the first term in equation \eqref{eq:EII}, $\alpha$, is
\begin{align*}
\alpha &= \sum_{t}\sum_{t'}\sum_{h}\sum_{h'}\psi_{j,k}(t)\psi_{j,k}(t')\psi_{j',k'}(h)\psi_{j',k'}(h')\E\Big[ X_{t}^{r}X_{t'}^{r}\Big]\E\left[ X_{h}^{r'}X_{h'}^{r'}\right] \\
&= \sum_{l=1}^{\infty}\sum_{m\in\mathbb{Z}} (\omega_{l,m;T}^{r;R})^{2}\sum_{t}\psi_{l,m}(t)\psi_{j,k}(t)\sum_{t'}\psi_{l,m}(t')\psi_{j,k}(t') \\
& \qquad\qquad \times \sum_{e=1}^{\infty}\sum_{f\in\mathbb{Z}} (\omega_{e,f;T}^{r';R})^{2}\sum_{h}\psi_{e,f}(h)\psi_{j',k'}(h)\sum_{h'}\psi_{e,f}(h')\psi_{j',k'}(h') \\
&= \left(\sum_{l=1}^{\infty}A_{j,l} S_{l}\left(\frac{k}{T},\frac{r}{R}\right) + \ORD(2^{j}T^{-1}) + \ORD(R^{-1})\right) \\
& \qquad\qquad \times \left(\sum_{e=1}^{\infty}A_{j',e} S_{e}\left(\frac{k'}{T},\frac{r'}{R}\right) + \ORD(2^{j'}T^{-1}) + \ORD(R^{-1})\right),
\end{align*}
where we used the asymptotic expectation result obtained above (see equation~\eqref{eq:EI}).

Using the result of  Lemma~\ref{lem:ASord}, we obtain
\begin{align*}
\alpha&=\left( \sum_{l=1}^{\infty}A_{j,l} S_{l}\left(\frac{k}{T},\frac{r}{R}\right)\right) \left(\sum_{e=1}^{\infty}A_{j',e} S_{e}\left(\frac{k'}{T},\frac{r'}{R}\right)\right)+ \ORD(2^{j+j'}T^{-1})+\ORD(2^{\mbox{max}\{j,j'\}}R^{-1}),
\end{align*}
where we used $\ORD(2^{\mbox{max}\{j,j'\}}R^{-1})=\ORD(2^{j}R^{-1})+\ORD(2^{j'}R^{-1})$.

In the same manner as for \eqref{eq:EXX} and \eqref{eq:EEXX}, we can show that
\begin{align} \label{eq:EEXX2}
\E\left[ X_{t}^{r}X_{h}^{r'}\right]\E\left[ X_{t'}^{r}X_{h'}^{r'}\right] &= \sum_{l=1}^{\infty}\sum_{m\in\mathbb{Z}} (\omega_{l,m;T}^{r;R})^{2}\psi_{l,m}(t)\psi_{l,m}(h)\sum_{l'=1}^{\infty}\sum_{m'\in\mathbb{Z}} (\omega_{l',m';T}^{r;R})^{2}\psi_{l',m'}(t')\psi_{l',m'}(h'), \nonumber \\
\E\left[ X_{t}^{r}X_{h'}^{r'}\right]\E\left[ X_{t'}^{r}X_{h}^{r'}\right] &= \sum_{l=1}^{\infty}\sum_{m\in\mathbb{Z}} (\omega_{l,m;T}^{r;R})^{2}\psi_{l,m}(t)\psi_{l,m}(h')\sum_{l'=1}^{\infty}\sum_{m'\in\mathbb{Z}} (\omega_{l',m';T}^{r;R})^{2}\psi_{l',m'}(t')\psi_{l',m'}(h),
\end{align}
noting that the above expressions are non-zero only when $r=r'$ since $\E\left[\xi_{j,k}^{r}\xi_{j',k'}^{r'}\right] = \delta_{j,j'}\delta_{k,k'}\delta_{r,r'}$. Then the second term in \eqref{eq:EII}, $\beta$, is
\begin{align*}
\beta &= \sum_{t}\sum_{t'}\sum_{h}\sum_{h'}\psi_{j,k}(t)\psi_{j,k}(t')\psi_{j',k'}(h)\psi_{j',k'}(h')\E\left[ X_{t}^{r}X_{h}^{r'}\right]\E\left[ X_{t'}^{r}X_{h'}^{r'}\right] \\
&= \sum_{l=1}^{\infty}\sum_{m\in\mathbb{Z}} (\omega_{l,m;T}^{r;R})^{2}\sum_{t}\psi_{l,m}(t)\psi_{j,k}(t)\sum_{h}\psi_{l,m}(h)\psi_{j',k'}(h)\\
& \qquad\qquad\qquad\qquad\qquad \times \sum_{l'=1}^{\infty}\sum_{m'\in\mathbb{Z}}(\omega_{l',m';T}^{r;R})^{2}\sum_{t'}\psi_{l',m'}(t')\psi_{j,k}(t') \sum_{h'}\psi_{l',m'}(h')\psi_{j',k'}(h')\\
&= \left(\sum_{l=1}^{\infty}\sum_{m\in\mathbb{Z}} (\omega_{l,m;T}^{r;R})^{2}\sum_{t}\psi_{l,m}(t)\psi_{j,k}(t)\sum_{h}\psi_{l,m}(h)\psi_{j',k'}(h)\right)^{2}.
\end{align*}
Making the substitution $m=n+k$, and noticing from the above that the term $\beta \geq 0$, it follows that
\begin{align*}
\beta^{\frac{1}{2}} &= \left| \sum_{l=1}^{\infty}\sum_{n\in\mathbb{Z}} (\omega_{l,n+k;T}^{r;R})^{2}
\sum_{t}\psi_{l,n+k-t}\psi_{j,k-t}\sum_{h}\psi_{l,n+k-h}\psi_{j',k'-h} \right|,\\
&= \left| \sum_{l=1}^{\infty}\sum_{n\in\mathbb{Z}} \left\{S_{l}\left(\frac{n+k}{T},\frac{r}{R}\right)+\ORD(D_{l}R^{-1})+ \ORD(C_{l}T^{-1}) \right\} \sum_{t}\psi_{l,n-t}\psi_{j,-t}\sum_{h}\psi_{l,n+k-h}\psi_{j',k'-h}\right|\\
&= \left| \sum_{l=1}^{\infty}\sum_{n\in\mathbb{Z}} \left\{S_{l}\left(\frac{k}{T},\frac{r}{R}\right)+\ORD(D_{l}R^{-1})+ \ORD(C_{l}T^{-1})+ \ORD(|n|L_{l}T^{-1}) \right\}\right. \\
&\quad \left. \times \sum_{t}\psi_{l,n-t}\psi_{j,-t}\sum_{h}\psi_{l,n+k-h}\psi_{j',k'-h}\right|,
\end{align*}
where we have used the spectrum definition and the Lipschitz continuity of $S_{j}(\cdotp,r/R)$ in the rescaled time argument.

Using a substitution $u=h-t$, the sums of wavelet products above can be manipulated as follows,
\begin{align}
\sum_{n\in\mathbb{Z}}\sum_{t}\psi_{l,n-t}\psi_{j,-t}
\sum_{h}\psi_{l,n+k-h}\psi_{j',k'-h} &= \sum_{n\in\mathbb{Z}}\sum_{t}\psi_{l,n-t}\psi_{j,-t}
\sum_{u}\psi_{l,n+k-u-t}\psi_{j',k'-u-t} \nonumber \\
&= \sum_{u}\sum_{t}\psi_{j,-t}\psi_{j',k'-u-t}
\sum_{n\in\mathbb{Z}}\psi_{l,n-t}\psi_{l,n+k-u-t} \nonumber\\
&= \sum_{u}\Psi_{j,j'}(u-k')\Psi_{l}(u-k), \nonumber \\
&= \sum_{u}\Psi_{j,j'}(u)\Psi_{l}(u-k+k'), \nonumber 
\end{align}
where for the last equality we used a substitution $u:=u-k'$. Equivalently, using the first result in the proof of Lemma~\ref{lem:ASordjj}, the above could have been directly written as $\sum_{u}\Psi_{j,l}(u)\Psi_{j',l}(u-k+k')$.

Hence
\beq\nonumber
\beta^{\frac{1}{2}} = \left| \sum_{l=1}^{\infty}\sum_{u\in\Z} \Psi_{j,j'}(u)\Psi_{l}(u-k+k') S_{l}\left(\frac{k}{T},\frac{r}{R}\right)
+I+II+III \right|.
\eeq

The term $I$ can be bounded as follows
\beqann
|I|&=&\left|\sum_l \ORD(D_{l}R^{-1})\sum_{u}\Psi_{j,l}(u)\Psi_{j',l}(u-k+k')\right|,\\
&\leq&
\ORD(R^{-1})\sum_{l} \left( D_l \sum_u \left| \Psi_{j,l}(u)\Psi_{j',l}(u-k+k') \right| \right), \mbox{then using Cauchy Schwarz inequality}\\
&\leq&
\ORD(R^{-1})\sum_{l} D_l \left(\sum_u \left| \Psi_{j,l}(u)\right|^2\right)^{1/2}
\left(\sum_u \left| \Psi_{j',l}(u-k+k')\right|^2\right)^{1/2},\\
&=&
\ORD(R^{-1})\sum_{l} D_l (A_{j,l})^{1/2} (A_{j',l})^{1/2}, \mbox{from the definition of }A_{j,l},\\
&=&
\ORD(R^{-1})\sum_{l} (D_l A_{j,l})^{1/2} (D_l A_{j',l})^{1/2},\\
&\leq&
\ORD(R^{-1})\left(\sum_{l} D_l A_{j,l}\right)^{1/2}
\left(\sum_{l} D_l A_{j',l}\right)^{1/2},\mbox{ using again Cauchy Schwarz inequality}\\
&\leq&
\ORD(R^{-1})\ORD(1), \mbox{ as we saw in the expectation proof part }\sum_{l} D_l A_{j,l}=\ORD(1),\\
&=&
\ORD(R^{-1}).
\eeqann

Similarly, $|II| =\left| \sum_l \ORD(C_{l}T^{-1})\sum_{u}\Psi_{j,l}(u)\Psi_{j',l}(u-k+k')\right|=\ORD(T^{-1})$.

We bound the term $III$  by noting that the function $\Psi_{j,l}\Psi_{j',l}$ is compactly supported, with the support bounded by $K'\left( 2^l +2^{\mbox{min}\{j,j'\}}\right)$, hence
\beqann
|III|&\leq& \sum_l \ORD(L_{l}T^{-1})\left( 2^l +2^{\mbox{min}\{j,j'\}}\right)\sum_{u} \left|\Psi_{j,l}(u)\Psi_{j',l}(u-k+k')\right|,\\
&\leq&
\ORD(T^{-1})\sum_l L_{l}\left( 2^l +2^{\mbox{min}\{j,j'\}}\right) (A_{j,l})^{1/2} (A_{j',l})^{1/2}, \mbox{ as above}\\
&=&
\ORD(T^{-1})\sum_l L_{l}2^l  (A_{j,l})^{1/2} (A_{j',l})^{1/2}\\
&+&
\ORD(T^{-1})\sum_l L_{l}2^{\mbox{min}\{j,j'\}} (A_{j,l})^{1/2} (A_{j',l})^{1/2},\\
&\leq&
\ORD(T^{-1})2^{(j+j')/2}\sum_l L_{l}2^l  (2^{-j/2}A_{j,l}^{1/2}) (2^{-j'/2}A_{j',l}^{1/2})\\
&+&
\ORD(T^{-1})2^{(j+j')/2}\sum_l L_{l}  (A_{j,l})^{1/2} (A_{j',l})^{1/2}, \mbox{ as } 2^{\mbox{min}\{j,j'\}} \leq 2^{(j+j')/2}.
\eeqann
The term $\sum_l L_{l}2^l  (2^{-j/2}A_{j,l}^{1/2}) (2^{-j'/2}A_{j',l}^{1/2})\leq \left(\sum_l L_{l}2^l 2^{-j}A_{j,l}\right)^{1/2}\left(\sum_l L_{l}2^l 2^{-j'}A_{j',l}\right)^{1/2}$ from the Cauchy-Schwarz inequality, and using $\sum_{j} 2^{-j}A_{j,l}=1$ and $\sum_{l} 2^l L_l<\infty$, we obtain its $\ORD(1)$ bound.

The term $\sum_l L_{l}  (A_{j,l})^{1/2} (A_{j',l})^{1/2} \leq \left(\sum_l L_{l}A_{j,l}\right)^{1/2} \left(\sum_l L_{l}A_{j',l}\right)^{1/2}=\ORD(1)$ since we already established that $\sum_l L_{l}A_{j,l}<\infty$ in the expectation part of the proof. Thus term $III$ is bounded by $\ORD(2^{(j+j')/2}T^{-1})$.

We therefore have \beq \nonumber
\beta^{\frac{1}{2}} = \left| \sum_{l=1}^{\infty}\sum_{u\in\Z} \Psi_{j,j'}(u)\Psi_{l}(u-k+k') S_{l}\left(\frac{k}{T},\frac{r}{R}\right) \right|
+\ORD(2^{(j+j')/2}T^{-1})+\ORD(R^{-1}),
\eeq
and using the result in Lemma~\ref{lem:ASordjj}, we obtain
\beq\nonumber
\beta = \left(\sum_{l=1}^{\infty}\sum_{u\in\Z} \Psi_{j,j'}(u)\Psi_{l}(u-k+k') S_{l}\left(\frac{k}{T},\frac{r}{R}\right)\right)^2
+\ORD(2^{(j+j')/2+(j+j')/2}T^{-1})+
\ORD(2^{(j+j')/2}R^{-1}),
\eeq
where we retained the highest order terms only.

The third term in \eqref{eq:EII}, $\gamma$, can be establish using precisely the same arguments.

Using all previous results, when $r \neq r'$
\begin{align*}
\E\left[I_{j,k;T}^{r;R}I_{j',k';T}^{r';R}\right] &= \alpha + \beta + \gamma = \alpha \nonumber\\
&= \sum_{l=1}^{\infty}A_{j,l} S_{l}\left(\frac{k}{T},\frac{r}{R}\right)\sum_{e=1}^{\infty}A_{j',e} S_{e}\left(\frac{k'}{T},\frac{r'}{R}\right)+ \ORD(2^{j+j'}T^{-1}) \nonumber \\
& \quad +\ORD(2^{\text{max}\{j,j'\}} R^{-1}),
\end{align*}
as the $\beta$ and $\gamma$ terms are both $0$ due to \eqref{eq:EEXX2}.

When $r = r'$, we obtain the following
\begin{align}\label{eq:covIdt}
\E\left[I_{j,k;T}^{r;R}I_{j',k';T}^{r;R}\right] &= \alpha + \beta + \gamma \nonumber\\
&= \sum_{l=1}^{\infty}A_{j,l} S_{l}\left(\frac{k}{T},\frac{r}{R}\right)\sum_{e=1}^{\infty}A_{j',e} S_{e}\left(\frac{k'}{T},\frac{r}{R}\right)+ \ORD(2^{j+j'}T^{-1})+\ORD(2^{\text{max}(j,j')} R^{-1}) \nonumber \\
& \quad+ 2\left(\sum_{l=1}^{\infty}\sum_{u\in\Z}\Psi_{j,j'}(u)\Psi_{l}(u+k-k')S_{l}\left(\frac{k}{T},\frac{r}{R}\right) \right)^2 \nonumber \\
& \quad+ \ORD(2^{j+j'}T^{-1})
+\ORD(2^{(j+j')/2}R^{-1}).
\end{align}

Thus for $r = r', j = j'$ and $k = k'$, the variance of $I_{j,k;T}^{r;R}$ can be written as
\begin{align*}
\text{var}\left(I_{j,k;T}^{r;R}\right) &= \E\left[I_{j,k;T}^{r;R}I_{j,k;T}^{r;R}\right] - \E\left[I_{j,k;T}^{r;R}\right]\E\left[I_{j,k;T}^{r;R}\right] \\
&= 3\left(\sum_{l=1}^{\infty}A_{j,l} S_{l}\left(\frac{k}{T},\frac{r}{R}\right)\right)^{2} + \ORD(2^{2j}T^{-1})+\ORD(2^{j}R^{-1}) \\
& \qquad\qquad\qquad\qquad\qquad - \left(\sum_{l=1}^{\infty}A_{j,l} S_{l}\left(\frac{k}{T},\frac{r}{R}\right) + \ORD(2^{j}T^{-1}) + \ORD(R^{-1})\right)^{2} \\
&= 2\left(\sum_{l=1}^{\infty}A_{j,l} S_{l}\left(\frac{k}{T},\frac{r}{R}\right)\right)^{2} + \ORD(2^{2j}T^{-1}) + \ORD(2^{j}R^{-1}),
\end{align*}
where for the last equality we used the result in Lemma~\ref{lem:ASord} and we have only kept the highest order terms.
\end{proof}

\subsection{Proof of Proposition~\ref{prop:ESI} }\label{app:proofs:itilde}

\subsubsection*{Proof of Proposition \ref{prop:ESI} (Expectation)}
\begin{proof}[\unskip\nopunct]
From the definition of the replicate-smoothed periodogram in \eqref{eq:defItilde}, we have
\begin{equation}\nonumber
\E\left[\It_{j,k;T}^{r;R}\right] = \frac{1}{2M+1}\sum_{s=-M}^{M}\E\left[I_{j,k;T}^{r+s;R}\right]
\end{equation}
and substituting the asymptotic result for the expectation  $\E\left[I_{j,k;T}^{r+s;R}\right]$ (see equation~\eqref{eq:EI}), we further obtain
\begin{equation}\nonumber
\E\left[\It_{j,k;T}^{r;R}\right] = (2M+1)^{-1}\sum_{s=-M}^{M}\left\{\sum_{l=1}^{\infty}A_{j,l} S_{l}\left(\frac{k}{T},\frac{r+s}{R}\right) + \ORD(2^jT^{-1}) + \ORD(R^{-1})\right\}.
\end{equation}
Using the Lipschitz continuity of the spectrum $S_{j}(k/T,\cdotp)$ in replicate-dimension, we obtain
\begin{align*}
\E\left[\It_{j,k;T}^{r;R}\right] &= (2M+1)^{-1}\sum_{s=-M}^{M}\left\{\sum_{l=1}^{\infty}A_{j,l} \left(S_{l}\left(\frac{k}{T},\frac{r}{R}\right)+ \ORD(|s|N_{l}R^{-1})\right) + \ORD(2^{j}T^{-1}) + \ORD(R^{-1})\right\}\\
&= (2M+1)^{-1}\sum_{s=-M}^{M}\left\{\sum_{l=1}^{\infty}A_{j,l} S_{l}\left(\frac{k}{T},\frac{r}{R}\right)+ \ORD(|s|R^{-1})\sum_l N_l A_{j,l}+ \ORD(2^{j}T^{-1}) + \ORD(R^{-1})\right\} \\
&= \sum_{l=1}^{\infty}A_{j,l} S_{l}\left(\frac{k}{T},\frac{r}{R}\right)+ \ORD(MR^{-1}) + \ORD(2^j T^{-1}),
\end{align*}
where we have used that $\sum_l N_l A_{j,l} < \infty$ as $\sum_{l}2^{l}N_{l} < \infty$ and $\sum_l 2^{-l} A_{j,l}=1$. \\
\end{proof}

\subsubsection*{Proof of Proposition \ref{prop:ESI} (Variance)}
\begin{proof}[\unskip\nopunct]
Let us take
\begin{align} \label{eq:CSII}
\text{var}\left(\It_{j,k;T}^{r;R}\right)&= \text{var}\left( (2M+1)^{-1} \sum_{s=-M}^{M} I_{j,k;T}^{r+s;R} \right) \nonumber\\
& = (2M+1)^{-2} \left( \E\left[\sum_{s=-M}^{M} I_{j,k;T}^{r+s;R}\sum_{s'=-M}^{M} I_{j,k;T}^{r+s';R} \right] - \E\left[\sum_{s=-M}^{M} I_{j,k;T}^{r+s;R}\right]\E\left[\sum_{s'=-M}^{M} I_{j,k;T}^{r+s';R}\right] \right) \nonumber\\
&= (2M+1)^{-2}\sum_{s=-M}^{M}\sum_{s'=-M}^{M} \cov\left(I_{j,k;T}^{r+s;R},I_{j,k;T}^{r+s';R}\right).
\end{align}
As the replicates are uncorrelated, we have $s = s'$ in \eqref{eq:CSII} and using the variance result for the raw periodogram (Proposition~\ref{prop:EI}) we obtain
\begin{align*}
\text{var}\left(\It_{j,k;T}^{r;R}\right)
&=(2M+1)^{-2}\sum_{s=-M}^{M}\text{var}\left(I_{j,k;T}^{r+s;R}\right) \nonumber \\
&= (2M+1)^{-2}\sum_{s=-M}^{M}\left\{2\left(\sum_{l=1}^{\infty}A_{j,l} S_{l}\left(\frac{k}{T},\frac{r+s}{R}\right)\right)^{2} + \ORD(2^{2j}T^{-1}) + \ORD(2^{j}R^{-1})\right\} \nonumber\\
&= (2M+1)^{-2}\sum_{s=-M}^{M}\left\{2\left(\sum_{l=1}^{\infty}A_{j,l} S_{l}\left(\frac{k}{T},\frac{r}{R}\right)+ \ORD(|s|R^{-1})\sum_l N_l A_{j,l}
\right)^{2}\right\}  \\
&\quad +(2M+1)^{-2}\sum_{s=-M}^{M}\left\{\ORD(2^{2j}T^{-1}) + \ORD(2^{j}R^{-1})\right\} \nonumber\\
&=\ORD(M^{-1})\left(\sum_{l=1}^{\infty}A_{j,l}S_{l}\left(\frac{k}{T},\frac{r}{R}\right)\right)^{2}+
(2M+1)^{-2}\sum_{s=-M}^{M}\left\{
\ORD(2^j|s|R^{-1})+
\ORD(|s|^2R^{-2})\right\} \\
&\quad +\ORD(2^{2j}(MT)^{-1}) + \ORD(2^{j}(MR)^{-1}),
\end{align*}
from Lemma~\ref{lem:ASord} and as $\sum_l N_l A_{j,l}<\infty$.

Retaining the largest order terms, it then follows that
\begin{align*}
\text{var}\left(\It_{j,k;T}^{r;R}\right)
&=\ORD(2^{2j}M^{-1})+\ORD(2^jR^{-1})+\ORD(MR^{-2})+\ORD(2^{2j}(MT)^{-1}) + \ORD(2^{j}(MR)^{-1}),\\
&=\ORD(2^{2j}M^{-1})+\ORD(2^jR^{-1})+\ORD(MR^{-2}).
\end{align*}

The expectation and variance results show that for fixed coarse enough scales $j$ (to guard against asymptotic bias and non-vanishing variance), the proposed replicate-smoothed periodogram is an asymptotically consistent estimator for the spectral quantity $\beta$, as it is asymptotically unbiased and its variance converges to zero as $T \to \infty, R \to \infty$, $M \to \infty$ and $M/R \to 0$. \\
\end{proof}

\subsection{Proof of Proposition~\ref{prop:CS} }\label{app:proofs:corrper}
\begin{proof}
As $M, T \rightarrow \infty$, for each $j$, $z$ and $\nu$, the consistency result $\hat{S}_{j}(z,\nu)\stackrel{P}{\rightarrow} {S}_{j}(z,\nu)$ follows from the consistency results $\It_{l,\lfloor zT \rfloor;T}^{\lfloor \nu R \rfloor;R} \stackrel{P}{\rightarrow}  \beta_{l}(z,\nu)$  for all fine enough scales $l$ (as shown in Proposition~\ref{prop:ESI}) and then using the continuous mapping theorem \citep{billingsley1999convergence} for the continuous function $g(x_1, \ldots, x_J)= \sum_{l=1}^J A^{-1}_{j,l} x_l$ that defines their linear combination with coefficients given by the matrix $A^{-1}$ entries.

Additionally, using the properties of the matrix $A$, we obtain the estimator asymptotic unbiasedness from the linearity of the expectation operator and from the asymptotic unbiasedness of the corrected periodogram, as follows
\beqann
\E(\hat{S}_{j}(z,\nu))&=&\E\left(\sum_{l=1}^J A^{-1}_{j,l} \It_{l,\lfloor zT \rfloor;T}^{\lfloor \nu R \rfloor;R}\right),\\&=& \sum_{l=1}^J A^{-1}_{j,l} \E\left(\It_{l,\lfloor zT \rfloor;T}^{\lfloor \nu R \rfloor;R}\right),\mbox{ then from the expectation part of Proposition~\ref{prop:ESI}}\\
&=&\sum_{l=1}^J A^{-1}_{j,l} \left(\beta_{l}(z,\nu) + \ORD(2^l T^{-1})+\ORD(MR^{-1})\right), \mbox{ and using the definition of }\beta\\
&=& \sum_{l=1}^J A^{-1}_{j,l}\sum_{l'} A_{l,l'}S_{l'}(z,\nu)+ \sum_{l=1}^J A^{-1}_{j,l}\left( \ORD(2^l T^{-1})+\ORD(MR^{-1}) \right),\\
&=& \sum_{l'} \left(\sum_{l=1}^J A^{-1}_{j,l}A_{l,l'}\right)S_{l'}(z,\nu)+ \sum_{l=1}^J A^{-1}_{j,l}\left( \ORD(2^l T^{-1})+\ORD(MR^{-1}) \right),\\
&=& \sum_{l'} \left( A^{-1}A\right)_{j,l'}S_{l'}(z,\nu)
+ \sum_{l=1}^J A^{-1}_{j,l}\left( \ORD(2^l T^{-1})+\ORD(MR^{-1}) \right),\\
&=& S_{j}(z,\nu)+\ORD(T^{\alpha-1})+\ORD(MR^{-1}),
\eeqann
where we used the boundedness of $A^{-1}_{j,l}$ and $\sum_{l=1}^J 2^l =\ORD(T^\alpha)$.

In fact, it can be shown that for Haar wavelets, the above approximation rate is  $\ORD(T^{-1/2})+\ORD(MR^{-1})$ since $A^{-1}_{j,l} =\ORD(2^{-(j+l)/2})$ \citep{nvsk:2000}.
\end{proof}

\subsection{Proof of Proposition~\ref{prop:ESIts} }\label{app:proofs:itildetilde}

\subsubsection*{Proof of Proposition \ref{prop:ESIts} (Expectation)}
\begin{proof}[\unskip\nopunct]
From the definition of the time- and replicate-smoothed periodogram in \eqref{eq:defItildetilde}, we have
\begin{equation}\nonumber
\E\left[\tilde{\It}_{j,k;T}^{r;R}\right] = \frac{1}{2M+1}\frac{1}{2M_T+1}\sum_{s=-M}^{M}\sum_{t=-M_T}^{M_T}\E\left[I_{j,k+t;T}^{r+s;R}\right]
\end{equation}
and substituting the asymptotic result for the expectation  $\E\left[I_{j,k+t;T}^{r+s;R}\right]$ (see for instance equation~\eqref{eq:EI}), we further obtain
\begin{equation*}
\E\left[\tilde{\It}_{j,k;T}^{r;R}\right] = (2M+1)^{-1}(2M_T+1)^{-1}\sum_{s=-M}^{M} \sum_{t=-M_T}^{M_T}\left\{\sum_{l=1}^{\infty}A_{j,l} S_{l}\left(\frac{k+t}{T},\frac{r+s}{R}\right) + \ORD(2^jT^{-1}) + \ORD(R^{-1})\right\}.
\end{equation*}
Now using the Lipschitz continuity of the spectrum $S_{l}(\cdotp,\cdotp)$ in both time and replicate dimensions, we have
\begin{align*}
\E\left[\tilde{\It}_{j,k;T}^{r;R}\right] &= (2M+1)^{-1}(2M_T+1)^{-1}\sum_{s=-M}^{M}\sum_{t=-M_T}^{M_T}
\left\{\sum_{l=1}^{\infty}A_{j,l} S_{l}\left(\frac{k}{T},\frac{r}{R}\right)\right\}\\
&+(2M+1)^{-1}(2M_T+1)^{-1}\sum_{s=-M}^{M}\sum_{t=-M_T}^{M_T}\sum_{l=1}^{\infty}A_{j,l}
\left(\ORD(|t|L_{l}T^{-1})+\ORD(|s|N_{l}R^{-1})\right)\\
&+ \ORD(2^{j}T^{-1}) + \ORD(R^{-1}),\\
&=  (2M+1)^{-1}(2M_T+1)^{-1}\sum_{s=-M}^{M}\sum_{t=-M_T}^{M_T}
\sum_{l=1}^{\infty}A_{j,l} S_{l}\left(\frac{k}{T},\frac{r}{R}\right)\\
&+ (2M+1)^{-1}(2M_T+1)^{-1}\sum_{s=-M}^{M}\sum_{t=-M_T}^{M_T} \left\{ \ORD(|t|T^{-1})\sum_l L_l A_{j,l} + \ORD(|s|R^{-1})\sum_l N_l A_{j,l} \right\} \\
&+ \ORD(2^{j}T^{-1}) + \ORD(R^{-1}),\\
&= \sum_{l=1}^{\infty}A_{j,l} S_{l}\left(\frac{k}{T},\frac{r}{R}\right)+ \ORD(M_T T^{-1})+ \ORD(MR^{-1}) + \ORD(2^j T^{-1}),
\end{align*}
where we have used that $\sum_l L_l A_{j,l} < \infty$ and $\sum_l N_l A_{j,l} < \infty$ as previously shown, and as usual we retained the highest order terms.
\end{proof}

\subsubsection*{Proof of Proposition \ref{prop:ESIts} (Variance)}
\begin{proof}[\unskip\nopunct]
Under the assumption of summable autocovariance,
$$\sup_{z,\nu}\sum_{n\in\Z}\left|c\left(z,\nu;n \right)\right| =\ORD(1),$$ the result in Lemma  \ref{lem:ASordtau} further implies
\begin{equation}\label{eq:Aeta}
\sum_{\uptau\in \Z}\left|\sum_{l=1}^{\infty}A_{l,j}^{\uptau} S_{l}\left(z,\nu\right)\right|=\ORD(2^{j}).
\end{equation}
This can be easily seen by taking
\begin{align*}
\sum_{\uptau}\left|\sum_{l=1}^{\infty}A_{l,j}^{\uptau} S_{l}\left(z,\nu\right)\right| & = \sum_{\uptau} \left|\sum_{l=1}^{\infty}\sum _{n}\Psi_{l}(n) \Psi_{j}(n+\uptau) S_{l}\left(z,\nu\right)\right| \mbox{ from the definition of the matrix }A^\uptau\\
&= \sum_{\uptau} \left|\sum _{n}\left(\sum_{l=1}^{\infty} S_{l}\left(z,\nu\right)\Psi_{l}(n)\right)\Psi_{j}(n+\uptau)\right|  \\
&= \sum_{\uptau} \left|\sum _{n}c\left(z,\nu;n \right)\Psi_{j}(n+\uptau)\right| \mbox { from the local autocovariance definition}\\
& \leq \sum _{n} \left( \sum_{\uptau} \left|\Psi_{j}(n+\uptau)\right| \right) \left|c\left(z,\nu;n \right)\right| \mbox { using the triangle inequality}\\
&= \ORD(2^{j}) \sum_n \left|c\left(z,\nu;n \right)\right| = \ORD(2^{j}),
\end{align*}
where we used $\sup_{z,\nu}\sum_n \left|c\left(z,\nu;n \right)\right|=\ORD(1)$ and that $\sum _{\uptau}\left|\Psi_{j}(n+\uptau)\right| = \ORD(2^{j})$.

In the same vein, we show that \begin{equation}\label{eq:Aetal}\sum_{\uptau\in\Z} \left|\sum_{l=1}^\infty L_{l} A_{j,l}^{\uptau} \right|= \ORD(2^j),\end{equation} and similarly for terms involving $N_l$, since
\begin{align*}
\sum_{\uptau}\left|\sum_{l=1}^{\infty}L_l A_{l,j}^{\uptau} \right| & = \sum_{\uptau} \left|\sum_{l=1}^{\infty} L_l \sum _{n}\Psi_{j}(n) \Psi_{l}(n+\uptau) \right| \mbox{ from the definition of the matrix }A^\uptau\\
&\leq \sum_{l=1}^{\infty} L_l \sum_{\uptau} \sum _{n} \left| \Psi_{j}(n)\Psi_{l}(n+\uptau)\right|  \mbox { using the triangle inequality} \\
&=  \sum_{l=1}^{\infty} L_l  \sum _{n} \left(\left| \Psi_{j}(n)\right| \sum_{\uptau}\left| \Psi_{l}(n+\uptau)\right| \right)\\
& \leq K \sum_{l=1}^{\infty} 2^l L_l  \sum _{n} \left| \Psi_{j}(n)\right| \\
&= \ORD(2^{j}),
\end{align*}
where we used (in order) that $\sum _{\uptau}\left|\Psi_{l}(n+\uptau)\right| = \ORD(2^{l})$,  $\sum _{n}\left|\Psi_{j}(n)\right| = \ORD(2^{j})$ and $\sum_l 2^l L_l <\infty$.

We also note here that  $\left| \sum_{l=1}^\infty 2^{-l} A_{j,l}^{\uptau} \right|= \ORD(1)$, which can be immediately obtained by taking
\begin{align*}
\left| \sum_l 2^{-l} A_{j,l}^{\uptau} \right|&= \left|\sum_{l} 2^{-l} \sum _{n}\Psi_{j}(n) \Psi_{l}(n+\uptau) \right|,\\
&= \left|\sum _{n}\Psi_{j}(n) \sum_{l} 2^{-l} \Psi_{l}(n+\uptau) \right|,\\
&= \left|\sum _{n}\Psi_{j}(n) \delta_{0,n+\tau} \right|, \\
&=\left| \Psi_{j}(-\uptau) \right|= \ORD(1).
\end{align*}
where we used $\sum_{l} 2^{-l} \Psi_{l}(n+\uptau)=\delta_{0,n+\tau}$ as shown in  \cite{fryz03:forecasting}.

Now let us take
\begin{align} \label{eq:CSIIts}
\text{var}\left(\tilde{\It}_{j,k;T}^{r;R}\right)&= \text{var}\left( (2M_T+1)^{-1}\sum_{t=-M_T}^{M_T} \It_{j,k+t;T}^{r;R} \right), \nonumber\\
&= (2M_T+1)^{-2} \sum_{t=-M_T}^{M_T} \sum_{\eta=-M_T-t}^{M_T-t}
\cov\left(\It_{j,k+t;T}^{r;R}, \It_{j,k+t+\eta;T}^{r;R}\right), \nonumber\\
&= (2M_T+1)^{-2} (2M+1)^{-2} \sum_{t=-M_T}^{M_T} \sum_{\eta=-M_T-t}^{M_T-t}  \sum_{s=-M}^{M} \cov\left(I_{j,k+t;T}^{r+s;R}, I_{j,k+t+\eta;T}^{r+s;R}\right),
\end{align}
where we have used the definition of the replicate-smoothed periodogram in~\eqref{eq:defItilde} and the assumption of uncorrelated replicate series.

Using the covariance definition and the result in equation~\eqref{eq:covIdt} of the variance proof for the raw periodogram (Appendix \ref{app:proofs:rawi} for Proposition~\ref{prop:EI}), as well as the expectation result for the raw periodogram (Proposition~\ref{prop:EI}), we can re-write equation~\eqref{eq:CSIIts} above as follows

\begin{align*}
\text{var}\left(\tilde{\It}_{j,k;T}^{r;R}\right)
&= (2M_T+1)^{-2} (2M+1)^{-2} \sum_{t=-M_T}^{M_T} \sum_{\eta=-M_T-t}^{M_T-t}  \sum_{s=-M}^{M}
\left\{2\left(\sum_{l=1}^{\infty}A_{j,l}^{\eta} S_{l}\left(\frac{k+t}{T},\frac{r+s}{R}\right)\right)^{2}\right\} \\
&+ (2M_T+1)^{-2} (2M+1)^{-2} \sum_{t=-M_T}^{M_T} \sum_{\eta=-M_T-t}^{M_T-t}  \sum_{s=-M}^{M} \left\{ \ORD(2^{2j}T^{-1}) + \ORD(2^{j}R^{-1})\right\}.
\end{align*}

From the Lipschitz continuity of $S_{l}(\cdotp,\cdotp)$, we have
\begin{align*}
\sum_{l=1}^{\infty}A_{j,l}^{\eta} S_{l}\left(\frac{k+t}{T},\frac{r+s}{R}\right)&= \sum_{l=1}^{\infty}A_{j,l}^{\eta} S_{l}\left(\frac{k}{T},\frac{r}{R}\right)+ \left(\ORD(|t|T^{-1})\sum_l L_l A_{j,l}^{\eta}+ \ORD(|s|R^{-1})\sum_l N_l A_{j,l}^{\eta}\right),\\
&= \sum_{l=1}^{\infty}A_{j,l}^{\eta} S_{l}\left(\frac{k}{T},\frac{r}{R}\right)+
\left(\mbox{term }I \right), 
\end{align*}
where the notation `term $I$' above is used for brevity. Replacing this into the variance formula above, we obtain

\begin{align*}
\text{var}\left(\tilde{\It}_{j,k;T}^{r;R}\right)
&\leq (2M_T+1)^{-2} (2M+1)^{-2} \sum_{t=-M_T}^{M_T} \sum_{s=-M}^{M} \sum_{\eta\in\Z}
2\left(\sum_{l=1}^{J}A_{j,l}^{\eta} S_{l}\left(\frac{k}{T},\frac{r}{R}\right)\right)^2\\
&+ (2M_T+1)^{-2} (2M+1)^{-2} \sum_{t=-M_T}^{M_T} \sum_{s=-M}^{M} \sum_{\eta\in\Z} 2\left(\mbox{term }I\right)^2  \\
&+ (2M_T+1)^{-2} (2M+1)^{-2} \sum_{t=-M_T}^{M_T}\sum_{s=-M}^{M}
\sum_{\eta\in\Z}
4\left(\sum_{l=1}^{J}A_{j,l}^{\eta}S_{l}\left(\frac{k}{T},\frac{r}{R}\right)\right)
\times \left(\mbox{term }I\right) \\
&+(2M_T+1)^{-2} (2M+1)^{-2} \sum_{t=-M_T}^{M_T} \sum_{\eta=-M_T-t}^{M_T-t}  \sum_{s=-M}^{M}\left\{\ORD(2^{2j}T^{-1}) + \ORD(2^{j}R^{-1})\right\}.
\end{align*}

As $\sum_{\eta} \left(\sum_{l}A_{j,l}^{\eta} S_{l}\left(\frac{k}{T},\frac{r}{R}\right)\right)^2 \leq \left(\sum_{\eta} \left|\sum_{l}A_{j,l}^{\eta} S_{l}\left(\frac{k}{T},\frac{r}{R}\right)\right| \right)^2$, from~\eqref{eq:Aeta} we obtain  \\ $\sum_{\eta} \left(\sum_{l}A_{j,l}^{\eta} S_{l}\left(\frac{k}{T},\frac{r}{R}\right)\right)^2=\ORD(2^{2j})$, hence the first term in the variance sum above is $\ORD(2^{2j}(M_TM)^{-1})$.

To determine the order of the second term in the variance sum, as above we first observe that $\sum_\eta \left(\sum_l L_l A_{j,l}^{\eta} \right)=\ORD(2^{2j})$ from equation~\eqref{eq:Aetal}, and similarly for the terms involving $N_l$ (and cross-terms), which leads to the second term in the variance sum above to be rephrased as
$$(2M_T+1)^{-2} (2M+1)^{-2} \ORD(2^{2j}) \sum_{t=-M_T}^{M_T} \sum_{s=-M}^{M} \left\{ \ORD(|t|^2T^{-2})+\ORD(|s|^2R^{-2})+\ORD(|t||s|(TR)^{-1})\right\}.$$

The third (cross-)term of the variance sum can be obtained from  $$\sum_\eta \left|\sum_{l}A_{j,l}^{\eta} S_{l}\left(\frac{k}{T},\frac{r}{R}\right)\right| \left|\sum_l L_l A_{j,l}^{\eta} \right|=\ORD(2^{2j}),$$ using similar arguments as above and results~\eqref{eq:Aeta} and ~\eqref{eq:Aetal}; the same holds for terms involving $N_l$. Hence the third term can be re-expressed as
$$(2M_T+1)^{-2} (2M+1)^{-2} \ORD(2^{2j}) \sum_{t=-M_T}^{M_T} \sum_{s=-M}^{M} \left\{\ORD(|t|T^{-1})+ \ORD(|s|R^{-1})\right\}.$$

Now replacing all order terms in the variance formula above leads to
\begin{align*}
\text{var}\left(\tilde{\It}_{j,k;T}^{r;R}\right)
&= \ORD(2^{2j}(M_TM)^{-1})\\
&+ \ORD(2^{2j}M_T T^{-2}M^{-1})+\ORD(2^{2j}(M_T)^{-1}M R^{-2})+\ORD(2^{2j}(TR)^{-1})\\
&+ \ORD(2^{2j}(MT)^{-1})+\ORD(2^{2j}(M_TR)^{-1})\\
&+ \ORD(2^{2j}(MT)^{-1}) + \ORD(2^{j}(MR)^{-1}).
\end{align*}

Retaining the largest order terms, it then follows that
\begin{align*}
\text{var}\left(\tilde{\It}_{j,k;T}^{r;R}\right)
&=\ORD(2^{2j}(M_TM)^{-1})+\ORD(2^{j}(MR)^{-1}),
\end{align*}
which completes the proof.
\end{proof}

\section{Proofs of results on the asymptotic behaviour of proposed estimators embedding replicate coherence}\label{app:proofsc}

In this section, we give details of the proofs in Section~\ref{sec:modelext}, using the notation described therein. In the proofs that follow, we make use of the following results.

\begin{lemma}\label{lem:lipz}
Under the assumptions of Definition~\ref{def:rlswc}, we have a sequence $\{B_j\}$ of uniformly bounded Lipschitz constants in $j$ with $\sum_j2^jB_j<\infty$ such that
\begin{align}
& \left|\Wt_{j}\left(\frac{k+n}{T},\frac{r}{R}\right)\Wt_{j}\left(\frac{k+n}{T},\frac{r'}{R}\right)\rho_{j}\left(\frac{k+n}{T},\frac{r}{R},\frac{r'}{R}\right) \right. \nonumber \\
& \qquad \qquad \left. - \Wt_{j}\left(\frac{k}{T},\frac{r}{R}\right)\Wt_{j}\left(\frac{k}{T},\frac{r'}{R}\right)\rho_{j}\left(\frac{k}{T},\frac{r}{R},\frac{r'}{R}\right)\right| \leq |n| B_{j}T^{-1}, \label{eq:crosslipz}
\end{align}
for any replicates $r$, $r'$ and times $k$, $n$.
\end{lemma}

Note that the above result means that the replicate cross-spectrum, \i.e. $S_{j}(\cdotp,\nu,\nu')$, is Lipschitz continuous in the rescaled time argument for any rescaled replicates $\nu$, $\nu'$.

\begin{proof}
The proof can be seen in Section~\ref{supp:proofs:sec5} of the Supplementary Material.
\end{proof}

\begin{lemma}\label{lem:lipr}
Under the assumptions of Definition~\ref{def:rlswc}, we have a sequence $\{B_j^\prime\}$ of uniformly bounded Lipschitz constants in $j$ with $\sum_j2^jB^\prime_j<\infty$ such that
\begin{align*}
& \left|\Wt_{j}\left(\frac{k}{T},\frac{r+s}{R}\right)\Wt_{j}\left(\frac{k}{T},\frac{r'+s}{R}\right)\rho_{j}\left(\frac{k}{T},\frac{r+s}{R},\frac{r'+s}{R}\right) \right. \nonumber\\
& \qquad \qquad \left. - \Wt_{j}\left(\frac{k}{T},\frac{r}{R}\right)\Wt_{j}\left(\frac{k}{T},\frac{r'}{R}\right)\rho_{j}\left(\frac{k}{T},\frac{r}{R},\frac{r'}{R}\right)\right| \leq |s| B^\prime_{j}R^{-1}, 
\end{align*}
for any times $k$, $n$ and replicates $r$, $r'$.
\end{lemma}

Note that the above result effectively states that the replicate cross-spectrum, \i.e. $S_{j}(z,\nu+\cdotp,\nu'+ \cdotp)$, is Lipschitz continuous in the rescaled replicate arguments.

\begin{proof}
The proof can be seen in Section~\ref{supp:proofs:sec5} of the Supplementary Material.
\end{proof}

\begin{lemma}\label{lem:ASordc}
Under the assumptions of Definition~\ref{def:rlswc}, we have
 $\sum_{l=1}^{J}A_{j,l} S_{l}\left(z,\nu,\nu'\right)=\ORD(2^{j})$.
\end{lemma}

\begin{proof}
The proof follows the same steps as for Lemma~\ref{lem:ASord}.
\end{proof}

\begin{lemma}\label{lem:ASordjjc}
Under the assumptions of Definition~\ref{def:rlswc}, we have
$$\sum_{l=1}^{J} \sum _{n}\Psi_{l}(n+\uptau) \Psi_{j,j'}(n) S_{l}\left(z,\nu,\nu'\right)=\ORD(2^{\mbox{max}\{j,j'\}}).$$
\end{lemma}

\begin{proof}
The proof follows the same steps as for Lemma~\ref{lem:ASordtau}.
\end{proof}

\subsection{Proof of Proposition \ref{prop:ESIc}}\label{app:proofs:itildec}
\subsubsection*{Proof of Proposition \ref{prop:ESIc} (Expectation)}
The proof follows similar steps to the proof of its non-coherence counterpart in Appendix~\ref{app:proofs:itilde}.

\subsubsection*{Proof of Proposition \ref{prop:ESIc} (Variance)}
\begin{proof}[\unskip\nopunct]
\begin{align} \label{eq:CSIIc}
\text{var}\left(\It_{j,k;T}^{(r,r');R}\right)&= \text{var}\left( (2M+1)^{-1} \sum_{s=-M}^{M} I_{j,k;T}^{(r+s,r'+s);R} \right) \nonumber\\
&= (2M+1)^{-2}\sum_{s=-M}^{M}\sum_{s'=-M}^{M} \cov\left(I_{j,k;T}^{(r+s,r'+s);R},I_{j,k;T}^{(r+s',r'+s');R}\right) \nonumber\\
&= (2M+1)^{-2}\sum_{s=-M}^{M}\sum_{\eta=-M-s}^{M-s} \cov\left(I_{j,k;T}^{(r+s,r'+s);R},I_{j,k;T}^{(r+s+\eta,r'+s+\eta);R}\right)
\end{align}
where we have let $\eta=s'-s$.

Let us now take
\begin{align}\label{eq:covtau}
\cov\left(I_{j,k;T}^{(r+s,r'+s);R},I_{j,k;T}^{(r+s+\eta,r'+s+\eta);R}\right) &= \E\left[d_{j,k}^{ r+s} d_{j,k}^{ r'+s}d_{j,k}^{ r+s+\eta} d_{j,k}^{ r'+s+\eta}\right] \nonumber \\
& \qquad - \E\left[d_{j,k}^{ r+s} d_{j,k}^{ r'+s}\right]\E\left[d_{j,k}^{ r+s+\eta} d_{j,k}^{ r'+s+\eta}\right] \nonumber \\
&= \E\left[d_{j,k}^{ r+s} d_{j,k}^{ r+s+\eta}\right]\E\left[d_{j,k}^{ r'+s} d_{j,k}^{ r'+s+\eta}\right] \nonumber \\
& \qquad + \E\left[d_{j,k}^{ r+s} d_{j,k}^{ r'+s+\eta}\right]\E\left[d_{j,k}^{ r'+s} d_{j,k}^{ r+s+\eta}\right],
\end{align}
where we have use Isserlis' theorem \citep{isserlis:1918} in the last equality.

Using the expectation of the wavelet cross-periodogram in Proposition \ref{prop:EIc} and the result in Lemma~\ref{lem:lipr}, we can rewrite the terms in equation~\eqref{eq:covtau} as follows.

\begin{align*}
\E\left[d_{j,k}^{ r+s} d_{j,k}^{ r+s+\eta}\right] &= \E\left[I_{j,k;T}^{(r+s,r+s+\eta);R}\right] \\
&= \sum_{l=1}^{J}A_{j,l} S_{l}\left(\frac{k}{T},\frac{r+s}{R},\frac{r+s+\eta}{R}\right) + \ORD(2^j T^{-1}) + \ORD(R^{-1}), \\
&= \sum_{l=1}^{J}\left\{A_{j,l} S_{l}\left(\frac{k}{T},\frac{r}{R},\frac{r+\eta}{R}\right) + \ORD(|s| B^\prime_l R^{-1})\right\} + \ORD(2^j T^{-1}) + \ORD(R^{-1}), \\
&= \sum_{l=1}^{J} A_{j,l} S_{l}\left(\frac{k}{T},\frac{r}{R},\frac{r+\eta}{R}\right)+ \ORD(2^j T^{-1}) + \ORD(|s| R^{-1}),
\end{align*}
where we have also used the Lipschitz constants' property $\sum_l 2^l B^\prime_l <\infty$  that yields $\sum_l B^\prime_l A_{j,l}<\infty$.

Using the same steps, we also have
\begin{align*}
\E\left[d_{j,k}^{ r'+s} d_{j,k}^{ r'+s+\eta}\right] &= \E\left[I_{j,k;T}^{(r'+s,r'+s+\eta);R}\right] \\
&= \sum_{l=1}^{J} A_{j,l} S_{l}\left(\frac{k}{T},\frac{r'}{R},\frac{r'+\eta}{R}\right)+ \ORD(2^j T^{-1}) + \ORD(|s| R^{-1}).
\end{align*}

Similarly, from the expectation of the wavelet cross-periodogram in Proposition \ref{prop:EIc} and the cross-spectrum Lipschitz continuity in replicate time, it can be shown that
\begin{align*}
\E\left[d_{j,k}^{ r+s} d_{j,k}^{ r'+s+\uptau}\right] &= \E\left[I_{j,k;T}^{(r+s,r'+s+\uptau);R}\right] \\
&= \sum_{l=1}^{J} A_{j,l} S_{l}\left(\frac{k}{T},\frac{r+s}{R},\frac{r'+s+\eta}{R}\right)+ \ORD(2^j T^{-1}) + \ORD(R^{-1}),\\
&= \sum_{l=1}^{J} A_{j,l} S_{l}\left(\frac{k}{T},\frac{r+s}{R},\frac{r+s+\eta}{R}\right)+ \ORD(2^j T^{-1}) + \ORD(|r-r'| R^{-1}),
\end{align*}
and using the same arguments as above, one obtains
\begin{align*}
\E\left[d_{j,k}^{ r+s} d_{j,k}^{ r'+s+\eta}\right] &= \sum_{l=1}^{J} A_{j,l} S_{l}\left(\frac{k}{T},\frac{r}{R},\frac{r+\eta}{R}\right)+ \ORD(2^j T^{-1}) + \ORD(|s| R^{-1}),
\end{align*}
under the condition that the replicates $r, r'$ are such that $|r-r'|<\infty$.

Similarly,
\begin{align*}
\E\left[d_{j,k}^{ r'+s} d_{j,k}^{ r+s+\eta}\right] &= \sum_{l=1}^{J} A_{j,l} S_{l}\left(\frac{k}{T},\frac{r'}{R},\frac{r'+\eta}{R}\right)+ \ORD(2^j T^{-1}) + \ORD(|s| R^{-1}).
\end{align*}

Thus we can write the covariance in equation~\eqref{eq:covtau} as follows
\begin{align*}
\cov &\left(I_{j,k;T}^{(r+s,r'+s);R},I_{j,k;T}^{(r+s+\eta,r'+s+\eta);R}\right) \\
&= \E\left[d_{j,k}^{ r+s} d_{j,k}^{ r+s+\eta}\right]\E\left[d_{j,k}^{ r'+s} d_{j,k}^{ r'+s+\eta}\right] + \E\left[d_{j,k}^{ r+s} d_{j,k}^{ r'+s+\eta}\right]\E\left[d_{j,k}^{ r'+s} d_{j,k}^{ r+s+\eta}\right] \\
&= 2\left(\sum_{l=1}^{J}A_{j,l}S_{l}\left(\frac{k}{T},\frac{r}{R},\frac{r+\eta}{R}\right) + \ORD(2^j T^{-1}) + \ORD(|s|R^{-1}) \right) \\
& \qquad \times \left( \sum_{l=1}^{J}A_{j,l}S_{l}\left(\frac{k}{T},\frac{r'}{R},\frac{r'+\eta}{R}\right) + \ORD(2^j T^{-1}) + \ORD(|s|R^{-1}) \right).
\end{align*}

Using this expression in equation~\eqref{eq:CSIIc}, we obtain the variance of the replicate-smoothed wavelet cross-periodogram to be
\begin{align*}
\text{var}\left(\It_{j,k;T}^{(r,r');R}\right)&= \ORD(M^{-2})\sum_{s=-M}^{M}\sum_{\eta=-M-s}^{M-s} \left[ \sum_{l=1}^{J}A_{j,l}S_{l}\left(\frac{k}{T},\frac{r}{R},\frac{r+\eta}{R}\right)
\sum_{l=1}^{J}A_{j,l}S_{l}\left(\frac{k}{T},\frac{r'}{R},\frac{r'+\eta}{R}\right) \right. \\
& \qquad \qquad \qquad \quad \left. + \left(
\sum_{l=1}^{J}A_{j,l}S_{l}\left(\frac{k}{T},\frac{r}{R},\frac{r+\eta}{R}\right) \ORD(2^j T^{-1}) \right)\right. \\
& \qquad \qquad \qquad \quad \left. + \left(
\sum_{l=1}^{J}A_{j,l}S_{l}\left(\frac{k}{T},\frac{r}{R},\frac{r+\eta}{R}\right) \ORD(|s|R^{-1}) \right) \right. \\
& \qquad \qquad \qquad \quad \left. + \left(
\sum_{l=1}^{J}A_{j,l}S_{l}\left(\frac{k}{T},\frac{r'}{R},\frac{r'+\eta}{R}\right) \ORD(2^j T^{-1}) \right)\right. \\
& \qquad \qquad \qquad \quad \left. + \left(
\sum_{l=1}^{J}A_{j,l}S_{l}\left(\frac{k}{T},\frac{r'}{R},\frac{r'+\eta}{R}\right) \ORD(|s|R^{-1}) \right) \right. \\
& \qquad \qquad \qquad \quad \left. + \ORD(|s|^2R^{-2})+ \ORD(2^j |s| (TR)^{-1})+\ORD(2^{2j} T^{-2})
\vphantom{\sum_{l=1}^{J}S_{l}\left(\frac{k}{T},\frac{r}{R}\right)} \right].
\end{align*}

Under the assumption $\sup_{z,\nu,\uptau}\sum_{\eta\in\Z} \left| c(z,\nu,\nu+\frac{\eta}{R};\uptau) \right|=\ORD(1)$, we obtain for time $k$, replicate $r$ and replicate-lag $\eta$ that $\sum_{\eta\in\Z} \left| \sum_{l}A_{j,l}S_{l}\left(\frac{k}{T},\frac{r}{R},\frac{r+\eta}{R}\right) \right| =\ORD(2^j)$, since using the definition of the $A$ matrix and of the local cross-covariance we have
\begin{align*}
\sum_{\eta\in\Z} \left| \sum_{l}A_{j,l}S_{l}\left(\frac{k}{T},\frac{r}{R},\frac{r+\eta}{R}\right) \right| &=\sum_{\eta\in\Z} \left|  \sum_{\uptau\in\Z} \left( \sum_l S_{l}\left(\frac{k}{T},\frac{r}{R},\frac{r+\eta}{R}\right) \Psi_{l}(\uptau) \right) \Psi_{j}(\uptau)  \right|,\\
&= \sum_{\eta} \left|  \sum_{\uptau}  c\left(\frac{k}{T},\frac{r}{R},\frac{r+\eta}{R}; \uptau\right)\Psi_{j}(\uptau)  \right|, \\
&\leq \sum_{\uptau} \left( \left|\Psi_{j}(\uptau)\right|  \sum_{\eta} \left|c\left(\frac{k}{T},\frac{r}{R},\frac{r+\eta}{R}; \uptau\right) \right| \right), \\
&=\ORD(1) \sum_{\uptau}\left|\Psi_{j}(\uptau)\right|= \ORD(2^j),
\end{align*}
where we used the triangle inequality and the autocorrelation wavelet property $\sum_{\uptau}\left|\Psi_{j}(\uptau)\right|= \ORD(2^j)$.

Using Lemma~\ref{lem:ASordc} and the property above, we readily obtain that the first term in the variance equation is of order $\ORD(2^{2j}M^{-1})$; the second and fourth terms, are both of order $\ORD(2^{2j}(MT)^{-1})$; and the third and fifth terms are both of order $\ORD(2^j R^{-1})$.

Now considering the final order terms,
$$
\ORD(M^{-2})\sum_{s=-M}^{M}\sum_{\eta=-M-s}^{M-s} \left[\ORD(|s|^2R^{-2})+ \ORD(2^j |s| (TR)^{-1})+\ORD(2^{2j} T^{-2})\right],$$ we have
$$
\ORD(M^2R^{-2})+\ORD(2^j M(TR)^{-1})+ \ORD(2^{2j}T^{-2}).
$$

Putting these results together in the variance equation, we obtain
\begin{align*}
\text{var}\left(\It_{j,k;T}^{(r,r');R}\right)&=
\ORD(2^{2j}M^{-1})+\ORD(2^{2j}(MT)^{-1})+\ORD(2^j R^{-1})\\
& \qquad + \ORD(M^2R^{-2})+\ORD(2^j M(TR)^{-1})+ \ORD(2^{2j}T^{-2}),\\
&= \ORD(2^{2j}M^{-1})+ \ORD(2^j R^{-1})+ \ORD(M^2R^{-2}).
\end{align*}
\end{proof}

\subsection{Proof of Proposition \ref{prop:rhoc}}\label{app:proofs:rhoc}
\begin{proof}[\unskip\nopunct]
It follows directly from the continuous mapping theorem \citep{billingsley1999convergence} and the consistency results for the corrected replicate-smoothed periodogram in~\eqref{eq:corrspecc}, \i.e. $\hat{S}_{j}(z,\nu, \nu') \stackrel{P}{\to} {S}_{j}(z,\nu, \nu')$ as $M, T, R \to \infty$ and $M/R \to 0$.
\end{proof}

\end{appendices}

\end{spacing}

\newpage
\begin{spacing}{1}

\begin{center}
\vspace*{2em}
  {\LARGE Multiscale modelling of replicated nonstationary time series:\\[.2cm] Supplementary Material}\\[.6cm]
  {\large Jonathan Embleton\footnote[1]{Corresponding author: je687@york.ac.uk}$^{1}$, Marina I. Knight$^1$, and Hernando Ombao$^2$\\[.4cm]}
  {\large ${}^1$Department of Mathematics, University of York, UK\\
  ${}^2$King Abdullah University of Science and Technology (KAUST), Saudi Arabia\\[.7cm] \large}
{\large \today\\[1.5cm]}
\end{center}

\setcounter{page}{1}
\setcounter{equation}{57}
\setcounter{figure}{19}
\setcounter{table}{2}
\setcounter{lemma}{0}

\begin{appendix}
\section{Experimental data description and overview of implemented methodology}\label{supp:exper}
Each trial (replicate) consists of $T = 2048$ time points, corresponding to approximately 2 seconds of data. The design of the experiment splits each trial into four time blocks of 512 milliseconds each, as follows. For the first block the macaque fixated on a screen; a picture (one of four) was then presented on the screen for the next time block; this was followed by an empty screen for the next interval; for the last 512 milliseconds the macaque was presented with a picture of four doors, one of which associated with the picture visual from the second time block. The macaque's task was to select the correct door using a joystick. Correct and incorrect choices were signified via a visual cue and a juice reward was given each time a correct selection was made. The macaque had to learn the associations through repeated trials. The data has been grouped into sets of `correct' and `incorrect' responses, in order to investigate how the contributions of the Hc and NAc to the learning process differ between groups \citep{Gorrombao:monkeydepend}. The plots of the incorrect responses appear below, as well as their corresponding proposed spectral estimates.

For both the correct and incorrect sets of the hippocampus (Hc) trial data, we compute the wavelet periodograms using non-decimated discrete wavelets built by means of Daubechies Least Asymmetric wavelet family with 10 vanishing moments. Similarly, for both sets of the nucleus accumbens (NAc) trial data we again choose Daubechies Least Asymmetric wavelet family, but however we now opt for a coarser choice of wavelet with 6 vanishing moments to reflect the behaviour of the signal. In accordance with our simulation study findings, to obtain an asymptotically unbiased and consistent estimator for the replicate evolutionary wavelet spectrum (REWS), we smooth the wavelet periodograms using a local averaging window over 21 replicates ($M = 10$ neighbouring replicates) and then correct for bias. For completeness, we also run the analysis to include a time-smoothing step before locally averaging across replicates, as this was shown to lead to better performance (see the simulation study in Section \ref{sec:sims}). For comparison, we additionally report the LSW estimator embedding averaging over all replicates. Note that averaging over all replicates here refers to the averaging over the first 241 and 256 correct and incorrect response trials, respectively. To explore the consistency of the results, the analysis has also been repeated using wavelets with different vanishing moments and varying smoothing windows across the replicates, which yielded extremely similar results to those reported here.

\begin{figure}[h!]
\begin{center}
\hspace*{-2em}
\resizebox{!}{0.6\textwidth}{\includegraphics{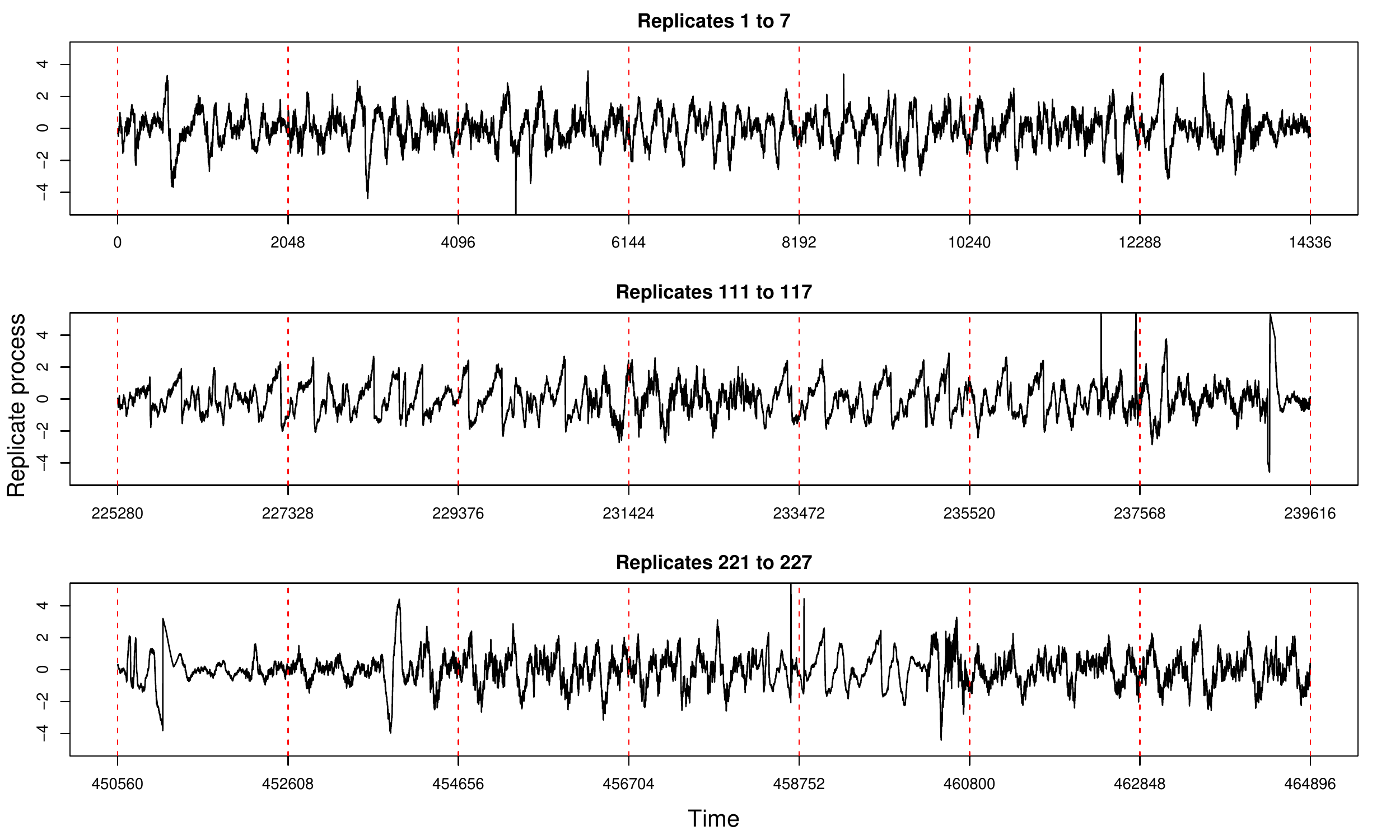}}
\caption{\small Concatenated series of the hippocampus (Hc) data in the incorrect response trials.}
\label{fig:rb7_It_Hc}
\end{center}
\end{figure}

\begin{figure}[h!]
\begin{center}
\hspace*{-2em}
\resizebox{!}{0.6\textwidth}{\includegraphics{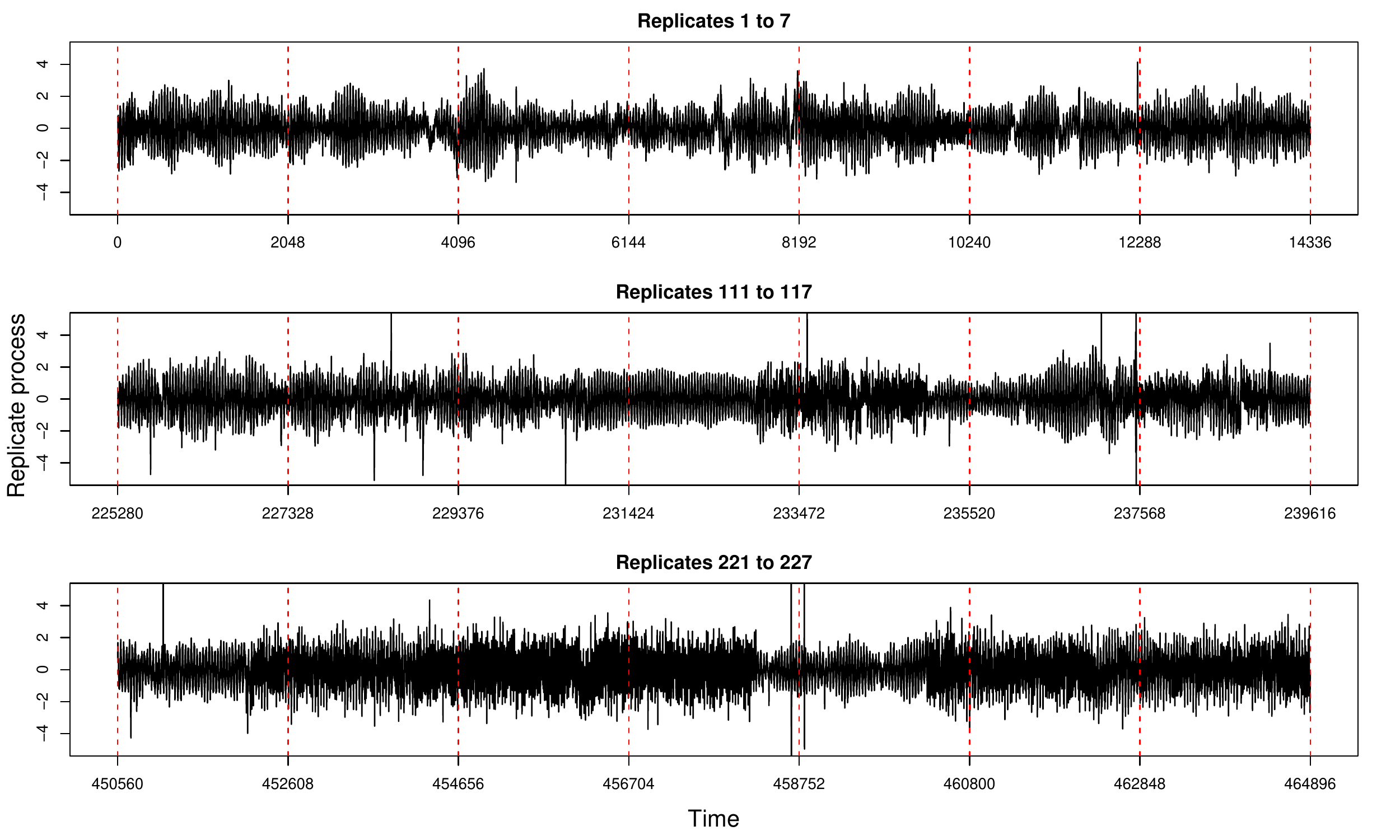}}
\caption{\small Concatenated series of the nucleus accumbens (NAc) data in the incorrect response trials.}
\label{fig:rb7_It_NAc}
\end{center}
\end{figure}

\begin{figure}[h!]
\centering
\resizebox{!}{0.6\textwidth}{\includegraphics{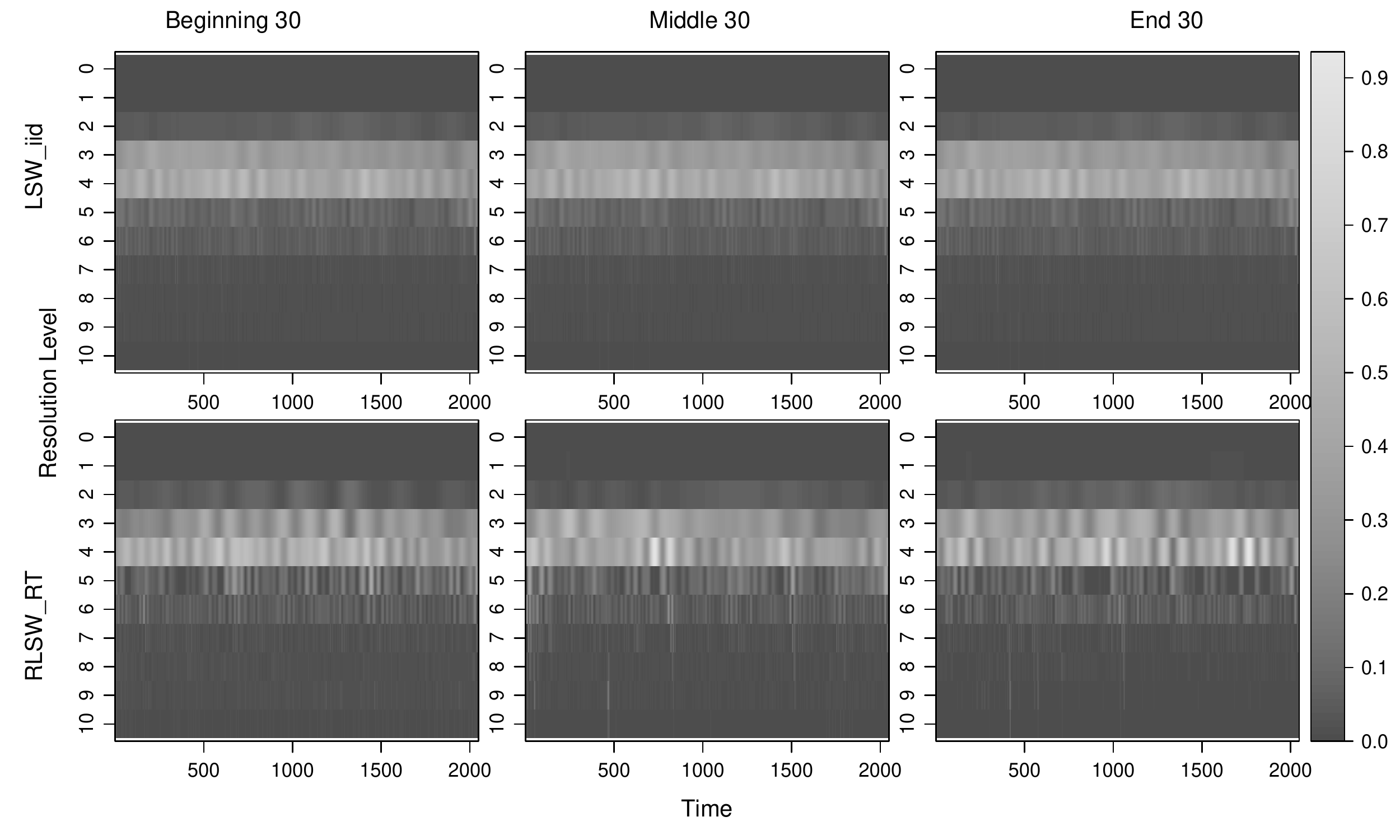}}
\caption{\small Time-scale hippocampus (Hc) plots for the incorrect trials. Spectral estimates are shown for the average over 30 replicates in the beginning, middle and end of the experiment. \textit{Top}: estimates from the LSW method averaged over all replicates; \textit{Bottom}: RLSW method with smoothing over time and replicates.}
\label{fig:coef_ItHc_M10D10}
\end{figure}

\begin{figure}[h!]
\centering
\resizebox{!}{0.6\textwidth}{\includegraphics{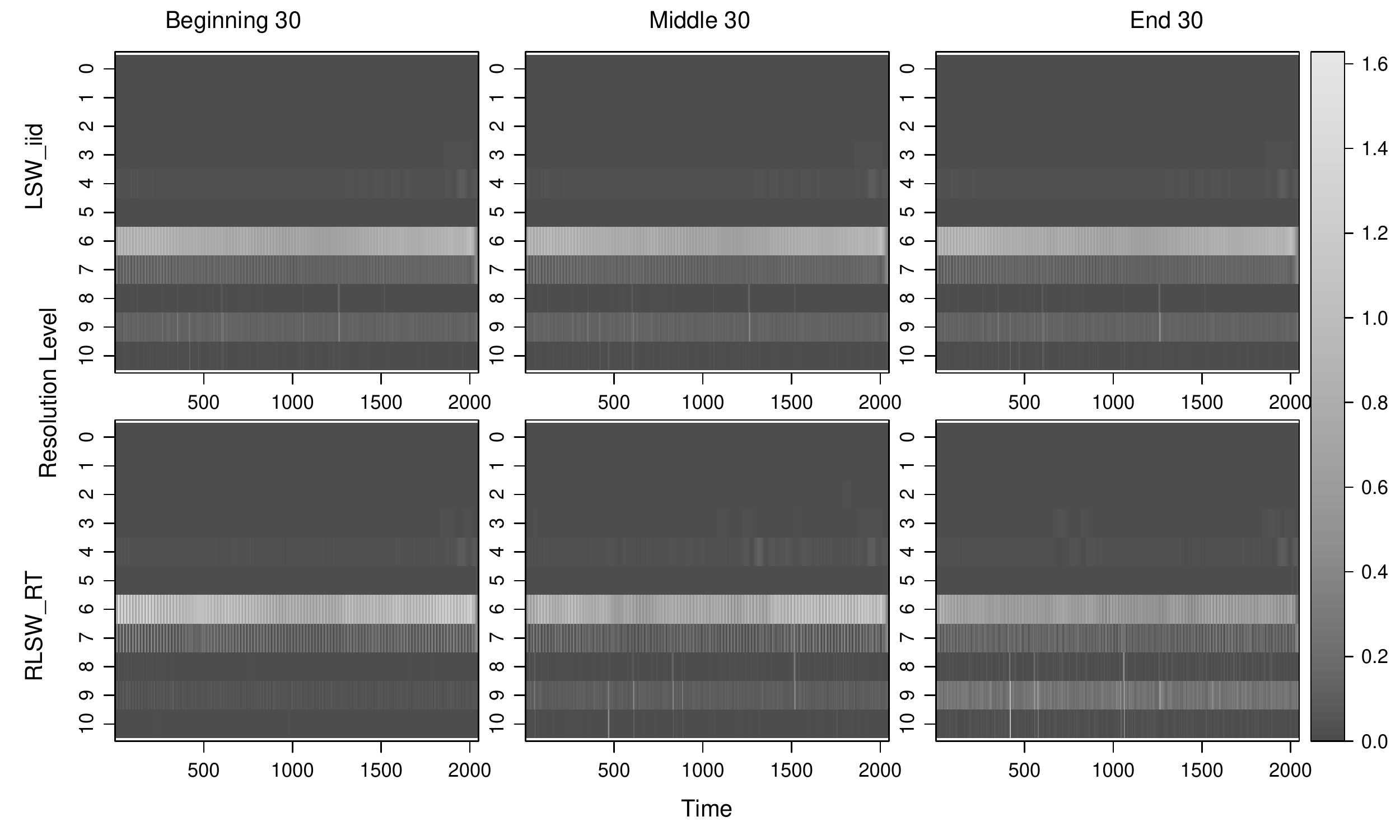}}
\caption{\small Time-scale nucleus accumbens (NAc) plots for the incorrect trials. Spectral estimates are shown for the average over 30 replicates in the beginning, middle and end of the experiment. \textit{Top}: estimates from the LSW method averaged over all replicates; \textit{Bottom}: RLSW method with smoothing over time and replicates.}
\label{fig:coef_ItNAc_M10D6}
\end{figure}

\clearpage
\section{Further simulation evidence for Section~\ref{sec:sims}}\label{supp:extra}

\subsection{Supporting evidence for the simulation study of Section~\ref{sec:sims}}\label{supp:extradetail}
For $R = 256$ replicates and $T=256=2^8$, we generated a RLSW process with spectral structure at scales 5 and 6, defined as
\begin{equation*}
S_{j}(z,\nu) =
  \begin{cases}
    4(1-\nu)\cos^{2}\left(\pi z \right), 		& \text{for } j=J(T)-3, z\in(65/256,1), \nu \in(0,1) \\
    4\cos^{2}\left(2\pi z + 5\nu \right), 		& \text{for } j=J(T)-2, z\in(0,128/256), \nu \in(0,1) \\
    0, 		& \text{otherwise}.
  \end{cases}
  \label{eq:specsim3}
\end{equation*}

The histogram in Figure~\ref{fig:hist_SIM3proc} illustrates that the expected (theoretical) asymptotic behaviour of the proposed estimator holds in practice.
\begin{figure}[h!]
\begin{center}
\resizebox{!}{0.4\textwidth}{\includegraphics{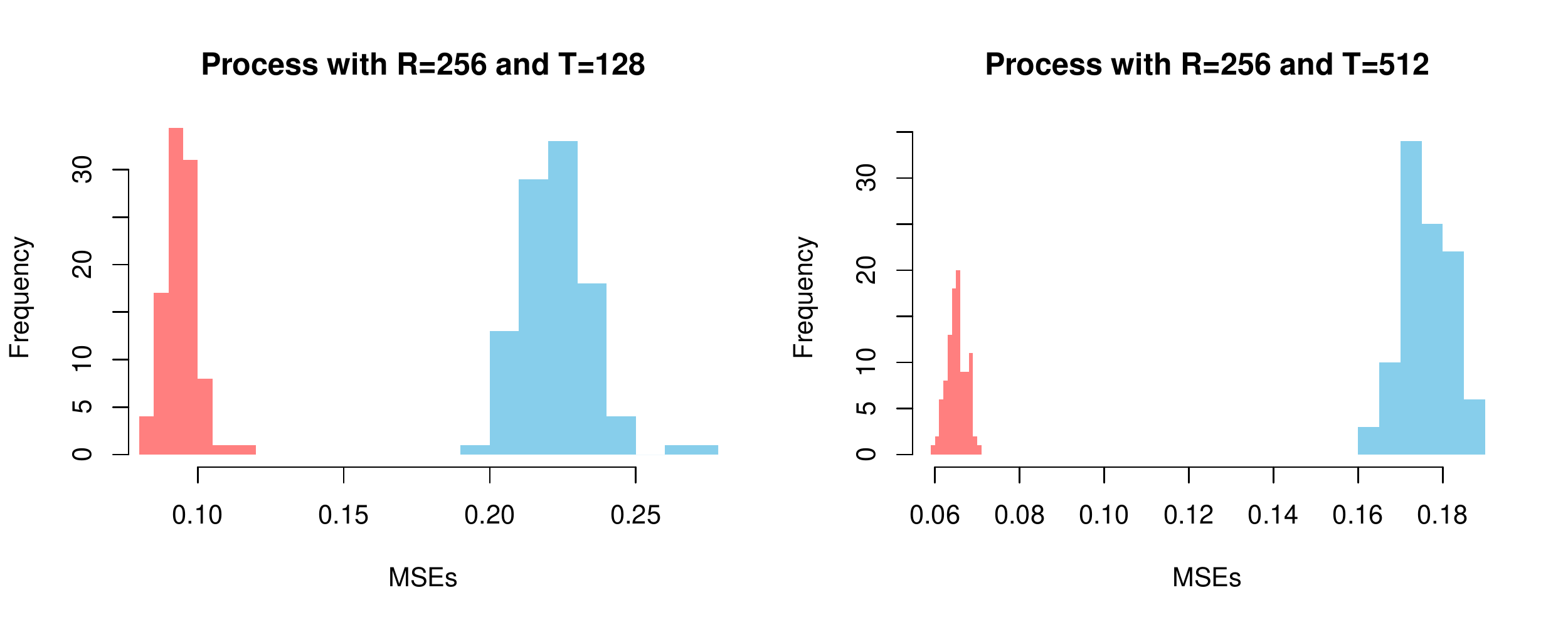}}
\end{center}
\vspace{-2em}
\caption{\small Histograms of the MSEs on the estimates from the RLSW model over 100 runs for the simulation in Section~\ref{sec:sims}. \textit{Red:} smoothing over replicates with $M=12$; \textit{Blue:} smoothing over replicates with $M=4$.}
\label{fig:hist_SIM3proc}
\end{figure}

\subsection{Further simulation studies}\label{supp:furthersims}

\noindent{\bf Simulation 1}.
We simulate a RLSW process consisting of $R = 128$ replicates, each of length $T=256=2^8$ and whose REWS, illustrated in Figure~\ref{fig:wavplots_usual}, evolves slowly over both rescaled time and replicates, as follows
\begin{equation}
S_{j}(z, \nu) =
  \begin{cases}
   4\nu \sin^{2}\left(2\pi z(1 + 2\nu) \right), 		& \text{for } j=J(T)-4, z\in(0,1), \nu \in(0,1) \\
    0, 		& \text{otherwise},
  \end{cases}
  \label{eq:specsim1}
\end{equation}
recalling that $z = k/T$ and $\nu = r/R$ for $k=0,\ldots,T-1$ and $r=0,\ldots,R-1$. The spectral characteristics thus appear at scale $j = 8-4=4$.

The periodicity and magnitude of the sine wave evolve slowly over the replicates in such a way that the spectral characteristics of neighbouring replicates do not look too dissimilar whilst there is a noticeable difference between replicates further apart. One concatenated realisation of the meta-process with the specified spectral structure in~\eqref{eq:specsim1}, viewed as a series of length $RT$, can be seen in Figure~\ref{fig:reptsBY7_usualproc}. 

\begin{figure}[h!]
\begin{center}
\resizebox{!}{0.4\textwidth}{\includegraphics{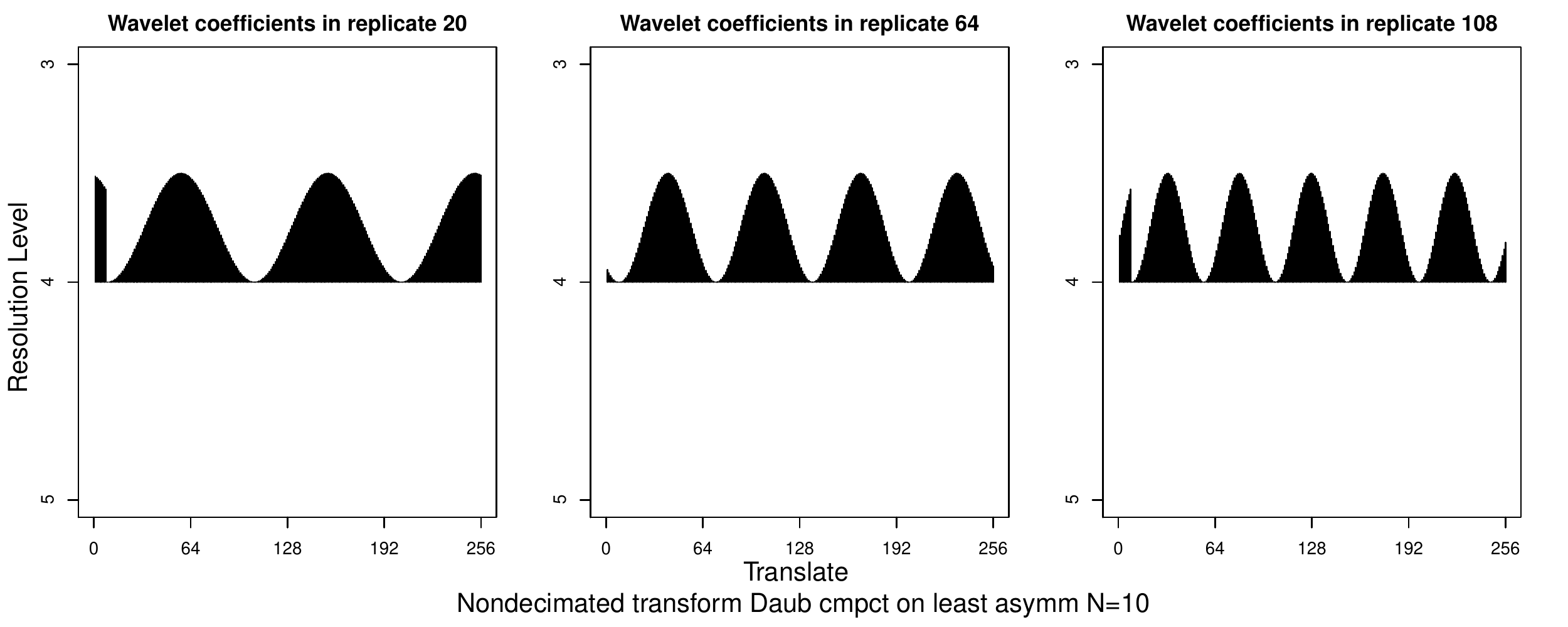}}
\end{center}
\vspace{-1em}
\caption{\small True wavelet spectra for replicates 20, 64 and 108 of Simulation 1.}
\label{fig:wavplots_usual}
\end{figure}

\begin{figure}[h!]
\begin{center}
\resizebox{!}{0.6\textwidth}{\includegraphics{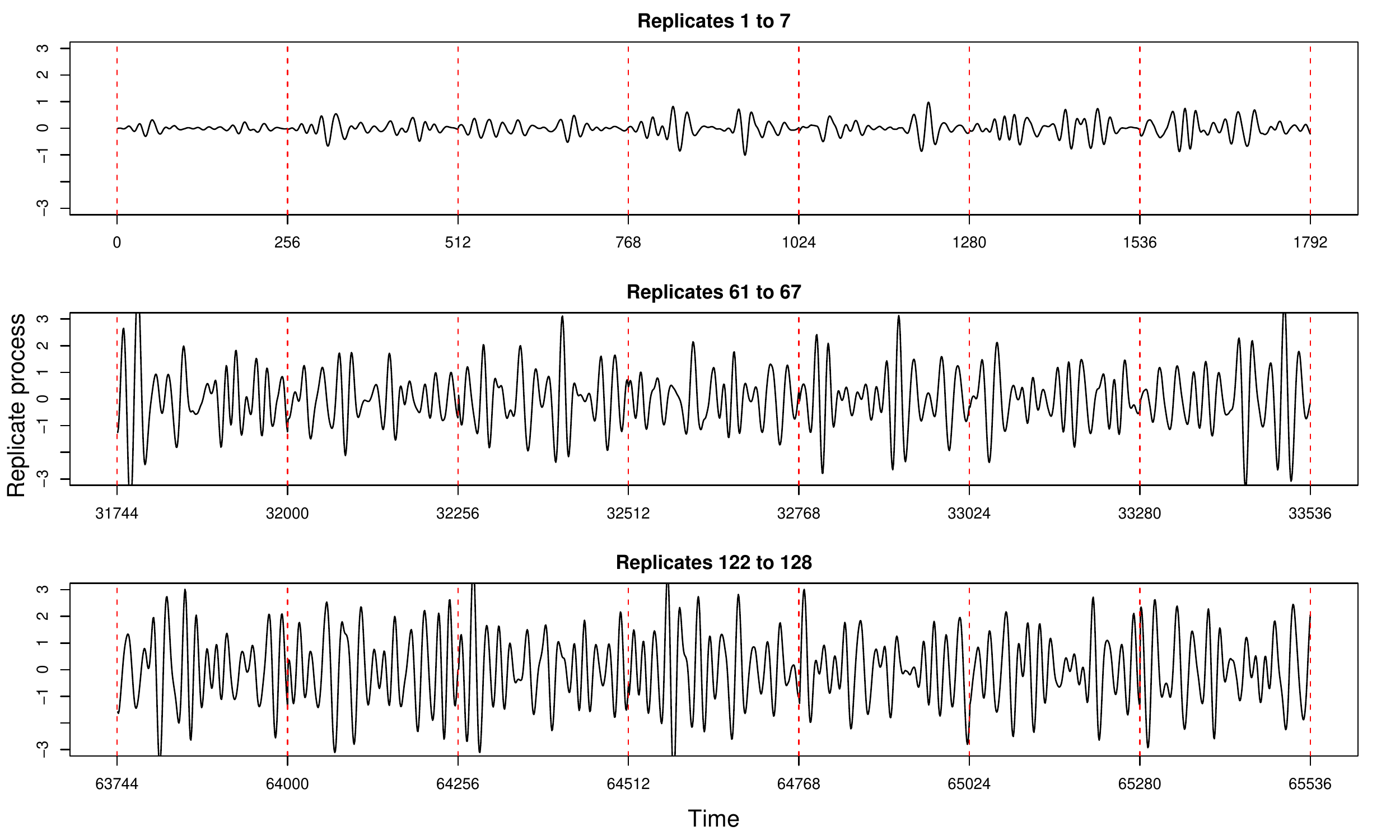}}
\end{center}
\vspace{-2em}
\caption{\small Realisation of a RLSW process with spectra defined in equation~\eqref{eq:specsim1}.}
\label{fig:reptsBY7_usualproc}
\end{figure}

We display in Figure \ref{fig:coef_usual_128_256} the true spectra and the average spectral estimates for replicates 20, 64 and 108. The non-decimated wavelet transform was computed using discrete wavelets built by means of Daubechies Least Asymmetric family with 10 vanishing moments and the local averaging for our RLSW$_1$ method was carried out using $M=4$, corresponding to a window of $9$ replicates (numerical MSE results in Table~\ref{tab:msesim1} highlight that we chose to visually present some of our least performant results). Figure \ref{fig:coef_usual_128_256} displays the danger of neglecting the possibility of an existing evolutionary behaviour over replicates (see e.g. level $4$ in the top row plots of the true spectrum), conducive to either under or over-estimation (see the middle row plots). The bottom row plots show that the RLSW$_1$ estimates do reflect the evolution over replicates. To further support this, Figure \ref{fig:lev4_usual_128_256} takes a closer look at the evolutionary behaviour of the spectral quantities over the time and replicates in level 4.

\begin{figure}[h!]
\begin{center}
\resizebox{!}{0.6\textwidth}{\includegraphics{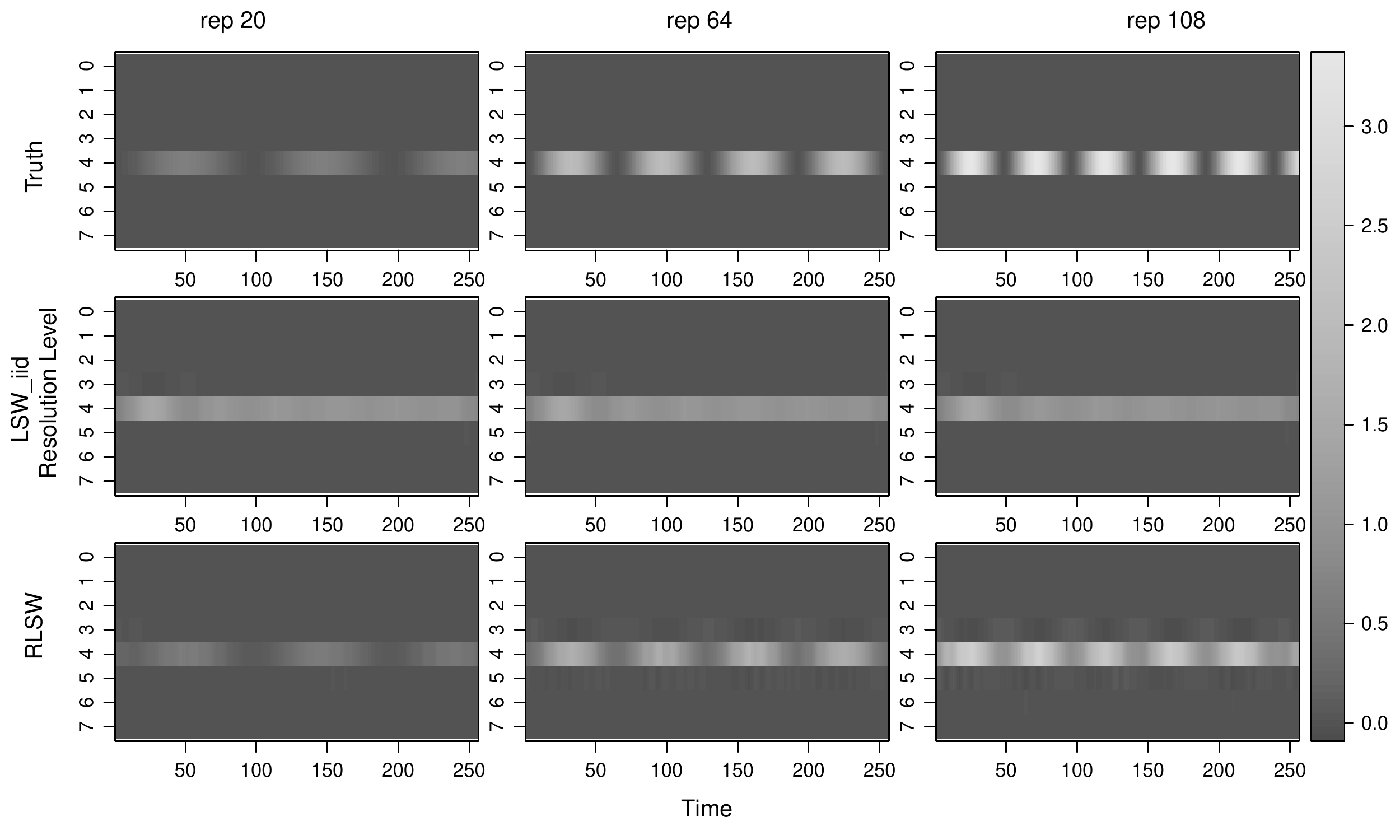}}
\end{center}
\vspace{-2em}
\caption{\small Simulation 1 time-scale plots for replicates 20, 64 and 108, respectively the first, second and third columns. Estimates are averaged  over 100 realisations. \textit{Top}: true spectra; \textit{Middle}: estimates from the LSW method averaged over all replicates; \textit{Bottom}: estimates using RLSW($_1$).}
\label{fig:coef_usual_128_256}
\end{figure}

\begin{figure}[h!]
\begin{center}
\resizebox{!}{0.6\textwidth}{\includegraphics{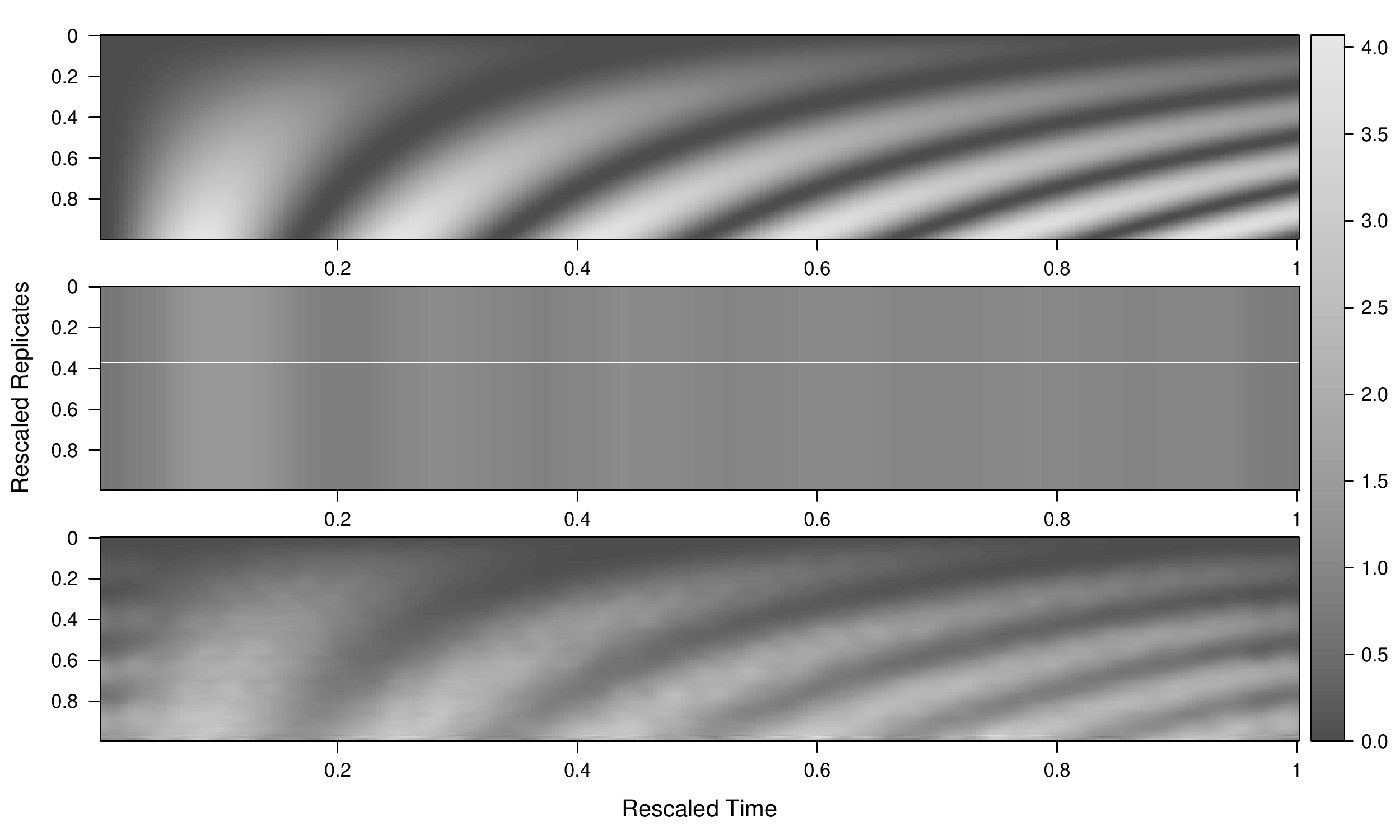}}
\end{center}
\vspace{-2em}
\caption{\small Simulation 1 time-replicate spectral plots in level 4. Estimates are averaged over 100 realisations. \textit{Top}: true spectra; \textit{Middle}:  estimated spectra from the LSW method averaged over all replicates; \textit{Bottom}: estimates using RLSW($_1$).}
\label{fig:lev4_usual_128_256}
\end{figure}

The MSE and squared bias results in Table \ref{tab:msesim1}, highlight that for this example, which adheres well to the RLSW assumptions, the MSEs for the LSW model are higher than those computed for the RLSW model. 
The benefit of taking a local smoothing approach over replicates results in spectral estimates with lower bias and MSE, although it is worth pointing out that taking a local smoothing approach over both time and replicates, while yielding lower MSEs, does increase the bias. Just as for the simulation in the main body of the paper, notice the RLSW methodology performance improves with the replicate local averaging window length increase ($2M + 1$).

Figure~\ref{fig:hist_usualproc} provides a visualisation on how the RLSW model performed over the 100 simulations via histograms of the simulation-specific MSE. The histograms highlight not only how the increase in $M$ improves performance but also how the increase in $R$ and $T$ reduces the MSEs, thus demonstrating the expected asymptotic behaviour of our smoothed estimator $\hat{S}_{j}\left(z,\nu\right)$.

\begin{figure}[h!]
\begin{center}
\resizebox{!}{0.4\textwidth}{\includegraphics{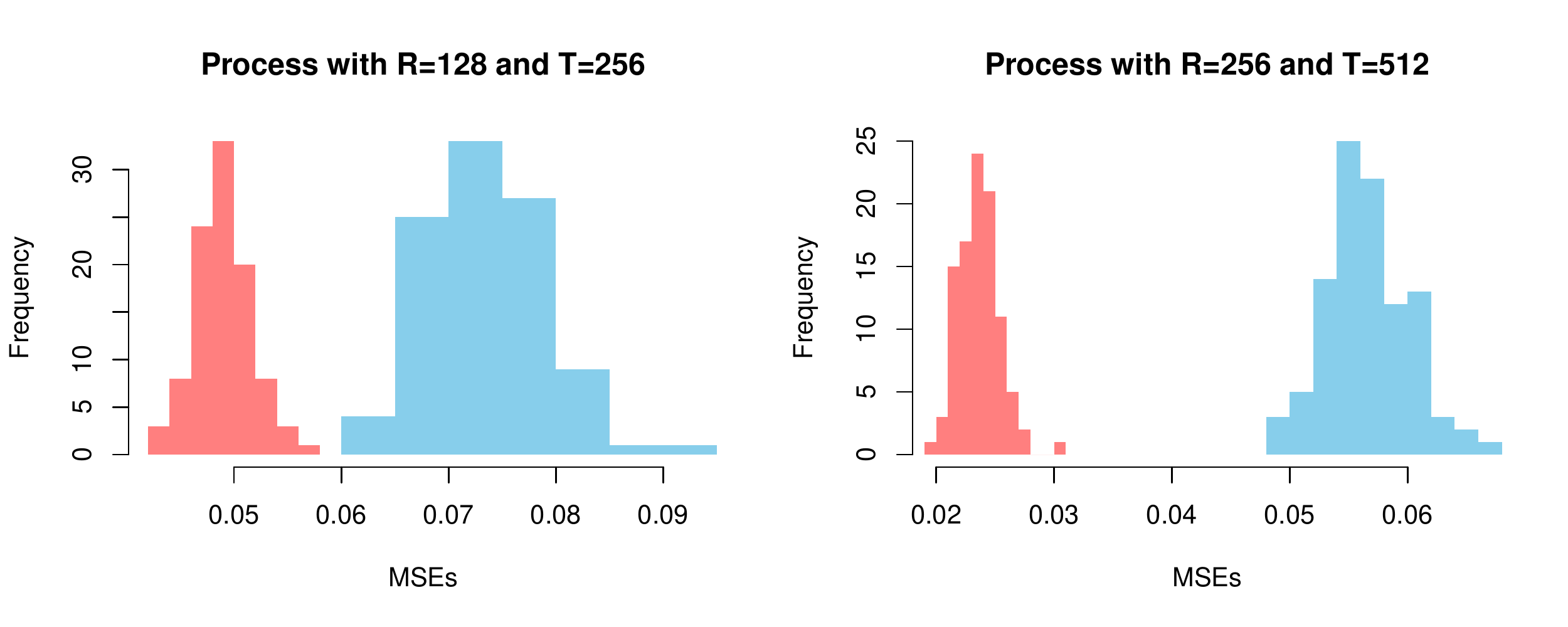}}
\end{center}
\vspace{-2em}
\caption{\small Histograms of the MSEs on the estimates from the RLSW model over 100 runs for Simulation 1. Smoothing over replicates with $M=12$ (\textit{red});  with $M=4$ (\textit{blue}). }
\label{fig:hist_usualproc}
\end{figure}

\begin{table}[h!]
\scriptsize
\hspace{-6em}
\begin{tabular}{@{\extracolsep{3pt}}ccccccccc@{}}
\multicolumn{1}{c}{} & \multicolumn{7}{c}{Mean squared errors $\left( \times 100 \right)$} & \multicolumn{1}{c}{}
\\
\cline{1-9}\noalign{\smallskip}
\multicolumn{3}{c}{} & \multicolumn{2}{c}{LSW} & \multicolumn{2}{c}{RLSW\tsbs{1}} & \multicolumn{2}{c}{RLSW\tsbs{2}}\\
\cline{4-5}\cline{6-7}\cline{8-9}\noalign{\smallskip}
R & T & M & mse & bias\tsps{2} & mse & bias\tsps{2} & mse & bias\tsps{2} \\
\firsthline\noalign{\smallskip}
64 & 128 & 4 & 13.21 & 12.68 & 12.92 & 7.91 & 11.59 & 7.97
\smallskip \\
 & & 7 & 12.01 & 11.47 & 10.74 & 7.81 & 9.99 & 7.87
\smallskip \\
 & & 10 & 11.01 & 10.47 & 9.88 & 7.81 & 9.36 & 7.87
\smallskip \\
 & & 12 & 10.48 & 9.94 & 9.51 & 7.79 & 9.08 & 7.84
\smallskip \\
 & 256 & 4 & 11.10 & 10.60 & 7.64 & 2.79 & 6.32 & 2.89
\smallskip \\
 & & 7 & 10.06 & 9.56 & 6.54 & 3.71 & 5.82 & 3.82
\smallskip \\
 & & 10 & 9.21 & 8.71 & 6.83 & 4.84 & 6.35 & 4.94
\smallskip \\
 & & 12 & 8.77 & 8.27 & 7.09 & 5.43 & 6.70 & 5.53
\smallskip \\
128 & 128 & 4 & 13.67 & 13.40 & 13.11 & 7.92 & 11.70 & 7.99
\smallskip \\
 & & 7 & 13.02 & 12.76 & 10.81 & 7.76 & 10.01 & 7.83
\smallskip \\
 & & 10 & 12.40 & 12.14 & 9.80 & 7.67 & 9.26 & 7.75
\smallskip \\
 & & 12 & 12.00 & 11.74 & 9.40 & 7.64 & 8.97 & 7.72
\smallskip \\
 & 256 & 4 & 11.49 & 11.25 & 7.34 & 2.41 & 6.01 & 2.52
\smallskip \\
 &  & 7 & 10.95 & 10.70 & 5.47 & 2.58 & 4.74 & 2.70
\smallskip \\
 & & 10 & 10.41 & 10.17 & 4.94 & 2.91 & 4.47 & 3.03
\smallskip \\
 & & 12 & 10.07 & 9.82 & 4.90 & 3.20 & 4.52 & 3.33
\smallskip \\
\lasthline
\end{tabular}
\quad
\vspace{1em}
\begin{tabular}{@{\extracolsep{3pt}}ccccccccc@{}}
\multicolumn{1}{c}{} & \multicolumn{7}{c}{Mean squared errors $\left( \times 100 \right)$} & \multicolumn{1}{c}{}
\\
\cline{1-9}\noalign{\smallskip}
\multicolumn{3}{c}{} & \multicolumn{2}{c}{LSW} & \multicolumn{2}{c}{RLSW\tsbs{1}} & \multicolumn{2}{c}{RLSW\tsbs{2}}\\
\cline{4-5}\cline{6-7}\cline{8-9}\noalign{\smallskip}
R & T & M & mse & bias\tsps{2} & mse & bias\tsps{2} & mse & bias\tsps{2} \\
\firsthline\noalign{\smallskip}
256 & 512 & 4 & 10.16 & 10.03 & 5.66 & 0.39 & 4.13 & 0.40
\smallskip \\
 & & 7 & 9.91 & 9.78 & 3.54 & 0.40 & 2.64 & 0.43
\smallskip \\
 & & 10 & 9.67 & 9.54 & 2.67 & 0.46 & 2.05 & 0.49
\smallskip \\
 & & 12 & 9.51 & 9.38 & 2.36 & 0.52 & 1.86 & 0.55
\smallskip \\
 & 1024 & 4 & 9.11 & 8.98 & 5.28 & 0.12 & 3.77 & 0.11
\smallskip \\
 & & 7 & 8.88 &  8.75 & 3.17 & 0.12 & 2.28 & 0.12
\smallskip \\
 & & 10 & 8.66 & 8.53 & 2.30 & 0.15 & 1.68 & 0.15
\smallskip \\
 & & 12 & 8.52 & 8.39 & 1.99 & 0.19 & 1.47 & 0.20
\smallskip \\
512 & 512 & 4 & 10.25 & 10.18 & 5.72 & 0.37 & 4.17 & 0.39
\smallskip \\
 & & 7 & 10.12 & 10.05 & 3.55 & 0.36 & 2.64 & 0.38
\smallskip \\
 & & 10 & 9.99 & 9.93 & 2.62 & 0.36 & 1.99 & 0.39
\smallskip \\
 & & 12 & 9.91 & 9.84 & 2.26 & 0.37 & 1.73 & 0.39
\smallskip \\
 & 1024 & 4 & 9.18 & 9.12 & 5.29 & 0.11 & 3.77 & 0.10
\smallskip \\
 &  & 7 & 9.06 & 9.00 & 3.19 & 0.09 & 2.28 & 0.09
\smallskip \\
 & & 10 & 8.95 & 8.89 & 2.29 & 0.09 & 1.65 & 0.09
\smallskip \\
 & & 12 & 8.87 & 8.81 & 1.93 & 0.09 & 1.40 & 0.09
\smallskip \\
\lasthline
\end{tabular}
\caption{\small MSE and squared bias results, averaged over all time-scale points and replicates for Simulation 1 above. `LSW' denotes the classical approach of averaging over the replicates. `RLSW$_1$' denotes our proposed approach using localised replicate smoothing. `RLSW$_2$' denotes our proposed approach using localised time and replicate smoothing. Our proposed methods `RLSW$_1$'and `RLSW$_2$' use a replicate smoothing window of length $(2M+1)$, while the time smoothing window for `LSW' and `RLSW$_2$' is automatically chosen.}
\label{tab:msesim1}
\end{table}

\noindent{\bf Simulation 2}.
We consider a RLSW process consisting of $R = 128$ replicates, each of length $T= 256 = 2^{8}$, and driven by the following REWS
\begin{equation}
S_{j}(z,\nu) =
  \begin{cases}
    \sin^{2}\left(2\pi z + 10\nu \right), 		& \text{for } j=J(T)-1, z\in(0,1), \nu \in(0,1) \\
    0, 		& \text{otherwise}.
  \end{cases}
  \label{eq:specsim2}
\end{equation}
Our spectra in this example are characterised by a squared sine wave in the finest level that experiences a shift from replicate to replicate. This behaviour is displayed in Figure~\ref{fig:wavplots_shift} and then a concatenated realisation of the meta-process appears in Figure~\ref{fig:reptsBY7_shiftproc}. Note that visually the meta-process behaviour does not offer any indication of transitioning through the replicates, despite this actually happening.

\begin{figure}[htbp]
\begin{center}
\resizebox{!}{0.4\textwidth}{\includegraphics{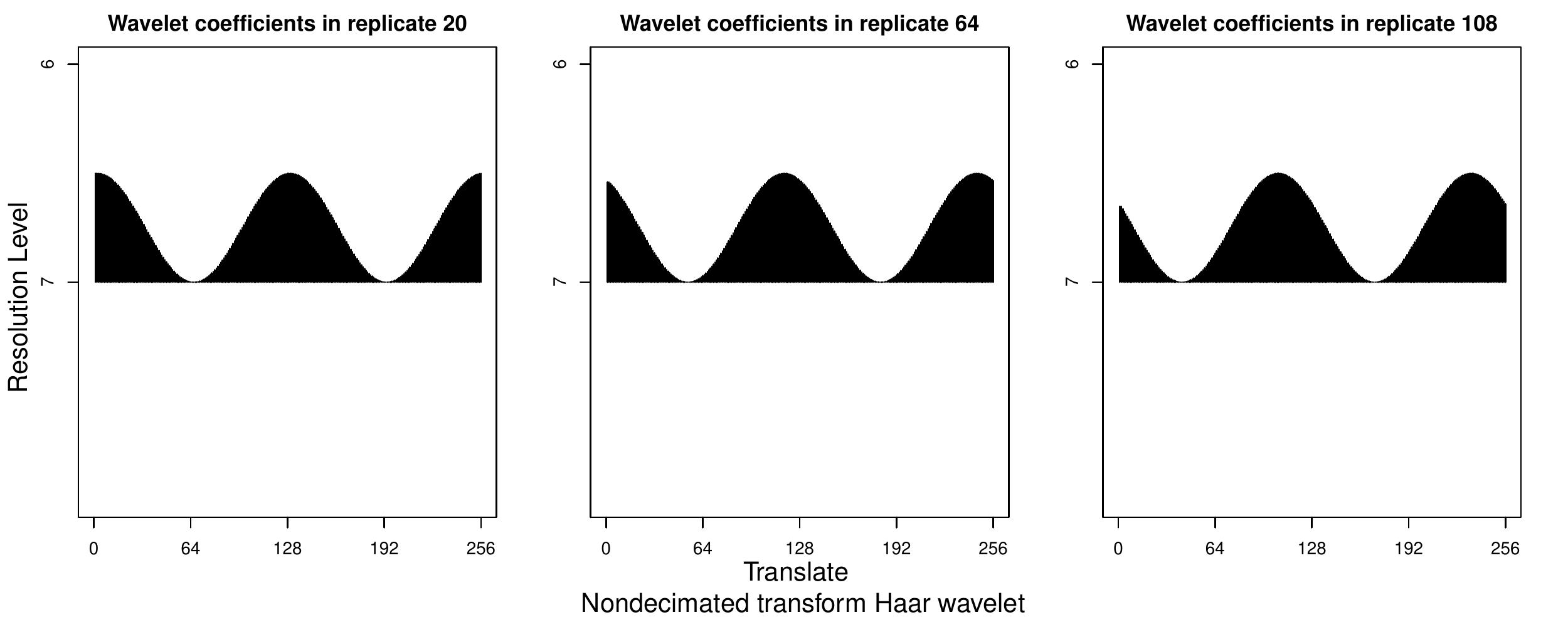}}
\end{center}
\vspace{-1em}
\caption{\small True wavelet spectra for replicates 20, 64 and 108 of Simulation 2.}
\label{fig:wavplots_shift}
\end{figure}

\begin{figure}[htbp]
\begin{center}
\resizebox{!}{0.6\textwidth}{\includegraphics{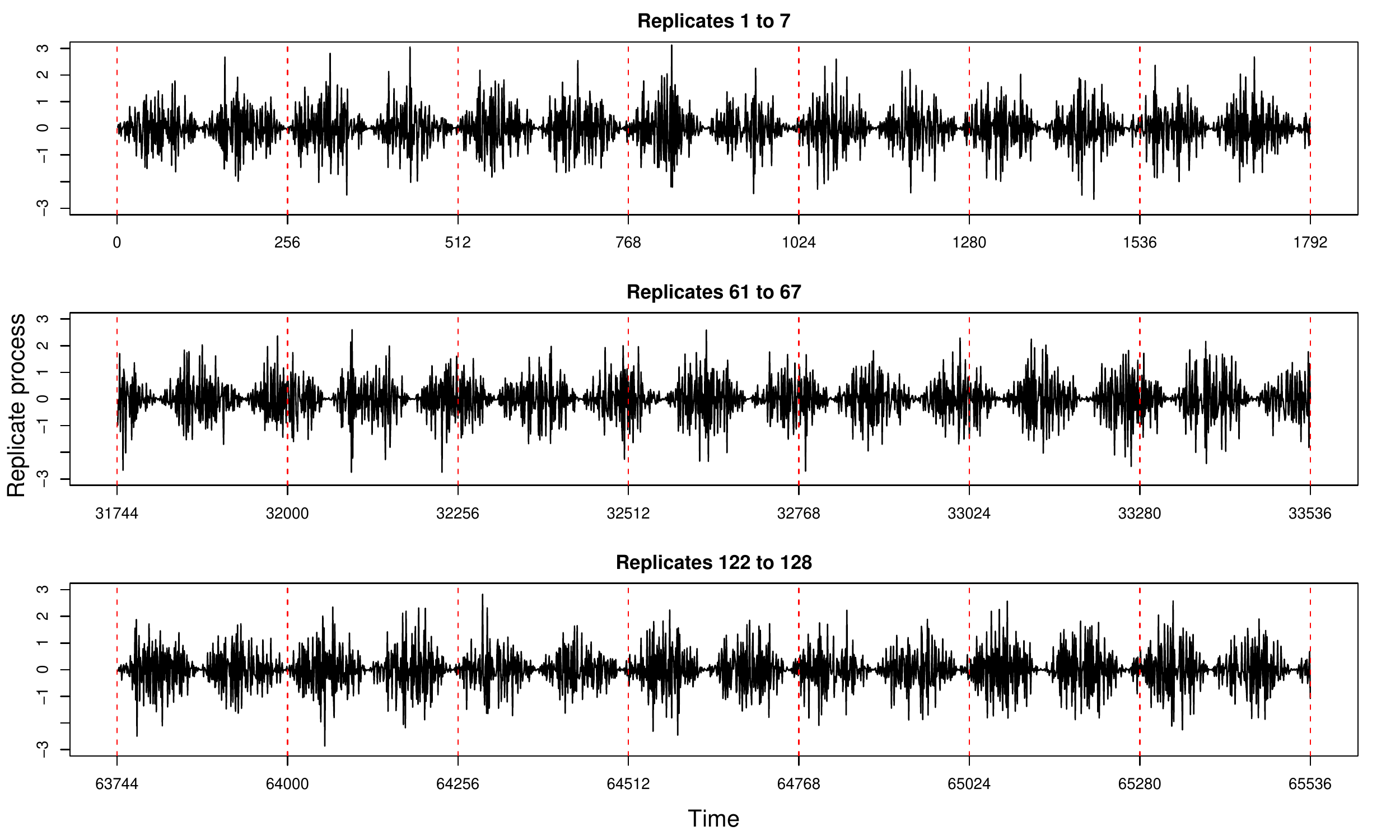}}
\end{center}
\vspace{-2em}
\caption{\small Realisation of a RLSW process with spectra defined in equation~\eqref{eq:specsim2}.}
\label{fig:reptsBY7_shiftproc}
\end{figure}

We obtain estimates for both the LSW model averaged over all replicates and for the RLSW model which adopts the local averaging procedure (over replicate, and over time and replicate dimensions). On visually examining the concatenated process in Figure \ref{fig:reptsBY7_shiftproc}, one may not question the existence of evolutionary behaviour across the replicates. Inspecting our RLSW estimates of the spectral characteristics of this process tells us otherwise. Figure \ref{fig:coef_shiftproc_R128_T256} highlights that the RLSW($_1$) method (bottom row) manages to capture the evolution of the spectra in the finest level. On the other hand, the LSW method (middle row) fails to capture this behaviour, and this is further demonstrated in Figure~\ref{fig:lev7_shiftproc_R128_T256} which shows the spectral estimates across rescaled time and replicate in the finest level. \\

Mean squared errors and squared bias results for Simulation 2 are given in Table~\ref{tab:msesim2}. When comparing models, we notice that for low $R$, performing local averaging over the replicates only (RLSW$_1$) appears to yield low bias estimates that nevertheless have poorer MSE results than the blanket LSW involving local time smoothing and then averaging over all replicates (despite its high bias). The small squared bias for the RLSW estimates implies that much of the MSE can be attributed to the variance. A possible explanation for this is the spectral leakage across neighbouring scales, a known artefact in the locally stationary spectral estimation context: the narrower choices of window to smooth over replicates in the RLSW model were not sufficient enough to remove the effects of the `leaked' characteristics. These simulations provide a stronger highlight of the impact of the ratio $M/R$ when choosing the replicate window width $(2M+1)$, with the guideline of a window width of 15\% of $R$ again appearing to yield competitive results. As $R$ and $T$ increase, and our choice of $M$ becomes larger with larger $R$ (increasing the smoothing window), the MSEs improve and RLSW$_2$ estimation  performs dramatically better than LSW. This is again in line with our asymptotic results. Also, we note here that while the spectra in Simulation 2 does not behave in quite as slowly evolving manner across replicates as in Definition \ref{def:rlsw}, nevertheless the RLSW methodology still performs well at estimating the spectra and capturing the evolutionary behaviour across replicates.

As with Simulation 1, we also provide histograms for the MSEs over 100 simulations involving $R=128$ replicates of length $T=256$. Figure~\ref{fig:hist_shiftproc} highlights how increasing the smoothing window improves the performance of our RLSW method and leads to better estimates. In this setting, Table~\ref{tab:msesim2} shows how the LSW model performs better in comparison to RLSW($_1$) when $M=4$ but as $M$ increased, the RLSW($_1$) model soon outperformed LSW.  This can also be visualised nicely through the histograms in Figure~\ref{fig:hist_shiftproc_B}, where we can see the switch in performance of the models and the MSE improvement with the increase of $M$.

\begin{figure}[htbp]
\begin{center}
\resizebox{!}{0.6\textwidth}{\includegraphics{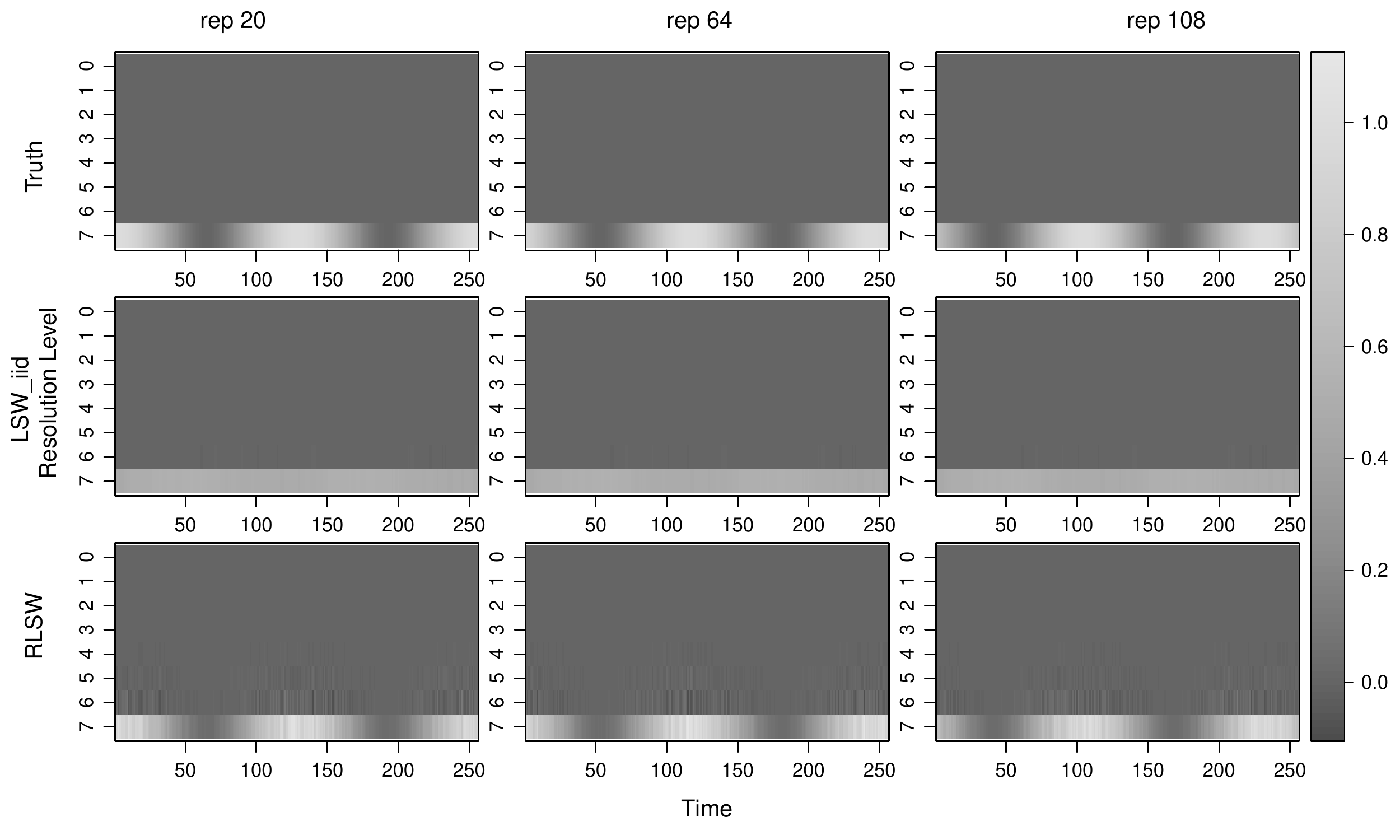}}
\end{center}
\vspace{-2em}
\caption{\small Simulation 2 time-scale plots for replicates 20, 60 and 108, respectively the first, second and third columns. Estimates are averaged  over 100 realisations. \textit{Top}: true spectra; \textit{Middle}: estimates from the LSW method averaged over all replicates; \textit{Bottom}: estimates using RLSW($_1$).}
\label{fig:coef_shiftproc_R128_T256}
\end{figure}

\begin{figure}[htbp]
\begin{center}
\resizebox{!}{0.6\textwidth}{\includegraphics{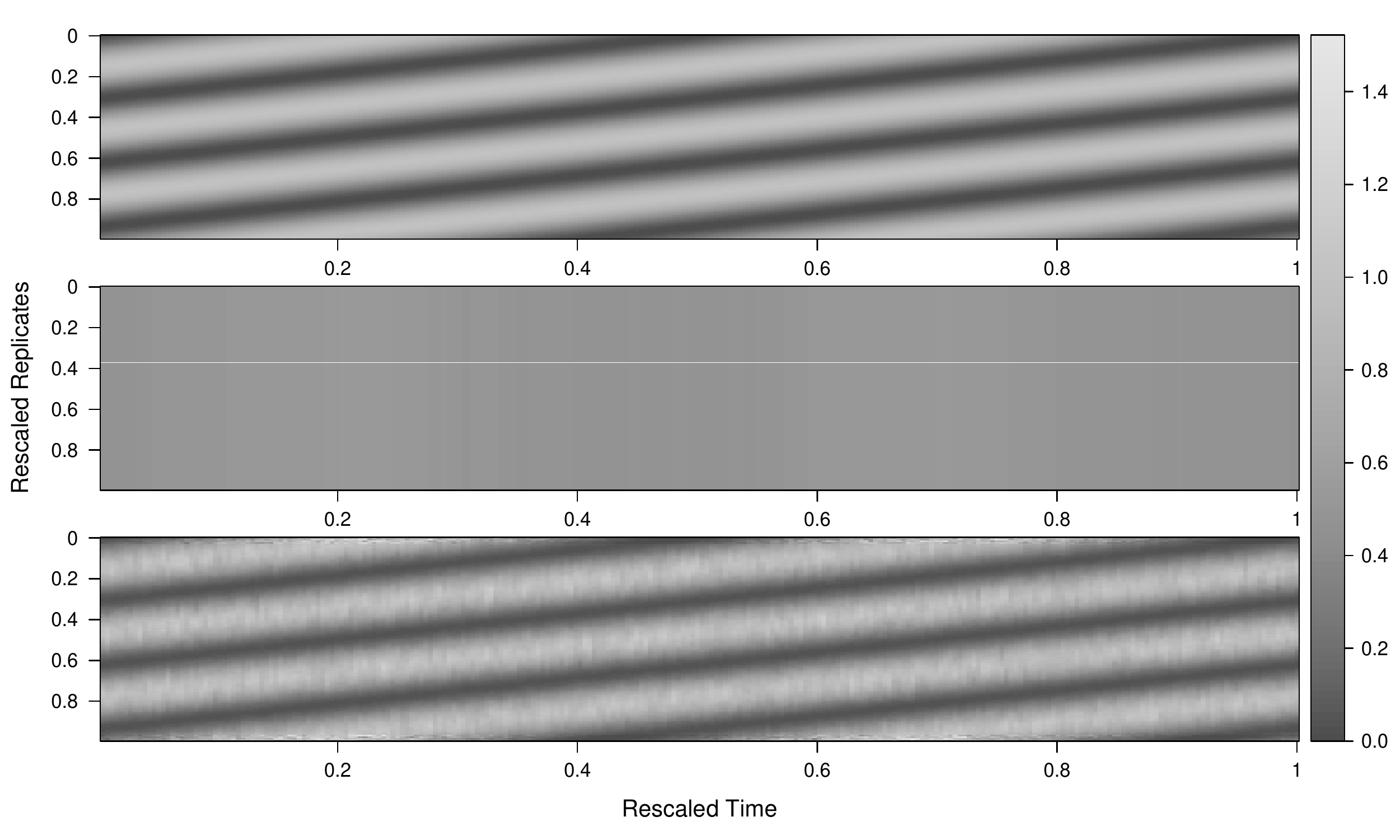}}
\end{center}
\vspace{-2em}
\caption{\small Simulation 2 time-replicate spectral plots in level 7. Estimates are averaged over 100 realisations. \textit{Top}: true spectra; \textit{Middle}:  estimated spectra from the LSW method averaged over all replicates; \textit{Bottom}: estimates using RLSW($_1$). }
\label{fig:lev7_shiftproc_R128_T256}
\end{figure}

\begin{table}[htbp!]
\scriptsize
\hspace{-6em}
\begin{tabular}{@{\extracolsep{3pt}}ccccccccc@{}}
\multicolumn{1}{c}{} & \multicolumn{7}{c}{Mean squared errors $\left( \times 1000 \right)$} & \multicolumn{1}{c}{}
\\
\cline{1-9}\noalign{\smallskip}
\multicolumn{3}{c}{} & \multicolumn{2}{c}{LSW} & \multicolumn{2}{c}{RLSW\tsbs{1}} & \multicolumn{2}{c}{RLSW\tsbs{2}}\\
\cline{4-5}\cline{6-7}\cline{8-9}\noalign{\smallskip}
R & T & M & mse & bias\tsps{2} & mse & bias\tsps{2} & mse & bias\tsps{2} \\
\firsthline\noalign{\smallskip}
64 & 128 & 4 & 19.37 & 18.09 & 33.91 & 1.97 & 10.69 & 1.85
\smallskip \\
 & & 7 & 19.49 & 18.21 & 28.04 & 8.85 & 14.10 & 8.79
\smallskip \\
 & & 10 & 19.39 & 18.11 & 33.34 & 19.57 & 23.24 & 19.44
\smallskip \\
 & & 12 & 19.20 & 17.92 & 36.49 & 24.87 & 27.88 & 24.69
\smallskip \\
 & 256 & 4 & 16.87 & 15.81 & 29.63 & 1.66 & 9.10 & 1.47
\smallskip \\
 & & 7 & 16.97 & 15.91 & 24.48 & 7.71 & 12.17 & 7.60
\smallskip \\
 & & 10 & 16.89 & 15.83 & 29.11 & 17.13 & 20.29 & 17.03
\smallskip \\
 & & 12 & 16.74 & 15.68 & 31.86 & 21.79 & 24.43 & 21.69
\smallskip \\
128 & 128 & 4 & 18.55 & 17.94 &  32.12 & 0.51 & 9.10 & 0.37
\smallskip \\
 & & 7 & 18.65 & 18.04 & 20.06 & 1.09 & 6.32 & 1.06
\smallskip \\
 & & 10 & 18.74 & 18.13 & 16.51 & 2.97 & 6.75 & 2.99
\smallskip \\
 & & 12 & 18.78 & 18.17 & 16.50 & 5.12 & 8.31 & 5.15
\smallskip \\
 & 256 & 4 & 16.23 & 15.69 & 28.27 & 0.39 & 7.85 & 0.21
\smallskip \\
 &  & 7 & 16.32 & 15.77 & 17.61 & 0.90 & 5.39 & 0.80
\smallskip \\
 & & 10 & 16.39 & 15.84 & 14.48 & 2.54 & 5.77 & 2.48
\smallskip \\
 & & 12 & 16.43 & 15.88 & 14.46 & 4.42 & 7.14 & 4.38
\smallskip \\
\lasthline
\end{tabular}
\quad
\begin{tabular}{@{\extracolsep{3pt}}ccccccccc@{}}
\multicolumn{1}{c}{} & \multicolumn{7}{c}{Mean squared errors $\left( \times 1000 \right)$} & \multicolumn{1}{c}{}
\\
\cline{1-9}\noalign{\smallskip}
\multicolumn{3}{c}{} & \multicolumn{2}{c}{LSW} & \multicolumn{2}{c}{RLSW\tsbs{1}} & \multicolumn{2}{c}{RLSW\tsbs{2}}\\
\cline{4-5}\cline{6-7}\cline{8-9}\noalign{\smallskip}
R & T & M & mse & bias\tsps{2} & mse & bias\tsps{2} & mse & bias\tsps{2} \\
\firsthline\noalign{\smallskip}
256 & 512 & 4 & 14.13 & 13.89 & 25.00 & 0.25 & 6.85 & 0.08
\smallskip \\
 & & 7 & 14.16 & 13.93 & 15.04 & 0.19 & 4.15 & 0.09
\smallskip \\
 & & 10 & 14.20 & 13.97 & 10.87 & 0.27 & 3.10 & 0.20
\smallskip \\
 & & 12 & 14.23 & 13.99 & 9.30 & 0.40 & 2.78 & 0.35
\smallskip \\
 & 1024 & 4 & 12.72 & 12.50 & 22.58 & 0.23 & 6.15 & 0.07
\smallskip \\
 & & 7 & 12.75 & 12.54 & 13.58 & 0.17 & 3.73 & 0.08
\smallskip \\
 & & 10 & 12.78 & 12.57 & 9.82 & 0.24 & 2.79 & 0.17
\smallskip \\
 & & 12 & 12.81 & 12.59 & 8.41 & 0.37 & 2.50 & 0.31
\smallskip \\
512 & 512 & 4 & 13.99 & 13.87 & 25.06 & 0.26 & 6.84 & 0.08
\smallskip \\
 & & 7 & 14.01 & 13.89 & 15.03 & 0.16 & 4.10 & 0.05
\smallskip \\
 & & 10 & 14.02 & 13.90 & 10.74 & 0.12 & 2.94 & 0.05
\smallskip \\
 & & 12 & 14.04 & 13.91 & 9.03 & 0.12 & 2.49 & 0.05
\smallskip \\
 & 1024 & 4 & 12.59 & 12.48 & 22.58 & 0.23 & 6.14 & 0.06
\smallskip \\
 &  & 7 & 12.60 & 12.50 & 13.55 & 0.14 & 3.69 & 0.04
\smallskip \\
 & & 10 & 12.62 & 12.51 & 9.69 & 0.11 & 2.64 & 0.04
\smallskip \\
 & & 12 & 12.63 & 12.52 & 8.15 & 0.10 & 2.23 & 0.04
\smallskip \\
\lasthline
\end{tabular}
\caption{\small MSE and squared bias results, averaged over all time-scale points and replicates for Simulation 2. `LSW' denotes the classical approach of averaging over the replicates. `RLSW$_1$' denotes our proposed approach using localised replicate smoothing. `RLSW$_2$' denotes our proposed approach using localised time and replicate smoothing. Our proposed methods `RLSW$_1$'and `RLSW$_2$' use a replicate smoothing window of length $(2M+1)$, while the time-smoothing for `LSW' and `RLSW$_2$' is automatically chosen.}
\label{tab:msesim2}
\end{table}

\begin{figure}[htbp]
\begin{center}
\resizebox{!}{0.4\textwidth}{\includegraphics{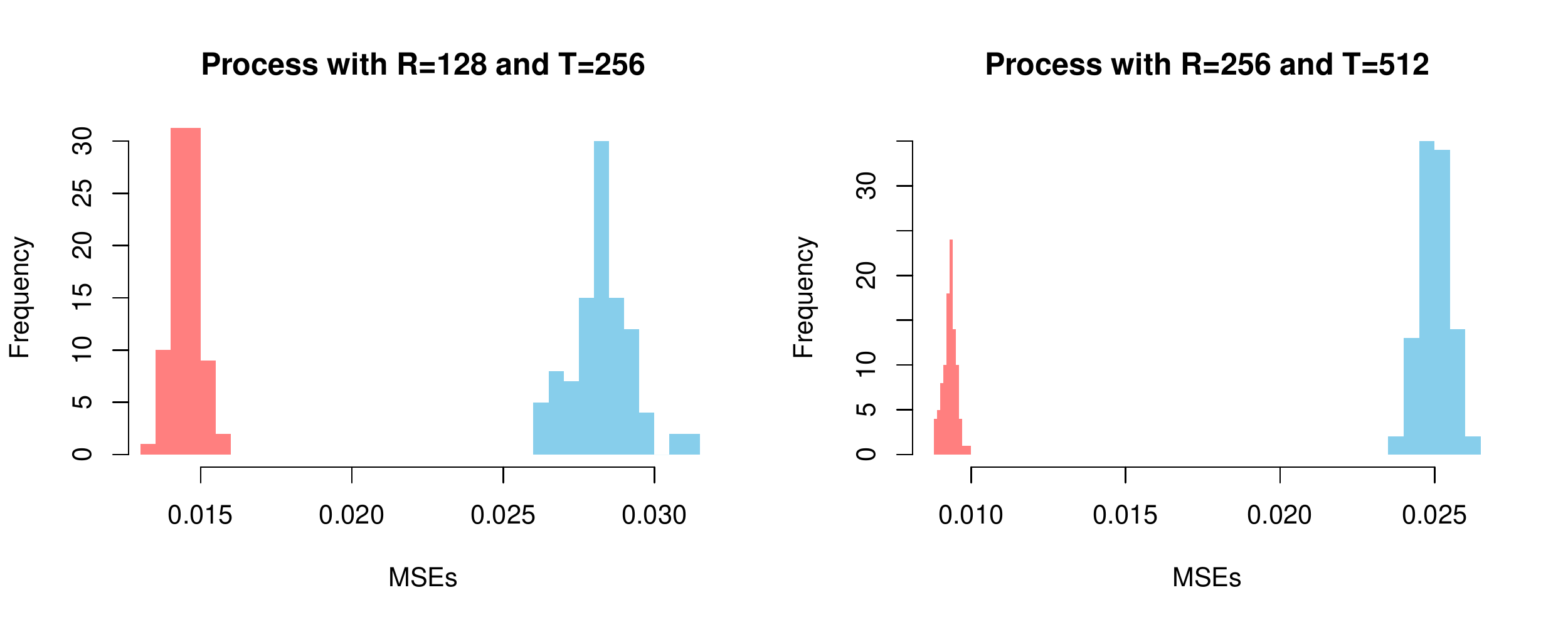}}
\end{center}
\vspace{-2em}
\caption{\small Histograms of the MSEs on the estimates from the RLSW model over 100 runs for Simulation 2. Smoothing over replicates with $M=12$ (\textit{red}); with $M=4$ (\textit{blue}). }
\label{fig:hist_shiftproc}
\end{figure}

\begin{figure}[htbp]
\begin{center}
\resizebox{!}{0.4\textwidth}{\includegraphics{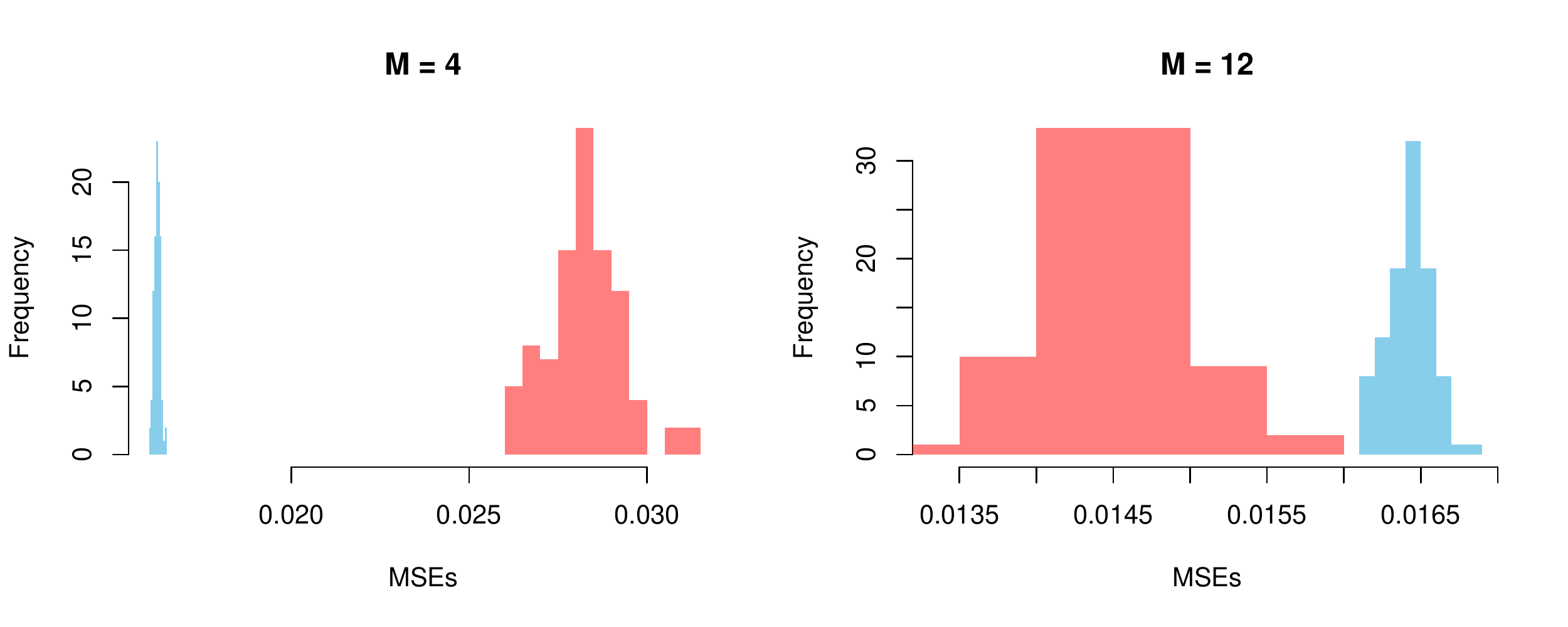}}
\end{center}
\vspace{-2em}
\caption{\small \small Histograms of the MSEs obtained over 100 runs for Simulation 2 with R = 128 and T = 256. Smoothing over replicates with $M=4$ (\textit{left}); with $M=12$ (\textit{right}). \textit{Red:} MSEs of RLSW($_1$) estimates; \textit{Blue:} MSEs of LSW estimates.}
\label{fig:hist_shiftproc_B}
\end{figure}

\clearpage
\section{Further simulation evidence for Section~\ref{sec:simsext}}\label{supp:furthersimscoh}

For level $j=5$ and time $k = 1,\ldots, 256$ we define the non-zero replicate coherence matrices of the illustrative example in Section~\ref{sec:simsext}, as follows
\begin{align*}
\left(\rho_{j}(\frac{k}{T},\frac{r}{R},\frac{r'}{R})\right)_{r,r'} &=
\begin{bmatrix}
	\rho_{5,k}^{1,1} & \cdots & \rho_{5,k}^{1,128} & \rho_{5,k}^{1,129} & \cdots & \rho_{5,k}^{1,R} \\
	\vdots & \ddots & \vdots & \vdots & \ddots & \vdots \\
	\rho_{5,k}^{128,1} & \cdots & \rho_{5,k}^{128,128} & \vdots & \cdots & \rho_{5,k}^{128,R} \\
	\rho_{5,k}^{129,1} & \cdots & \cdots & \rho_{5,k}^{129,129} & \cdots & \rho_{5,k}^{129,R} \\
	\vdots & \ddots & \vdots & \vdots & \ddots & \vdots \\
	\rho_{5,k}^{R,1} & \cdots & \rho_{5,k}^{R,128} & \rho_{5,k}^{R,129} & \cdots & \rho_{5,k}^{R,R}
\end{bmatrix} \\
&=
\begin{bmatrix}
	1 & 0.99 & \cdots & 0.99 & -0.71 & \cdots & \cdots & -0.71 \\
	0.99 & \ddots & \ddots & \vdots & \vdots & \ddots & \ddots & \vdots \\
	\vdots & \ddots & \ddots & 0.99 & \vdots & \ddots & \ddots & \vdots \\
	0.99 & \cdots & 0. 99 & 1 & -0.71 & \cdots & \cdots & -0.71 \\
	-0.71 & \cdots & \cdots & -0.71 & 1 & 0.5 & \cdots & 0.5\\
	\vdots & \ddots & \ddots & \vdots & 0.5 & \ddots & \ddots & \vdots \\
	\vdots & \ddots & \ddots & \vdots & \vdots & \ddots & \ddots & 0.5 \\	
	-0.71 & \cdots & \cdots & -0.71 & 0.5 & \cdots & 0.5 & 1
\end{bmatrix}
\end{align*}
where we have (abusively) let $\rho_{j}(\frac{k}{T}, \frac{r}{R},\frac{r'}{R}) = \rho_{j,k}^{r,r'}$ to ease notation.

\subsection*{Further simulation study}
We simulate a replicate locally stationary wavelet process with $R=256$ replicates that feature dependence, measured at $T=512=2^{9}$ time points. The within-replicate locally stationary wavelet (auto)spectra are as defined in Simulation 1 of Section~\ref{sec:sims}. Here we have $J(T)=9$ (in short, $J$) and the spectral characteristics are placed in level $j=J(T)-4=5$. In addition to the autospectral characteristics, we also define the cross-replicate spectral structure by means of defining their (true) coherence at each level $j$ and location $k$. For level $j = 5$, we choose a coherence of $0.7$ between all replicates over the first $256$ locations and zero (no) coherence over the last $256$ locations. All other levels have no coherence between replicates.  Visual representations of this dependence structure appear in Figure~\ref{fig:coh_sim1} (left panels). For $j=5$ and $k=1, \ldots,256$, the non-zero coherence matrices are defined as follows
\begin{equation}\nonumber
\left(\rho_{j}(\frac{k}{T},\frac{r}{R},\frac{r'}{R})\right)_{r,r'} =
\begin{bmatrix}
	\rho_{5,k}^{1,1} & \rho_{5,k}^{1,2} & \cdots & \rho_{5,k}^{1,R} \\
	\rho_{5,k}^{2,1} & \ddots & \ddots & \vdots \\
	\vdots & \ddots & \ddots & \rho_{5,k}^{R-1,R} \\
	\rho_{5,k}^{R,1} & \cdots & \rho_{5,k}^{R,R-1} & \rho_{5,k}^{R,R}
\end{bmatrix}
=
\begin{bmatrix}
	1 & 0.7 & \cdots & 0.7 \\
	0.7 & \ddots & \ddots & \vdots \\
	\vdots & \ddots & \ddots & 0.7 \\
	0.7 & \cdots & 0.7 & 1
\end{bmatrix}.
\end{equation}

Using the spectral estimation methodology proposed in Section~\ref{sec:modelext_est}, we obtain coherence estimates visually represented in Figure~\ref{fig:coh_sim1} (right panels) for replicate 50 (top row) and replicate 200 (bottom row). The non-decimated wavelet transform was computed using discrete wavelets built by means of Daubechies least asymmetric family with 10 vanishing moments and a value $M=4$ was chosen for the local replicate-smoothing.

\begin{figure}[htbp!]
\begin{center}
\resizebox{!}{0.6\textwidth}{\includegraphics{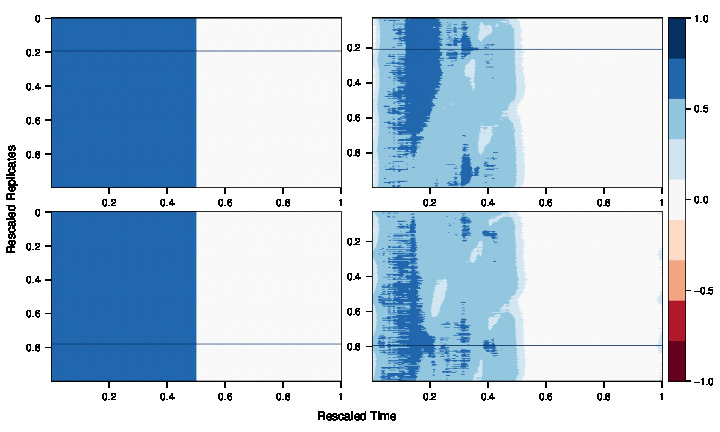}}
\end{center}
\vspace{-2em}
\caption{\small Further simulation: coherence plots for replicates 50 (\textit{top row}) and 200 (\textit{bottom row}) over rescaled time and replicates in level $5$. \textit{Left}: true coherence; \textit{Right}: coherence estimates averaged over 100 simulations.}
\label{fig:coh_sim1}
\end{figure}

It is apparent that the coherence structure is being picked up by the modelling framework in Section~\ref{sec:modelext} in terms of the locations and the positiveness of the defined true coherence. We do however note that the intensity of the estimated coherence is not quite as strong as the true coherence, just as for the simulation in Section~\ref{sec:simsext}. The numerical MSE and squared bias results reported in Table~\ref{tab:rhomsesim1} highlight that our correction procedure that aims to ensure positive spectral estimates has the undesired effect of increasing the MSEs and introducing bias as $R$ and $T$ increase, just as reported in a bivariate coherence estimation framework by \cite{sand:2010}.

\begin{table}[htbp!]
\small
\hspace{1em}
\centering
\begin{tabular}{cccccccc}
\cline{1-8}\noalign{\smallskip}
R & T & M & & mse\tsbs{1} & bias\tsps{2}\tsbs{1} & mse\tsbs{2} & bias\tsps{2}\tsbs{2} \\
\firsthline\noalign{\smallskip}
128 & 256 & 7 & & 17.29 & 7.41 & 15.16 & 9.66
\smallskip \\
 & & 12 & & 15.06 & 7.42 & 14.04 & 9.66
\smallskip \\
128 & 512 & 7 & & 17.64 & 8.11 & 16.01 & 10.66
\smallskip \\
 &  & 12 & & 15.29 & 8.01 & 14.73 & 10.52
\smallskip \\
256 & 512 & 7 & & 18.06 & 8.27 & 16.34 & 10.85
\smallskip \\
 & & 12 & & 15.85 & 8.18 & 15.18 & 10.75
\smallskip \\
\lasthline
\end{tabular}
\caption{\small Further simulation: MSE and squared bias $\left( \times 100 \right)$, averaged over all time-scale points and replicates. Subscripts 1 and 2 denote the models with smoothing over replicates only and time-replicate smoothing, respectively.}\label{tab:rhomsesim1}
\end{table}

\clearpage
\section{Further proofs for the material of Appendix~\ref{app:proofs}}\label{supp:proofs:sec3}

\begin{lemma}
Under the assumptions of Definition~\ref{def:rlsw}, we have for every rescaled time and replicate, $z, \, \nu$ respectively, and scale $j$,
 $\sum_{l=1}^{\infty}A_{j,l} S_{l}\left(z,\nu\right)=\ORD(2^{j})$.
\end{lemma}

\begin{proof}
\begin{align*}
\left|\sum_l A_{j,l} S_{l}\left(z,\nu\right)\right| &= \left|\sum_{l}\sum _{\uptau}\Psi_{j}(\uptau)\Psi_{l}(\uptau) S_{l}\left(z,\nu\right)\right| \mbox{ from the definition of the matrix }A \\
&= \left|\sum _{\uptau}\left(\sum_{l} S_{l}\left(z,\nu\right)\Psi_{l}(\uptau)\right)\Psi_{j}(\uptau)\right|\\
&=  \left|\sum _{\uptau}c\left(z,\nu;\uptau \right)\Psi_{j}(\uptau)\right|   \mbox { from the local autocovariance definition}\\
& \leq \sum _{\uptau}\left|c\left(z,\nu;\uptau \right)\right| \left|\Psi_{j}(\uptau)\right| \mbox { using the triangle inequality}\\
&= \ORD(2^{j}),
\end{align*}
where we used $\left|c\left(z,\nu;\uptau \right)\right| <\infty$ for all $\nu, \uptau$ and $\sum _{\uptau}\left|\Psi_{j}(\uptau)\right| = \ORD(2^{j})$ \citep{nvsk:2000}.
\end{proof}

\begin{lemma}
Under the assumptions of Definition~\ref{def:rlsw}, we have for every rescaled time and replicate, $z, \, \nu$ respectively, lag $\uptau$ and scale $j$,
$\sum_{l=1}^{\infty}A_{l,j}^{\uptau} S_{l}\left(z,\nu\right)=\ORD(2^{j})$.
\end{lemma}

\begin{proof}
\begin{align*}
\left|\sum_{l}A_{l,j}^{\uptau} S_{l}\left(z,\nu\right)\right|&= \left|\sum_{l}\sum _{n}\Psi_{l}(n) \Psi_{j}(n+\uptau) S_{l}\left(z,\nu\right)\right| \mbox{ from the definition of the matrix }A^\uptau \\
&= \left|\sum _{n}\left(\sum_{l} S_{l}\left(z,\nu\right)\Psi_{l}(n)\right)\Psi_{j}(n+\uptau)\right|  \\
&=  \left|\sum _{n}c\left(z,\nu;n \right)\Psi_{j}(n+\uptau)\right| \mbox { from the local autocovariance definition}\\
& \leq \sum _{n}\left|c\left(z,\nu;n \right)\right| \left|\Psi_{j}(n+\uptau)\right| \mbox { using the triangle inequality}\\
&= \ORD(2^{j}),
\end{align*}
where we used $\left|c\left(z,\nu;n \right)\right| <\infty$ for all $\nu, n$ and $\sum _{n}\left|\Psi_{j}(n+\uptau)\right| = \sum _{n}\left|\Psi_{j}(n)\right|= \ORD(2^{j})$.
\end{proof}

\begin{lemma}
Under the assumptions of Definition~\ref{def:rlsw}, we have  for every rescaled time and replicate, $z, \, \nu$ respectively, lag $\uptau$ and scales $j, j'$,
$$\sum_{l=1}^{\infty} \left|\sum _{n\in \Z} \Psi_{j,j'}(n) \Psi_{l}(n+\uptau) S_{l}\left(z,\nu\right)\right|=\ORD(2^{(j+j')/2}).$$
\end{lemma}

\begin{proof}
Let us first show that for any $\uptau$ and scales $j,\, j',\, l$, we have $$\sum _{n\in\Z}\Psi_{j,j'}(n)\Psi_{l}(n+\uptau)= \sum _{u\in\Z}\Psi_{j,l}(u)\Psi_{j',l}(u+\uptau).$$

This can be seen by re-expressing the right hand of the equality above by means of the definition of the cross-correlation wavelets, hence obtaining
\begin{align*}
\sum _{u}\Psi_{j,l}(u)\Psi_{j',l}(u+\uptau)
&= \sum_{u} \left( \sum_{p} \psi_{j,p}\psi_{l,p-u}\right)\left( \sum_{p'} \psi_{j',p'}\psi_{l,p'-u-\uptau}\right),\\
&= \sum_{p}\sum_{p'} \psi_{j,p} \psi_{j',p'} \left( \sum_u \psi_{l,p-u} \psi_{l,p'-u-\uptau}\right),\\
&= \sum_{p}\sum_{p'} \psi_{j,p} \psi_{j',p'} \Psi_l(p-p'+\uptau),\\
&= \sum_{n} \left(\sum_{p'} \psi_{j,n+p'} \psi_{j',p'} \right)\Psi_l(n+\uptau), \mbox{ where }n:=p-p',\\
&= \sum_{n} \Psi_{j,j'}(n)\Psi_l(n+\uptau).
\end{align*}

Using the above equality, we use the fact that the spectrum is positive and now take the triangle inequality
\begin{align*}
\sum_{l} S_{l}\left(z,\nu\right) \left| \sum _{n}\Psi_{j,l}(n) \Psi_{j',l}(n+\uptau)\right|  &
\leq \sum_{l} S_{l}\left(z,\nu\right) \sum _{n} \left| \Psi_{j,l}(n) \Psi_{j',l}(n+\uptau)\right|, \\
&\leq \sum_{l} S_{l}\left(z,\nu\right) \left( \sum _{n} \Psi_{j,l}^2(n)\right)^{1/2} \left(\sum _{n}  \Psi_{j',l}^2(n+\uptau)\right)^{1/2}, \\
&= \sum_{l} S_{l}\left(z,\nu\right) A_{j,l}^{1/2}A_{j',l}^{1/2}, \mbox{ from the definition of matrix }A\\
&\leq \left( \sum_{l} A_{j,l}S_{l}\left(z,\nu\right) \right) ^{1/2}
\left(\sum_{l'}  A_{j',l'}S_{l'}\left(z,\nu\right) \right) ^{1/2},\\
&= \ORD(2^{(j+j')/2}),
\end{align*}
where for the last two inequalities above we used the Cauchy-Schwarz inequality and for the last equality we used the result in Lemma~\ref{lem:ASord}.
\end{proof}

\section{Proof of Proposition~\ref{prop:autocov} }\label{supp:proofs:autocov}

\begin{proof}
As the RLSW process is defined to have zero-mean, we have $\cov\left(X_{\lfloor zT\rfloor;T}^{\lfloor \nu R\rfloor;T}, X_{\lfloor zT  \rfloor+\uptau;T}^{\lfloor \nu R\rfloor;T} \right)= \E\left(X_{\lfloor zT\rfloor;T}^{\lfloor \nu R\rfloor;T} X_{\lfloor zT  \rfloor +\uptau ;T}^{\lfloor \nu R\rfloor;T} \right)$ and using the RLSW process definition in equation~\eqref{eq:rlsw}, we obtain
\beqann
\cov\left(X_{\lfloor zT\rfloor;T}^{\lfloor \nu R\rfloor;R}, X_{\lfloor zT  \rfloor +\uptau ;T}^{\lfloor \nu R\rfloor;R} \right)&=&
\sum_{j=1}^\infty\sum_{k\in\Z} \left( \omega_{j,k;T}^{\lfloor \nu R\rfloor;R} \right)^2
\psi_{j,k}(\lfloor zT\rfloor)\psi_{j,k}(\lfloor zT\rfloor+\uptau), \mbox{ take }k:=n+\lfloor zT\rfloor,\\
&=&\sum_{j=1}^\infty\sum_{n\in\Z} \left(\omega_{j,n+\lfloor zT\rfloor;T}^{\lfloor \nu R\rfloor;R}\right)^2
\psi_{j,n+\lfloor zT\rfloor}(\lfloor zT\rfloor)\psi_{j,n+\lfloor zT\rfloor}(\lfloor zT\rfloor+\uptau), \\
&=& \sum_{j=1}^\infty\sum_{n\in\Z} \left(\omega_{j,n+\lfloor zT\rfloor;T}^{\lfloor \nu R\rfloor;R}\right)^2 \psi_{j,n}(0)\psi_{j,n}(\uptau).
\eeqann
Denoting $A=\left| \cov\left(X_{\lfloor zT\rfloor;T}^{\lfloor \nu R\rfloor;R}, X_{\lfloor zT  \rfloor +\uptau ;T}^{\lfloor \nu R\rfloor;R} \right)-c(z,\nu;\uptau) \right|$ and using the local autocovariance definition as well as the amplitude approximations in equation~\eqref{eq:approx}, we obtain
\beq\nonumber
A =\left| \sum_{j=1}^\infty\sum_{n\in\Z} \left(\omega_{j,n+\lfloor zT\rfloor;T}^{\lfloor \nu R\rfloor;R}\right)^2 \psi_{j,n}(0)\psi_{j,n}(\uptau)
- c(z,\nu; \uptau)\right|,
\eeq
a quantity bounded by\\
$\left| \sum_{j}\sum_{n} \left(S_{j}(z+\frac{n}{T},\nu)+ (D_j R^{-1})+(C_j T^{-1})\right)
\psi_{j,n}(0)\psi_{j,n}(\uptau)- \sum_{j}S_j(z,\nu)\Psi_j(\uptau)\right|$. Using the spectrum Lipschitz continuity in time, we further bound $A$ by\\
$\left| \sum_{j}\sum_{n} \left(S_{j}(z,\nu)+ (D_j R^{-1})+(C_j T^{-1})+ (L_{j}|n|T^{-1})\right)
\psi_{j,n}(0)\psi_{j,n}(\uptau)- \sum_{j}S_j(z,\nu)\Psi_j(\uptau)\right|$.

Hence
\beqann
A &\leq& \left| \sum_{j}\sum_{n} \left((D_j R^{-1})+(C_j T^{-1})+ (L_{j}|n|T^{-1})\right) \psi_{j,n}(0)\psi_{j,n}(\uptau)\right| \mbox{ as } \Psi_{j}(\uptau)=\sum_{k}\psi_{j,k}(0)\psi_{j,k}(\uptau),\\
&\leq& \left| \sum_{j}\sum_{n} (D_j R^{-1}) \psi_{j,n}(0)\psi_{j,n}(\uptau)\right| \\
&+& \left| \sum_{j}\sum_{n} (C_j T^{-1} \psi_{j,n}(0)\psi_{j,n}(\uptau)\right|\\
&+& \left| \sum_{j}\sum_{n} (L_{j}|n|T^{-1})\psi_{j,n}(0)\psi_{j,n}(\uptau)\right|,\\
&=& \ORD(R^{-1})+ \ORD(T^{-1}),
\eeqann
where the last equality follows as the terms  $\sum_{j}\sum_{n} (D_j R^{-1}) \psi_{j,n}(0)\psi_{j,n}(\uptau)=R^{-1} \sum_{j} D_j \Psi_{j}(\uptau)= \ORD(R^{-1})$ since $\Psi_{j}(\uptau)=\ORD(1)$ \citep{nvsk:2000} and we assumed $\sum_{j} D_j <\infty$, and similarly $\sum_{j}\sum_{n} (C_j T^{-1}) \psi_{j,n}(0)\psi_{j,n}(\uptau)=\ORD(T^{-1})$ using $\sum_{j} C_j <\infty$; the last term \\ $\left|\sum_{j}\sum_{n} (L_{j}|n|T^{-1})\psi_{j,n}(0)\psi_{j,n}(\uptau)\right|\leq \ORD(T^{-1}) \sum_{j} 2^j L_{j}$, since there are at most order $2^j$ non-zero terms in the wavelet product, $\Psi_{j}(\uptau)=\ORD(1)$ and we assumed $\sum_{j} 2^j L_{j}<\infty$ (see also \cite{nvsk:2000}).
\end{proof}

\section{Further proofs for the material of Appendix~\ref{app:proofsc}}\label{supp:proofs:sec5}

\begin{lemma}
Under the assumptions of Definition~\ref{def:rlswc}, we have a sequence $\{B_j\}$ of uniformly bounded Lipschitz constants in $j$ with $\sum_j2^jB_j<\infty$ such that
\begin{align}
& \left|\Wt_{j}\left(\frac{k+n}{T},\frac{r}{R}\right)\Wt_{j}\left(\frac{k+n}{T},\frac{r'}{R}\right)\rho_{j}\left(\frac{k+n}{T},\frac{r}{R},\frac{r'}{R}\right) \right. \nonumber \\
& \qquad \qquad \left. - \Wt_{j}\left(\frac{k}{T},\frac{r}{R}\right)\Wt_{j}\left(\frac{k}{T},\frac{r'}{R}\right)\rho_{j}\left(\frac{k}{T},\frac{r}{R},\frac{r'}{R}\right)\right| \leq |n| B_{j}T^{-1}, \label{eq:crosslipz}
\end{align}
for any replicates $r$, $r'$ and times $k$, $n$.
\end{lemma}

Note that the above result means that the replicate cross-spectrum, \i.e. $S_{j}(\cdotp,\nu,\nu')$, is Lipschitz continuous in the rescaled time argument for any rescaled replicates $\nu$, $\nu'$.

\begin{proof}
Under the assumptions of Definition~\ref{def:rlswc}, the functions $\Wt_{j}(\cdotp,\nu)$, $\Wt_{j}(\cdotp,\nu')$ and $\rho_{j}(\cdotp, \nu, \nu')$ are Lipschitz continuous in rescaled time, with Lipschitz constants $L_{j}^{\nu}$, $L_{j}^{\nu'}$ and $Q_{j}^{\nu,\nu'}$ respectively. Equivalently, this can be written as $\left|\Wt_{j}((k+n)/T,\nu)- \Wt_{j}(k/T,\nu)\right| \leq |n| L_{j}^{\nu}T^{-1}$, \\ $\left|\Wt_{j}((k+n)/T,\nu')- \Wt_{j}(k/T,\nu')\right| \leq |n| L_{j}^{\nu'}T^{-1}$ and $\left|\rho_{j}((k+n)/T,\nu,\nu')- \rho_{j}(k/T,\nu,\nu')\right| \leq |n| Q_{j}^{\nu,\nu'}T^{-1}$.

From standard Lipschitz function theory, the product of the Lipschitz continuous functions defined on a compact interval is also Lipschitz continuous with Lipschitz constant the maximum of the individual constants. Hence working across replicates and denoting  $B_{j}=\sup_{\nu,\nu'}\mbox{max}\{L_{j}^{\nu},L_{j}^{\nu'},Q_{j}^{(\nu,\nu')}\}$, we readily obtain that $S_{j}(\cdotp,\nu,\nu')=\Wt_{j}(\cdotp,\nu)\Wt_{j}(\cdotp,\nu')\rho_{j}(\cdotp, \nu, \nu')$ is Lipschitz continuous (the result in equation~\eqref{eq:crosslipz}) with constants $\{B_j\}$. The properties of the $\{B_j\}$ follow immediately from the similar properties of the individual Lipschitz constants $L_j=\sup_{\nu}L_j^\nu$ and $Q_j=\sup_{\nu,\nu'}Q_j^{\nu,\nu'}$.
\end{proof}

\begin{lemma}
Under the assumptions of Definition~\ref{def:rlswc}, we have a sequence $\{B_j^\prime\}$ of uniformly bounded Lipschitz constants in $j$ with $\sum_j2^jB^\prime_j<\infty$ such that
\begin{align}
& \left|\Wt_{j}\left(\frac{k}{T},\frac{r+s}{R}\right)\Wt_{j}\left(\frac{k}{T},\frac{r'+s}{R}\right)\rho_{j}\left(\frac{k}{T},\frac{r+s}{R},\frac{r'+s}{R}\right) \right. \nonumber\\
& \qquad \qquad \left. - \Wt_{j}\left(\frac{k}{T},\frac{r}{R}\right)\Wt_{j}\left(\frac{k}{T},\frac{r'}{R}\right)\rho_{j}\left(\frac{k}{T},\frac{r}{R},\frac{r'}{R}\right)\right| \leq |s| B^\prime_{j}R^{-1}, \label{eq:crosslipr}
\end{align}
for any times $k$, $n$ and replicates $r$, $r'$.
\end{lemma}

Note that the above result effectively states that the replicate cross-spectrum, \i.e. $S_{j}(z,\nu+\cdotp,\nu'+ \cdotp)$, is Lipschitz continuous in the rescaled replicate arguments.

\begin{proof}
Under the assumptions of Definition~\ref{def:rlswc}, for any rescaled time $z$ and replicates respectively $\nu'$, $\nu$, the limiting coherence functions $\rho_{j}(z, \cdotp, \nu')$ and $\rho_{j}(z, \nu, \cdotp)$  are Lipschitz continuous with Lipschitz constants $\{P_j^z\}$ and are defined on a compact interval. Then it immediately follows that there exists a bounded constant $C$ such that $\left|\rho_{j}(z, \nu+a, \nu'+a) - \rho_{j}(z, \nu, \nu') \right| \leq C \, P_{j}^{z} |a|$ for any $a$ such that $\nu+a, \nu'+a \in (0,1)$.

Specifically, the Lipschitz continuity assumption in rescaled replicate-dimension implies that for each $z$ we have $\left|\Wt_{j}(z,(r+s)/R)- \Wt_{j}(z,r/R)\right| \leq |s| N_{j}^{z}R^{-1}$, $\left|\Wt_{j}(z,(r'+s)/R)- \Wt_{j}(z,r'/R)\right| \leq |s| N_{j}^{z}R^{-1}$ and\\ $\left|\rho_{j}(z,(r+s)/R,(r'+s)/R)- \rho_{j}(z,r/R,r'/R)\right| \leq |s| P_{j}^{z}R^{-1}$ at any replicates $r$, $r'$ and $s$.

Then, from standard Lipschitz function theory it follows that their product $S_{j}(z, \nu+\cdotp, \nu'+\cdotp)=\Wt_{j}(z,\nu+\cdotp)\Wt_{j}(z,\nu'+\cdotp)\rho_{j}(z, \nu+\cdotp, \nu'+\cdotp)$ is also Lipschitz continuous (the desired result in equation~\eqref{eq:crosslipr}) with constants $B^\prime_{j}=\sup_{z}\mbox{max}\{N_{j}^{z},P_{j}^{z}\}$. The properties of the $\{B^\prime_j\}$ follow immediately from the similar properties of the individual Lipschitz constants $N_j=\sup_{z}N_j^z$ and $P_j=\sup_{z}P_j^z$.
\end{proof}

\section{Proof of Proposition \ref{prop:autocovc}}\label{supp:proofs:autocovc}
\begin{proof}
The proof follows immediately from the amplitude and coherence approximation properties in Definition~\ref{def:rlswc} of the RLSW process, which yield
\beq \nonumber
\left|\omega_{j,\lfloor zT\rfloor;T}^{\lfloor \nu R\rfloor;R} \,
\omega_{j,\lfloor zT\rfloor;T}^{\lfloor \nu' R\rfloor;R} \,
\rho_{j,\lfloor zT\rfloor;T}^{\lfloor \nu R\rfloor, \lfloor \nu' R\rfloor;R}-
S_{j}(z,\nu,\nu')\right|=\ORD(C_j^\prime T^{-1})+ \ORD(D_j ^\prime R^{-1}),
\eeq
where $C_j^\prime$ and $D_j^\prime$ can be shown to also fulfill equations of the type in Definition~\ref{def:rlswc},~\ref{enum:seq1} and~\ref{enum:seq2}.

As the limiting amplitudes ($\Wt_{j}(\cdotp,\nu)$, $\Wt_{j}(\cdotp,\nu')$) and coherence ($\rho_{j}(\cdotp, \nu, \nu')$) are Lipschitz continuous functions in rescaled time for rescaled replicates $\nu$, $\nu'$ and are defined on a compact interval, it follows from standard Lipschitz function theory that their product ($S_{j}(\cdotp,\nu,\nu')$) is also Lipschitz continuous in rescaled time, as shown in Lemma~\ref{lem:lipz}. The desired conclusion then follows using the same type of arguments as employed in Section~\ref{supp:proofs:autocov}.
\end{proof}

\section{Proof of estimated coherence limits}\label{supp:proofs:rhoclim}
\begin{proof}
Let us first show that $\left|\It_{l,k;T}^{(r,r');R}\right| \leq \sqrt{\left|\It_{l,k;T}^{(r,r);R}\right|} \sqrt{\left|\It_{l,k;T}^{(r',r');R}\right|}$ for any scale $l$, time $k$ and replicates $r,\, r'$. Equivalently, we want to show that $\left(\It_{l,k;T}^{(r,r');R}\right)^2 \leq \left(\It_{l,k;T}^{(r,r);R}\right) \left(\It_{l,k;T}^{(r',r');R}\right)$.

Recalling that we use the same smoothing window and using the definitions of the smoothed periodogram in~\eqref{eq:reppersmooth} and of the raw periodogram in~\eqref{eq:rawIc}, the above is equivalent to having to show that
\beq\nonumber
\left( \sum_s d_{l,k;T}^{r+s;R}d_{l,k;T}^{r'+s;R} \right)^2\leq
\left( \sum_s \left(d_{l,k;T}^{r+s;R}\right)^2 \right)
\left( \sum_{s'} \left(d_{l,k;T}^{r'+s';R}\right)^2 \right),
\eeq
which indeed follows from the Cauchy-Schwarz inequality.

Recalling that the components in the estimated replicate--coherence equation~\eqref{eq:rhoc} are of the form $\hat{S}_{j,k;T}^{(r,r');R} = \sum_l A^{-1}_{j,l} \It_{l,k;T}^{(r,r');R}$ (as in equation~\eqref{eq:corrspecc}), next let us take
\beqann
\left| \sum_l A^{-1}_{j,l} \It_{l,k;T}^{(r,r');R} \right|^2
&\leq& \left( \sum_l \left|A^{-1}_{j,l}\right| \left|\It_{l,k;T}^{(r,r');R}\right| \right)^2, \mbox{ and using the relationship above }\\
&\leq& \left( \sum_l \left(\sqrt{\left|A^{-1}_{j,l}\right|}\right)^2 \sqrt{\left|\It_{l,k;T}^{(r,r);R}\right|}\sqrt{ \left|\It_{l,k;T}^{(r',r');R}\right|}
\right)^2, \\
&=& \left( \sum_l \left(\sqrt{\left|A^{-1}_{j,l}\right|}\sqrt{\left|\It_{l,k;T}^{(r,r);R}\right|}\right) \left(\sqrt{\left|A^{-1}_{j,l}\right|}\sqrt{\left|\It_{l,k;T}^{(r',r');R}\right|}\right)
\right)^2,\\
&\leq& \left( \sum_l A^{-1}_{j,l} \It_{l,k;T}^{(r,r);R} \right)
\left( \sum_{l'} A^{-1}_{j,l'} \It_{l',k;T}^{(r',r');R} \right),
\eeqann
where the last inequality resulted from the Cauchy-Schwarz inequality. Hence the values of the squared estimated replicate-coherence have the property $\left| \hat{\rho}_{j,k;T}^{(r,r');R} \right|\leq 1$.
\end{proof}

\end{appendix}

\end{spacing}
\end{document}